\documentclass[11pt]{article}

\usepackage{amsmath}
\usepackage{graphicx}
\usepackage{url} 
\usepackage{authblk}
\usepackage{algorithm}
\usepackage{amsthm}
\usepackage{amsfonts}
\usepackage{amssymb}
\usepackage{algpseudocode}
\usepackage{bbm}
\usepackage{bm}
\usepackage{booktabs}
\usepackage{caption}
\usepackage{color}
\usepackage{enumerate}
\usepackage{enumitem}
\usepackage{float}
\usepackage{mathtools} 
\usepackage{xr-hyper}
\usepackage{hyperref}
\usepackage{cleveref}
\usepackage{mathtools}
\usepackage{longtable}

\def\L{B}
\def\P{\mathbb{P}}
\def\fulldata{ (\by^{(k)})_{k \in [K]}}
\def\data{\mathcal{D}}
\def\post{\Pi}
\def\prior{\Pi}

\def\resid{\tilde{\by}} 
\def\A{\bm{B}} 
\def\J{\mathcal{A}}
\def\SSVS{muSSVS}

\def\vb{\tau}
\def\ak{a}
\def\L{B}
\def\P{\mathbb{P}}
\def\fulldata{ (\by^{(k)})_{k \in [K]}}
\def\data{\mathcal{D}}
\def\post{\Pi}
\def\prior{\Pi}

\def\resid{\tilde{\by}} 
\def\A{\bm{B}} 
\def\J{\mathcal{A}}
\def\SSVS{muSSVS}

\renewcommand{\mid}{\,|\,}
\def\dset{\chi}
\def\pvb{\phi}
\def\sel{u}
\def\consta{\eta}
\def\constb{\tilde{\eta}}

\newtheorem{lemma}{Lemma}

\newtheorem{theorem}{Theorem}

\theoremstyle{definition}

\theoremstyle{remark}
\newtheorem*{remark}{Remark}
 
\newcommand{\ELBO}{\mathrm{ELBO}}
\newcommand{\muSER}{\mathrm{muSER}}
\newcommand{\T}{\mathrm{T}}
\newcommand{\BF}{\mathrm{BF}}
\newcommand{\KL}{\mathrm{KL}}

\newcommand{\ind}{\mathbbm{1}}

\newcommand{\bsLambda}{\boldsymbol{\Lambda}}

\newcommand{\bsSigma}{\boldsymbol{\Sigma}}

\newcommand{\bsPhi}{\boldsymbol{\Phi}}

\newcommand{\bsgamma}{\boldsymbol \gamma}

\newcommand{\bstheta}{\boldsymbol \theta}

\newcommand{\bspi}{\boldsymbol{\pi}}

\newcommand{\sphi}{\mathsf{\phi}}

\newcommand{\bsalpha}{\boldsymbol{\alpha}}

\newcommand{\bsbeta}{\boldsymbol{\beta}}

\newcommand{\bsmu}{\boldsymbol{\mu}}

\newcommand{\bA}{\bm{A}}

\newcommand{\bB}{\bm{B}}

\newcommand{\cD}{\mathcal{D}}

\newcommand{\cE}{\mathcal{E}}

\newcommand{\bbE}{\mathbb{E}}

\newcommand{\cG}{\mathcal{G}}

\newcommand{\bI}{\bm{I}}

\newcommand{\bM}{\bm{M}}

\newcommand{\bbM}{\mathbb{M}}

\newcommand{\cN}{\mathcal{N}}

\newcommand{\bbN}{\mathbb{N}}

\newcommand{\cQ}{\mathcal{Q}}

\newcommand{\bR}{\bm{R}}

\newcommand{\bbR}{\mathbb{R}}

\newcommand{\bbS}{\mathbb{S}}

\newcommand{\bU}{\bm{U}}

\newcommand{\bV}{\bm{V}}

\newcommand{\bX}{\bm{X}}

\newcommand{\sX}{\mathsf{X}}

\newcommand{\bb}{\bm{b}}

\newcommand{\be}{\bm{e}}

\newcommand{\br}{\bm{r}}

\newcommand{\bs}{\bm{s}}

\newcommand{\by}{\bm{y}}

\renewcommand{\sphi}{\zeta}
\renewcommand{\bsPhi}{\bm{\Psi}} 
\renewcommand{\cD}{\bX} 
\renewcommand{\br}{S}
\renewcommand{\bs}{T} 

\usepackage[round]{natbib} 

\def\tcr{\textcolor{black}}
 
\title{Bayesian Multi-task Variable Selection with an Application to Differential DAG Analysis}
\author{Guanxun Li}
\author{Quan Zhou}
\affil{Department of Statistics, Texas A\&M University}

\begin{document}

\maketitle

\begin{abstract}
We study the Bayesian multi-task variable selection problem,  where the goal is to select activated variables for multiple related data sets simultaneously. 
We propose a new variational Bayes algorithm which generalizes and improves the recently developed ``sum of single effects" model of~\citet{wang2020simple}. 
Motivated by differential gene network analysis in biology, we further extend our method to joint structure learning of multiple directed acyclic graphical models,  a problem known to be computationally highly challenging. 
We propose a novel order MCMC sampler where our multi-task variable selection algorithm is used to quickly evaluate the posterior probability of each ordering. 
Both simulation studies and real gene expression data analysis are conducted to show the efficiency of our method. 
Finally, we also prove a posterior consistency result for multi-task variable selection, which provides a theoretical guarantee for the proposed algorithms.
Supplementary materials for this article are available online.   
\end{abstract}

\section{Introduction}
In machine learning, multi-task learning refers to the paradigm where we  simultaneously learn multiple related tasks instead of learning each task independently~\citep{zhang2021survey}.  
In the context of model selection, we can formulate the problem as follows: given $K$ observed data sets where the $k$-th data set is generated from some statistical model $\mathfrak{M}^{(k)}$, simultaneously estimate $\mathfrak{M}^{(1)}, \dots, \mathfrak{M}^{(K)}$ so that the estimation of $\mathfrak{M}^{(k)}$ (for any $k=1, \dots, K$) utilizes information from all $K$ data sets. In real problems where the $K$ models tend to share many common features, this joint estimation approach is expected to have better performance than separate estimation (i.e, estimating $\mathfrak{M}^{(k)}$ using only the $k$-th data set). In this work, we consider \tcr{multi-task model selection problems where each task may be variable selection or structure learning.}
 
We first study the multi-task variable selection problem, where each data set is generated from a sparse linear regression model. 
The majority of the existing research  has been conducted under the strict assumption that the ``activated'' covariates (i.e., covariates with nonzero regression coefficients) are shared across all data sets~\citep{lounici2009taking, lounici2011oracle}. Recent works have relaxed this assumption by taking a more adaptable strategy that splits each regression coefficient into a shared and an individual component~\citep{jalali2010dirty, hernandez2015probabilistic}. 
We propose a more flexible Bayesian method which generalizes the well-known spike-and-slab prior~\citep{george1993variable, ishwaran2005spike}  and allows a covariate to be activated in an arbitrary number of data sets with varying effect sizes. We prove the posterior consistency for our model in high-dimensional scenarios. 
While there is a large literature on frequentists' approaches to multi-task learning, the corresponding Bayesian methodology has received less attention  and in particular theoretical results are lacking~\citep{bonilla2007multi, guo2011sparse, hernandez2015probabilistic}.  
To our knowledge, this is the first work that establishes the theoretical guarantee for the high-dimensional Bayesian multi-task variable selection problem. 

The traditional method for obtaining the posterior distribution for a Bayesian model is to use Markov Chain Monte Carlo (MCMC) sampling, which is often computationally intensive, especially for multi-task learning problems where the space of candidate models can be enormous. 
A more scalable alternative is variational Bayes (VB), which recasts posterior approximation as an optimization problem  \citep{ray2021variational}. To carry out efficient VB inference, we approximate our spike-and-slab prior model using a novel multi-task sum of single effects (muSuSiE) model, which extends the sum of single effects (SuSiE) model of \cite{wang2020simple} to multiple data sets. Then, we propose to fit muSuSiE using an iterative Bayesian stepwise selection (IBSS) method, which may be thought of as a coordinate ascent algorithm for maximizing the evidence lower bound over a particular variational family. 

To illustrate the application of the proposed methodology to more complex multi-task learning problems, we consider differential network analysis based on directed acyclic graphs (DAGs), which is essentially a multi-task structure learning problem. 
Differential network analysis has emerged as a significant topic in biology and received increasing attention over recent years. 
Its application can be found in the analysis of various diseases and biological mechanisms such as lung cancer~\citep{li2020bayesian}, breast cancer~\citep{liu2019joint}, Parkinson's disease~\citep{lee2022bayesian}, brain connectivity network~\citep{zhang2020causal} and the study of phosphorylated proteins and phospholipid components~\citep{castelletti2020bayesian}.  
Because learning a DAG model can be equivalently viewed as a set of variable selection problems when the order of nodes is known~\citep{agrawal2018minimal}, learning multiple DAG models with a known order is likewise equivalent to a set of multi-task variable selection problems. 
However, when the order is not known (which is usually the case in practice), learning the order of nodes from the data can be very challenging. 
To overcome this issue, we employ MCMC sampling over the permutation space to average over the uncertainty in learning the order of nodes and then compute the DAG model for each given order via the proposed Bayesian multi-task variable selection approach. Simulation studies and a real data example are used to demonstrate the effectiveness of the proposed method. 

The rest of this paper is organized as follows. In Section \ref{sec:model}, we introduce our model for Bayesian multi-task variable selection, prove the high-dimensional posterior consistency and describe the VB algorithm for model-fitting. 
Section \ref{sec:vs-simu} presents simulation results for the multi-task variable selection problem. 
In Section \ref{sec:Multi-DAGs}, we generalize our method to joint estimation of multiple DAG models and propose an order MCMC sampler. 
Simulation studies and real data analysis for differential DAG analysis are presented in Sections~\ref{sec:dags-simu} and~\ref{sec:real.data}, respectively. Section~\ref{sec:conclusion} concludes the paper with a brief discussion. 
Proofs, additional simulation results and more details about the algorithm implementation are deferred to the appendices in supplementary materials.  

\section{Bayesian Multi-task Variable Selection}\label{sec:model}
\subsection{Model, prior and posterior distributions} \label{sec:model.prior.posterior}
We  introduce some notation to be used throughout the paper. Denote the cardinality of a set $S$ by $|S|$. 
For any $k \in \bbN$, let $[k] = \{1, 2, \cdots, k\}$, and \tcr{let $2^{[k]} = \{ S \colon S \subseteq [k] \}$ denote the power set on it; note that $|2^{[k]}| = 2^k$.}  
For any vector $\bb$ and matrix $\bA$, let $\bb_{S}$ be the subvector of $\bb$ with index set $S$ and $\bA_{S}$ be the submatrix of $\bA$ containing columns indexed by $S$.   
Let $\ind$ denote the indicator function.  

For the multi-task variable selection problem, let $K$ denote the number of data sets we have, which is treated as fixed in this paper. 
We assume the same $p$ covariates are observed in all $K$ data sets. 
For the $k$-th data set, let $n_k$ denote the sample size, $\by^{(k)}\in\bbR^{n_{k}}$ the response vector, and $\bX^{(k)} \in \bbR^{ n_{k}\times p }$ the design matrix containing $n_{k}$ observations of the $p$ covariates. 
Consider the linear regression model 
\begin{equation}\label{eq:linear-regression}
    \by^{(k)} = \bX^{(k)}\bsbeta^{(k)} + \be^{(k)}, \quad\text{where } \be^{(k)} \sim \cN_{n_{k}}(0, \sigma^2 \bI_{n_{k}}),  \quad \forall \, k \in [K], 
\end{equation}
where $\cN_{n}$ denotes the $n$-dimensional normal distribution, $\bI_n$ denotes the $n$-dimensional identity matrix,  and the vector of regression coefficients, $\bsbeta^{(k)}$, is assumed to be sparse. 
For ease of presentation, we assume the error variance $\sigma^2$ is the same across all data sets, but this assumption can be relaxed straightforwardly in the theory and algorithms to be developed in this paper. For now, we also assume that  $\sigma^2$ is known, and we will explain in Appendix B, supplementary materials how to estimate it in practice.  

The main parameter of interest is the set-valued  vector $\bsgamma \in (2^{[K]})^p$, where $\gamma_j = I$ means that the $j$-th covariate has a nonzero regression coefficient (i.e., it is activated) in the $k$-th data set for each $k \in I$. For instance, $\gamma_1 = \{1, 2\}$ indicates that the first covariate is activated in both the first and second datasets; whereas $\gamma_2 = \emptyset$ indicates that the second covariate is deactivated across all datasets. Let $|\bsgamma| = \sum_{j = 1}^p \ind_{ \{ \gamma_j \neq \emptyset \} }$ denote the number of covariates that are activated in at least one data set, and let 
\begin{equation*}
  \ak_k(\bsgamma) = |\{ j \in [p] \colon  |\gamma_j| = k  \} |
\end{equation*}
be the number of covariates that are activated in $k$ distinct data sets. Note that $|\bsgamma| = \ak_1 + \cdots + \ak_K$.
The main idea behind our construction of the prior on $(\bsgamma, \{ \bsbeta^{(k)} \}_{k = 1}^K)$, denoted by $\prior(\bsgamma,  ( \bsbeta^{(k)} )_{k = 1}^K)$, is similar to the spike-and-slab prior for single-task variable selection.  
First, given $\bsgamma$, we assume that $\beta^{(k)}_j = 0$ if $k \notin \gamma_j$, and put a normal prior on $\beta^{(k)}_j$ otherwise. Next, to achieve sparsity, we put a prior on $\bsgamma$ that favors sparser models. 
Explicitly, our prior is given by 
\begin{align}
 \beta^{(k)}_j \mid \bsgamma \;& \overset{\mathrm{ind}}{\sim} \ind_{\{ k \notin \gamma_j \}}   \delta_0 +  \ind_{\{ k \in \gamma_j \}}  \cN_1(0,  \vb^{(k)}_j),  \quad \forall\, j \in [p], \, k \in [K], \label{eq:cond-prior-beta} \\ 
 \prior(\bsgamma)  \propto  \;& 
    \ind_{ \{ |\bsgamma| \leq L \} }  f(| \bsgamma |, L) \prod_{k = 1}^K p^{-\omega_k   \ak_k(\bsgamma ) },  \label{eq:prior-gamma} 
\end{align}
where  $L \in \bbN$, $ \vb^{(k)}_j > 0$ for $j \in [p], k \in [K]$ and $\omega_k > 0$ for $k \in [K]$ are hyperparameters, and $\delta_0$ denotes the Dirac measure at 0.
The function $f(|\bsgamma|, L)$ is introduced for generality, and in our theoretical analysis it will be assumed  to be ``asymptotically negligible'' compared to the product term in~\eqref{eq:prior-gamma}. 
Hence, the sparsity is mainly promoted by the hard threshold $L$, which is the maximum number of activated covariates (in at least one data set) we allow, and the hyperparameters $(\omega_k)_{k=1}^K$.  
We can view $\omega_k$ as the ``cost'' we pay for activating one covariate simultaneously in $k$ data sets. 

For most multi-task variable selection problems in reality, it is reasonable to assume that activated covariates tend to be shared across data sets, and to reflect this prior belief, we propose to choose $(\omega_k)_{k=1}^K$ such that 
\begin{align}\label{eq:kappa-condition} 
& \frac{\omega_K}{K} < \frac{\omega_{K - 1}}{K - 1} < \frac{\omega_{K - 2}}{K - 2} < \cdots <  \frac{\omega_2}{2} <  \omega_1.  
\end{align}
To see the reasoning behind~\eqref{eq:kappa-condition}, consider the case $K = 2$ where the above condition is reduced to $ \omega_2 < 2 \omega_1$. 
\tcr{Suppose that the first two covariates are identical in both data sets, and consider two models $\bsgamma, \bsgamma'$ such that $\gamma_1 = \{1\}, \gamma_2 = \{2\}$, $\gamma'_1 = \{1, 2\}$, $\gamma'_2 = \emptyset$ and $\gamma_j =\gamma'_j = \emptyset$ for any $j > 2$. 
Then, $\bsgamma, \bsgamma'$ have the same marginal likelihood, but $\ak_1(\bsgamma) = 2, \ak_2(\bsgamma) = 0$ and $\ak_1(\bsgamma') = 0, \ak_2(\bsgamma') = 1$.}  
It can be seen that $\omega_2 < 2 \omega_1$ ensures we favor $\bsgamma'$. 
An analogous argument for the general case with $K \geq 2$ leads to~\eqref{eq:kappa-condition}. 
Note that the choice of $\omega_1, \dots, \omega_K$ only reflects the experimenter's prior belief on $\bsgamma$, and one can even use $\omega_k \ll \omega_1$ for all $k \geq 2$ if prior information reveals that the majority of activated covariates must be shared in multiple data sets. 
However, in all of our numerical studies, we only use $(\omega_k)_{k=1}^K$  such that~\eqref{eq:kappa-condition} is satisfied and $\omega_1 \leq \omega_2 \leq \cdots \leq \omega_K$, the latter of which appears to be a natural condition in situations where not much prior information is available.  
We will refer to the model specified by Equations~\eqref{eq:linear-regression} to~\eqref{eq:prior-gamma} as \SSVS{} (multi-task Spike-and-Slab Variable Selection).  

\subsection{Posterior Consistency for  Multi-task Spike-and-slab Variable Selection}\label{sec:posterior.consistency}
In this section, we prove the posterior consistency for the \SSVS{} model, which generalizes the existing results for \tcr{single-task variable selection}~\citep{johnson2012bayesian, narisetty2014bayesian, yang2016computational, jeong2021unified}.  
We only consider in our proof the special case $n_k = n$ and $\vb^{(k)}_j = \vb$ for $k \in [K]$ and $j \in [p]$. 
Analogous arguments can be used to prove the posterior consistency in the more general case where  $( \vb^{(k)}_j)_{k \in [K], j \in [p]}$   are bounded and $n_1, \dots, n_K$ are different with $\min_{k \in [K]} n_k$ being sufficiently large. 

Suppose the data is  generated by~\eqref{eq:linear-regression}  with $\bsbeta^{(k)*}$ being the vector of true regression coefficients for the $k$-th data set. Our goal is to show that covariates with a relatively high signal strength (aggregated over multiple data sets) can be recovered with high probability.  To this end,  define the ``true'' model $\bsgamma^*$ as follows.   
Let $C_{\beta, 1}, \dots, C_{\beta, K}$ be constants that depend on $n$, $p$, $\sigma^2$ and $\vb$.  For each $j \in [p]$, define 
\begin{align*}
    m_j^* = \max \left\{ m \in [K] \colon  | \{ k  \in [K] \colon   (\beta^{(k)*}_j )^2 \geq C_{\beta, m} \} | = m \right \}, 
\end{align*}
and set $\gamma^*_j = \{ k  \in [K] \colon   (\beta^{(k)*}_j )^2 \geq C_{\beta, m_j^*}   \}$. 
If $k \in \gamma^*_j$,  we say the $j$-th covariate is ``influential'' in the $k$-th data set (a ``non-influential'' covariate may have a small but nonzero regression coefficient). 
In words, $C_{\beta, k}$ can be seen as the detection threshold for covariates that have relatively large nonzero regression coefficients in $k$ distinct data sets. 
For our posterior consistency result, we will assume that $C_{\beta, 1} > \cdots > C_{\beta, K}$, which reflects the advantage of multi-task learning: if a covariate is activated in more data sets, the signal size in each data set required for detection can be smaller.    

We assume the following five conditions hold for $k = 1, \cdots, K$,  which were also used in the consistency analysis for single-task variable selection conducted in~\cite{yang2016computational}. 
However, since we use an independent normal prior on the nonzero entries of $\bsbeta^{(k)}$ while~\cite{yang2016computational} considered the g-prior (which significantly simplifies the calculation), some of our conditions are slightly more stringent. 
\tcr{We use 
\begin{equation}\label{eq:def.S.m}
    S_k(\bsgamma) = \{ j \in [p] \colon k \in \gamma_j \} 
\end{equation}
to denote the set of covariates that are activated in the $k$-th data set, and we simply denote the set of truly influential covariates by $S^*_k = S_k(\bsgamma^*)$.} 
 
\begin{enumerate}[label=(\arabic*)]
\item The first condition is on the true regression coefficients $\bsbeta^{(k)*}$. \label{cond1}
\begin{enumerate}[label=(1\alph*)]
\item  For some $\L_1 \geq 1$,  $\frac{1}{n} \|\bX^{(k)}\bsbeta^{(k)*}  \|_2^2 \leq \L_1\sigma^2\log p$.
\label{cond1.1}
\item  For some $\L_2 \geq 0$, $\frac{1}{n} \|\bX^{(k)}_{ (S^*_k)^c}\bsbeta^{(k)*}_{(S^*_k)^c} \|_2^2 \leq \L_2\sigma^2\frac{\log p}{n}$.
\label{cond1.2}
\end{enumerate}
\end{enumerate} 
Condition~\ref{cond1.1} requires that the order of the total signal size in each data set, $\| \bX^{(k)}\bsbeta^{(k)*} \|_2^2$, is at most $n\log p$, and  Condition~\ref{cond1.2} requires that non-influential covariates cannot contribute significantly to the variation in $\by^{(k)}$. Both are reasonable assumptions for most high-dimensional problems. If one assumes all nonzero entries of $\bsbeta^{(k)*}$ are sufficiently large in absolute value, then $\bsbeta_{(S^*_k)^c}^* = 0$ and Condition~\ref{cond1.2} holds trivially. If one further assumes the influential covariates have bounded regression coefficients (i.e., coefficients do not grow with $n$), 
Condition~\ref{cond1.1} allows each data set to have $O(\log p)$ independent influential covariates,  which is not restrictive when $p \gg n$. More discussion on   Condition~\ref{cond1.1} will be given after Condition~\ref{cond5}.

\begin{enumerate}[label=(\arabic*)] \addtocounter{enumi}{1}
\item The second condition is on the design matrix. For any symmetric matrix $\bA$, denote its smallest eigenvalue by $\lambda_{\min}(\bA)$.
\begin{enumerate}[label=(2\alph*)]
\item  $\|X_j^{(k)}\|_2^2 = n$ for all $j = 1, \cdots, p$. \label{cond2.1}
\item  
For some $\nu \in (0, 1]$, $\min_{|S| \leq L} \lambda_{\min}\left(\frac{1}{n}(\bX_{S}^{(k)})^\T\bX_{S}^{(k)}\right) \geq \nu$.  \label{cond2.2}
\item Let $\bm{Z} \sim \cN_{n}(0, \bI)$. For some $\L_3 \ge 8 / \nu$, we have  \label{cond2.3}
\[\frac{1}{\sqrt{n}} \, \bbE_{\bm{Z}} \left[\max_{S \colon |S|\leq L}\max_{j \in S^c }   \big| \bm{Z}^\T (\bI - \bsPhi_{S}^{(k)}) \bX_j^{(k)} \big| \right]\leq \frac{1}{2}\sqrt{\L_3\nu\log p},\] 
where  $\bsPhi_{S}^{(k)} = \bX_{S}^{(k)}\bigl((\bX_{S}^{(k)})^\T\bX_{S}^{(k)}\bigr)^{-1}(\bX_{S}^{(k)})^\T$ is the projection matrix. 
\end{enumerate}
\end{enumerate} 
Condition~\ref{cond2.1} assumes all columns of $\bX^{(k)}$ are normalized and is used to simplify the calculation.  
Condition~\ref{cond2.2} is known as the lower restricted eigenvalue condition and is a modest constraint necessary for theoretical analysis of Bayesian variable selection problems~\citep{narisetty2014bayesian}. 
Condition~\ref{cond2.3} is called the sparse projection condition~\citep{yang2016computational}. 
\tcr{ 
Since Condition (2a) ensures that $ \|(\bI - \bsPhi_{S}^{(k)})\bX_j^{(k)}\|_2 \leq \sqrt{n}$ for all $k\in[K]$, $|S|\leq L$ and $j\in[p]$,  one can use a standard inequality for maximum of Gaussian random variables to show that Condition~\ref{cond2.3} always holds for some $B_3 = O(L \nu^{-1})$. But when the design matrix consists of independent covariates, $B_3$ can be much smaller; see~\citet{yang2016computational} for more details. 
}

\begin{enumerate}[label=(\arabic*)] \addtocounter{enumi}{2}
\item The third condition is on the choice of prior hyperparameters. Let $\widetilde{\vb} = \vb / \sigma^2$, and $C$ denote some universal constant (i.e., a constant that does not depend on $n$).  
\begin{enumerate}[label=(3\alph*)]
\item $ 1 + n\widetilde{\vb}  \leq C p^{2 \consta}$ for some 
$\consta > 0$. \label{cond3.1}
\item $L \leq C p^{\constb}$ for some $\constb \in (0, 1)$. \label{cond3.2}
\item  $( \omega_k )_{k = 1}^K$ satisfies~\eqref{eq:kappa-condition} and
$\frac{\omega_k}{k} > \frac{3}{2}\left(\frac{\L_1}{\nu \widetilde{\vb} } + \L_2 + \L_3\right) + \constb + 2$. \label{cond3.3}
\item The function $f$ in~\eqref{eq:prior-gamma} satisfies \label{cond3.4}
$1 \leq \frac{f(s + 1, L)}{f(s, L)} \leq L$ for every $s \in \bbN$. 
\end{enumerate}
\end{enumerate}   
\tcr{Condition~\ref{cond3.1} is only used to bound a determinant term in the posterior distribution of $\bsgamma$.} 
In high-dimensional settings  with $n \ll p$, both Conditions~\ref{cond3.1} and~\ref{cond3.2} are very natural and easy to satisfy. 
Condition~\ref{cond3.3} requires the parameter $\omega_k$ to be sufficiently large, which is needed to ensure that the posterior mass concentrates on sparse models. 
Condition~\ref{cond3.4} implies that $f(|\bsgamma |, L) \leq L^{|\bsgamma|}$. 
By Condition~\ref{cond3.3}, we have $\omega_k > 2 k \geq 2$, and thus the  product term in~\eqref{eq:prior-gamma} is at most $p^{-2|\bsgamma|}$. 
Since $L = o (p)$ by Condition~\ref{cond3.2}, we see that the magnitude of $\prior(\bsgamma)$ depends little on the function $f(|\bsgamma|, L)$. 

\begin{enumerate}[label=(\arabic*)]\addtocounter{enumi}{3}
\item The true sparsity level $\left| S^*_k \right|$ satisfies
$\max \left\{1,\left| S^*_k \right| \right\} \leq \frac{n}{25\log p}$. \label{cond4} 
\end{enumerate}

\begin{enumerate}[label=(\arabic*)]\addtocounter{enumi}{4}
\item The constant $C_{\beta, k}$ is given by \label{cond5}
$C_{\beta, k} = \left\{8\left(\frac{\omega_{k}}{k} + 2 + \consta \right) + \frac{12 \L_1}{\nu\widetilde{\vb} }\right\} \frac{ \sigma^2 \log p }{ n\nu }$.
\end{enumerate}
Condition~\ref{cond5} is known as the beta-min condition~\citep{yang2016computational}. 
By inequality~\eqref{eq:kappa-condition}, it further implies that $C_{\beta, K} < C_{\beta, K - 1} < \cdots < C_{\beta, 1}$; that is, the more data sets in which the covariate is influential, the lower the signal strength level required to detect it.  
To gain further insights into this condition, consider the case where $\consta, B_1, \nu, \widetilde{\vb}, \sigma^2$ are all universal constants. 
\tcr{Then, the order of $C_{\beta, k}$ is given by $ \frac{ \omega_k \log p }{k n}$,} which typically goes to zero 
in the high-dimensional asymptotic regimes considered in the literature, implying that  we can identify activated covariates with diminishing signal sizes.  
\tcr{Note that Conditions~\ref{cond1.1} and~\ref{cond5} are compatible with each other.  
For example, assuming $\omega_k / k$ is a constant, to satisfy Condition~\ref{cond5}, all entries of  $(\bsbeta^{(k)*})^2$ corresponding to influential covariates only need to have order $n^{-1}\log p$; in this case, we have $\| \bX^{(k)}\bsbeta^{(k)*} \|_2^2  = O( |S^*_k| \log p)$, which is much smaller than the order $n \log p$ required by Condition~\ref{cond1.1}.}

\begin{theorem}\label{posterior-concen} 
Suppose for each $k$, $\by^{(k)}$ is generated by~\eqref{eq:linear-regression} with $\bsbeta^{(k)} = \bsbeta^{(k)*}$. 
If Conditions~\ref{cond1} to~\ref{cond5} hold, we have
\begin{align*}
    \P \left( \left\{ \post(\bsgamma^* \mid  \fulldata ) \geq 1 - c_1 p^{-1} \right\} \right) \geq 1 - c_2 p^{-c_3}, 
\end{align*} 
where $\post( \cdot \mid   \fulldata )$ denotes the posterior measure for the model specified by Equations~\eqref{eq:linear-regression} to~\eqref{eq:prior-gamma},  $\P$ denotes the probability measure for the true data-generating process, and $c_1$, $c_2$ and $c_3$ are positive universal constants. 
\end{theorem}
\begin{proof}
We defer the proof to Appendix A, supplementary materials.
\end{proof}

\begin{remark} 
\tcr{The main difference between Theorem~\ref{posterior-concen} and existing consistency results for single-task spike-and-slab variable selection~\citep{narisetty2014bayesian, yang2016computational} is that the detection threshold $C_{\beta, k}$ in our Condition~\ref{cond5} depends on $k$. 
When~\eqref{eq:kappa-condition} holds, $C_{\beta, k}$ is smaller for larger $k$, which means that by combining information from multiple data sets and properly choosing $(\omega_k)_{k=1}^K$ (see Condition~\ref{cond3.3}), we can detect activated covariates with smaller signal sizes. This rigorously justifies the advantage of multi-task variable selection over separate analysis. 
}
\end{remark}

\subsection{Multi-task Sum of Single Effects Model}
For Bayesian problems, posterior distributions are typically calculated through Markov Chain Monte Carlo (MCMC) sampling. 
But in our case, the huge discrete model space can make the sampling converge very slowly.  
In this section, we approximate our \SSVS{} model by a multi-task sum of single effects (muSuSiE) model, generalizing the recently developed sum of single effects (SuSiE) model of~\citet{wang2020simple} for \tcr{single-task  variable selection}.  
The muSuSiE model assumes that for each $k \in [K]$, 
\begin{equation}\label{eq:muSuSiE.reg}
    \by^{(k)} \sim  \cN_{n_k}(\bX^{(k)}\bsbeta^{(k)}, \sigma^2 \bI_{n_k}),  \text{ where } \bsbeta^{(k)} = \sum_{l = 1}^L \bsbeta^{(k, l)},  
\end{equation}
and each $\bsbeta^{(k, l)} \in \bbR^p$ has at most one nonzero entry; that is, we decompose each $\bsbeta^{(k)}$ into a ``sum of single effects.'' We will call each $\bsbeta^{(k, l)}$ a  single-effect  regression coefficient vector. 
Similarly, we introduce $L$ set-valued  single-effect selection vectors $\bsgamma^{(1)}, \dots, \bsgamma^{(L)}$ such that  $\gamma_j^{(l)} = I$ means that $\beta^{(k, l)}_j$ is nonzero for each $k \in I$ (i.e., covariate $j$ is the $l$-th single effect and is activated in the data sets indexed by $I$). 
Let $\dset$ denote a probability distribution on $2^{[K]} \setminus \emptyset$ and $\mathrm{Unif}([p])$ denote the uniform distribution on $[p]$. 
The prior distribution we put on $\{ \bsgamma^{(l)}\colon  l \in [L] \}$  encodes the following procedure for selecting and activating covariates:  
for each $l \in [L]$, we  draw  $\sphi_l \sim \mathrm{Bernoulli}(\pi_{\sphi})$, $\sel_l \sim \mathrm{Unif}([p])$ and $I_l \sim \dset$; if $\sphi_l = 1$, we activate the $\sel_l$-th covariate in the data sets indexed by $I_l$,  and we do nothing if $\sphi_l = 0$. 
So $\sphi_l$ indicates whether the $l$-th single effect is indeed activated.  
For each activated covariate in each data set, we still use a normal prior distribution on its effect size as in~\eqref{eq:cond-prior-beta}. Note that we assume $\sel_1, \dots, \sel_L$ are generated independently and thus a covariate can be activated multiple times, which  is the key difference between muSuSiE  and \SSVS{}.

Formally,  the prior distribution of muSuSiE  can be expressed as follows: 
\begin{equation}\label{eq:muSuSiE}
\begin{aligned}
 \sel_l &\overset{\mathrm{ind}}{\sim} \mathrm{Uniform}([p]), & \forall \, l \in [L], \\ 
 \sphi_l &\overset{\mathrm{ind}}{\sim}\mathrm{Bernoulli}(\pi_{\sphi}), & \forall \, l \in [L], \\ 
 \gamma^{(l)}_j \mid  ( \sel_l, \sphi_l )_{l \in [L]} &\overset{\mathrm{ind}}{\sim} (1 - \sphi_l \ind_{\{ \sel_l = j \}})   \delta_\emptyset +  \sphi_l \ind_{\{ \sel_l = j \}} \dset ,  & \forall \, j \in [p], l \in [L], \\ 
 \beta^{(k, l)}_j \mid  ( \bsgamma^{(l)} )_{l \in [L]}  
&\overset{\mathrm{ind}}{\sim} \ind_{\{ k \notin \gamma_j^{(l)} \}}   \delta_0 +  \ind_{\{ k \in \gamma_j^{(l)} \}}  \cN_1(0,  \vb^{(k, l)}_j),  & \forall\, j \in [p], \, k \in [K], \,  l \in [L],   
\end{aligned}
\end{equation}
where $( \vb^{(k, l)}_j )_{j, k, l}$ are hyperparameters and $\delta_{\emptyset}$ denotes the Dirac measure that assigns unit probability mass to the empty set. 
Though in~\eqref{eq:muSuSiE.reg} we write $\bsbeta^{(k)}$ as the sum of $L$ terms, the actual sparsity is controlled by the hyperparameter $\pi_{\sphi}$. 
Each $\bsgamma^{(l)}$ has zero (if $\sphi_l = 0$) or one (if $\sphi_l = 1$) covariate activated. 

We now discuss how to choose the probability distribution $\dset$. 
We introduce hyperparameters $\pi_1 > \pi_2 > \cdots > \pi_K > 0$ and  set
\begin{align*}
    \dset(  I ) = p \, \pi_{|I|}, \quad \forall \, I \in 2^{[K]} \setminus \emptyset. 
\end{align*}
Assume  $\pi_1, \dots, \pi_K$ are normalized so that  $\dset(2^{[K]} \setminus \emptyset) = 1$. 
Let $s_\sphi = | \{ l \colon \sphi_l = 1\}|$ denote the number of activated single effects,  
$\{ \bsgamma^{(l)} \colon \sphi_l = 1 \}$ denote the unordered set of activated single-effect selection vectors, and  $I_l$ denote the value of $\gamma^{(l)}_{\sel_l}$. 
Note that $\{ \bsgamma^{(l)} \colon \sphi_l = 1 \}$ is completely determined by $\bigl( (\sel_l, I_l)\bigr)_{l \in [L]}$, since we always have $\gamma^{(l)}_j = \emptyset$ for any $j \neq \sel_l$. 
Let $\tilde{\prior}$ denote the probability measure under the muSuSiE model given by~\eqref{eq:muSuSiE}.  
If no covariate is activated more than once (i.e., for any $l \neq l'$ such that $\sphi_l = \sphi_{l'} = 1$, we have $\sel_l \neq \sel_{l'}$), 
\begin{equation}\label{eq:prior.susie}
    \tilde{\prior}( \{ \bsgamma^{(l)} \colon \sphi_l = 1 \} ) = 
   f( s_\sphi, L) 
    \prod_{l = 1}^L \pi_{\sphi}^{\sphi_l} (1 - \pi_{\sphi})^{1 - \sphi_l}\pi_{ |I_l|  }^{\sphi_l}, 
\end{equation}
where $f(s, L) = L\times(L - 1)\times \cdots \times (L - s + 1)$  satisfies Condition \ref{cond3.4}. A straightforward calculation  shows that~\eqref{eq:prior.susie} and~\eqref{eq:prior-gamma} are equivalent if 
\begin{equation}\label{eq:pi_to_omega}
    \frac{\pi_{\sphi}\pi_k}{1 - \pi_{\sphi}} = p^{-\omega_k},
\end{equation}
for each $k \in [K]$. This shows why muSuSiE is an approximation to the \SSVS{} model. 
Again, the two models are not equivalent because we may have $\sel_l = \sel_{l'}$ for some $l \neq l'$ in~\eqref{eq:muSuSiE}, though this happens with very small probability when $p$ is large.  
While the repeated activation of a covariate may seem artificial and slightly unnatural, this feature enables us to propose an efficient VB method (to be introduced in the next subsection) which can quickly yield an approximate Bayesian solution to the multi-task  variable  selection problem.

\begin{remark}\label{rmk:bernoulli}
While  muSuSiE is based on the SuSiE model proposed by~\citet{wang2020simple} for single-task variable selection, 
our model~\eqref{eq:muSuSiE} with $K=1$ still differs from SuSiE in that 
we use Bernoulli random variables $\sphi_1, \cdots, \sphi_L$ to control the actual sparsity of $(\bsbeta^{(k)})_{k \in [K]}$. 
The prior distribution used in model~\eqref{eq:muSuSiE} assumes that the number of activated covariates (including duplicates) follows $\mathrm{Binomial}(L, \pi_\sphi)$, and given a sufficiently large sample size, the model~\eqref{eq:muSuSiE} is able to learn the actual number of activated covariates, which can range from $0$ to $L$. 
This also implies that an increase in the value of $L$ is not likely to have a significant impact on the posterior distribution. 
In contrast, SuSiE  assumes there are exactly $L$ activated single effects and relies on an ad-hoc procedure to determine which covariates are truly activated from the output of a VB algorithm. 
\end{remark}

\subsection{Iterative Bayesian Stepwise Selection for Fitting muSuSiE}\label{sec:ibss}
We propose an iterative Bayesian stepwise selection (IBSS) method for fitting the model given in \eqref{eq:muSuSiE} by generalizing the IBSS algorithm of~\cite{wang2020simple}. 
The main idea is to iteratively find $\bsgamma^{(l)}$ for $l = 1, \dots, L$ in the muSuSiE model by conditioning on the other $L-1$ single effects. 
The starting point for our algorithm is the muSuSiE model with $L=1$, which we will refer to as the ``multi-task single-effect regression'' (muSER) model and we recall below with superscript $l$ dropped: 
\begin{equation}\label{eq:muSER} 
\begin{aligned}
\by^{(k)} &\sim \cN_{n_k}(\bX^{(k)} \bsbeta^{(k)}, \sigma^2\bI_{n_k}), & \forall \, k \in [K], \\
 \sel &\overset{\mathrm{ind}}{\sim} \mathrm{Uniform}([p]), \\ 
 \sphi &\overset{\mathrm{ind}}{\sim}\mathrm{Bernoulli}(\pi_{\sphi}), \\ 
 \gamma_j \mid  ( \sel, \sphi) &\overset{\mathrm{ind}}{\sim} (1 - \sphi\ind_{\{ \sel = j \}})   \delta_\emptyset +  \sphi \ind_{\{ \sel = j \}} \dset, & \forall \, j \in [p], \\ 
 \beta_j^{(k)} \mid \bsgamma
&\overset{\mathrm{ind}}{\sim} \ind_{\{ k \notin \gamma_j \}}   \delta_0 +  \ind_{\{ k \in \gamma_j \}}  \cN_1(0,  \vb), \quad\quad  & \forall\, j \in [p], k\in[K]. 
\end{aligned}
\end{equation}
Since we only allow at most one covariate to be activated in~\eqref{eq:muSER}, the joint posterior distribution of $(\bsgamma, 1 - \sphi)$ given $\sigma^2$ and $\vb$ can be quickly calculated, which is given by a multinoimal distribution with 
\begin{equation*}
\prior_{\muSER}( \sphi = 0 \mid (\by^{(k)})_{k \in [K]} ) = \alpha_0, \quad  \prior_{\muSER}(\gamma_j = I\mid (\by^{(k)})_{k \in [K]} ) = \alpha_{j, I},   
\end{equation*}
where expressions for $\alpha_{j, I}$ and $\alpha_0$ are given in Appendix B, supplementary materials. 
By definition, $\alpha_0 + \sum_{j \in [p]} \sum_{I \neq \emptyset}\alpha_{j, I} = 1$. 
Further, the posterior distribution of $\beta_j^{(k)}$ given $\sphi = 1, u = j, k \in \gamma_j$ (i.e., the $j$-th covariate is activated in the $k$-th data set) is 
$$\beta_j^{(k)}|\fulldata, \sigma^2, \vb, \sphi = 1, u= j,  k \in \gamma_j \sim \cN(\mu_{j}^{(k)}, \pvb^{(k)}_j),$$
where we defer the explicit expressions for $\mu_{j}^{(k)}$ and $\pvb^{(k)}_j$ to Appendix B, supplementary materials.  
(Note that whenever $\sphi = 0$, $u \neq j$ or $k \notin I$, the posterior distribution of $\beta_j^{(k)}$ is $\delta_0$.)
For ease of notation, we introduce a function, $f_{\muSER}$, which returns the posterior distribution for $\bsbeta$ under the muSER model. 
Since this posterior distribution is determined by the values of  $\alpha_0$, $\bsalpha = ( \alpha_{j, I} )_{j \in [p], I\in2^{[K]}\setminus \emptyset}$,  $\bsmu^{(k)} = (\mu_{1}^{(k)}, \cdots, \mu_{p}^{(k)})$ and $\bm{\pvb}^{(k)} = (\pvb^{(k)}_1, \cdots, \pvb^{(k)}_p)$ for $k = 1, \cdots, K$, we define  $f_{\muSER}$ by 
\begin{equation}\label{eq:muSER-fun}
f_{\muSER}(\fulldata; \sigma^2, \vb) \coloneqq \left(\bsalpha, \alpha_0, ( \bsmu^{(k)} )_{k \in [K]}, ( \bm{\pvb}^{(k)} )_{k \in [K]} \right).  
\end{equation}
Observe that for the muSuSiE model, if $\{ \bsbeta^{(k, l')} \colon l' \in [L] \text{ and } l' \neq l\}$ is given, calculating the posterior distribution of $\bsbeta^{(k, l)}$ is very straightforward: one just needs to fit the muSER model by substituting the residual $\by^{(k)} - \bX^{(k)} \sum_{l' \neq l}\bsbeta^{(k, l')}$ for the response  $\by^{(k)}$ for each $k$ in the muSER model~\eqref{eq:muSER}. This suggests an iterative strategy for fitting muSuSiE, which we detail in Algorithm~\ref{alg:IBSS}. 
The implementation of our algorithm is analogous to the IBSS algorithm for the original SuSiE model.

{\small 
\begin{algorithm}[h!]
\caption{Iterative Bayesian stepwise selection (IBSS) for fitting muSuSiE}\label{alg:IBSS}
\begin{algorithmic}
\Require data $ ( \bX^{(k)} )_{k = 1}^K$, $( \by^{(k)} )_{k = 1}^K$, number of single effects $L$
\Require a function $f_{\muSER}$ which is defined in \eqref{eq:muSER-fun}
\State initialize posterior means $\widehat \bsbeta^{(k, l)} = 0$ for $l = 1, \cdots, L$ and $k = 1, \cdots, K$
\State initialize $\widehat{\sigma}^2$ and $( \tau^{(l)} )_{l = 1}^L$
\If{the stopping criterion is not satisfied}
    \For{$l = 1, \cdots, L$}
    \For{$k = 1, \cdots, K$}
    \State $\resid^{(k, l)} \gets \by^{(k)} - \bX^{(k)}\sum_{l' \neq l}\widehat \bsbeta^{(k, l')}$
    \EndFor
    \State estimate $\tau^{(l)}$ by maximizing Equation~(B.2) in Appendix B.1  
    \State $ \left(\bsalpha^{(l)}, \alpha_{0}^{(l)},  (\bsmu^{(k, l)} )_{k = 1}^K, ( \bm{\pvb}^{(k, l)} )_{k = 1}^K\right) \gets f_{\muSER}(  (\resid^{(k, l)})_{k \in [K]}; \widehat{\sigma}^2, \tau^{(l)})$.
    \For{$k = 1, \cdots, K$} 
    \For{$j = 1, \cdots, p$} 
    \State $\widehat \beta^{(k, l)}_j \gets   \mu^{(k, l)}_j   \sum_{\{I \colon k\in I\}} \alpha^{(l)}_{j, I}$   
    \EndFor
    \EndFor
\EndFor
\State update $\widehat{\sigma}^2$ by Equation~(B.8) in Appendix B.2  
\EndIf
\State  \Return $\widehat{\sigma}^2$, $(\bsalpha^{(l)})_{l = 1}^L$,  $( \widehat \bsbeta^{(k, l)} )_{l \in [L], k \in [K]}$
\end{algorithmic}
\end{algorithm}
}

 
Let $\widehat \bsbeta^{(k, l)}$ be as given in the output of Algorithm \ref{alg:IBSS}, which denotes the estimated $l$-th single-effect regression coefficient vector for the $k$-th data set. 
We can express the posterior mean regression coefficient vector for the $k$-th data set by 
\begin{equation}\label{eq:post-reg-coef}
\widehat \bsbeta^{(k)} = \sum_{l = 1}^L \widehat \bsbeta^{(k, l)}.  
\end{equation}  
Further, taking all $L$ single-effect selection vectors into account, we can approximate the probability that the $j$-th covariate is activated in the $k$-th data set by 
\begin{equation}\label{eq:sel-prob}
 r^{(k)}_j = 1 - \prod_{l = 1}^L \left(1 - r^{(k, l)}_{j}\right), \quad \text{ where }     r^{(k, l)}_j = \sum_{\{I \colon k\in I\}} \alpha^{(l)}_{j, I}
\end{equation}
is the probability that the $j$-th coordinate is activated in the $k$-th data set in the $l$-th single-effect model, conditioning on the other $L-1$ single effects.

By an argument similar to that in~\cite{wang2020simple}, we can show that this IBSS algorithm coincides with the coordinate ascent variational inference (CAVI) algorithm \citep{blei2017variational} for maximizing the evidence lower bound over a particular variational family for the muSuSiE model; see Appendix B, supplementary materials, where we also explain how to choose the stopping criterion and estimate $\sigma^2$ and $\tau^{(l)}$ empirically in Algorithm~\ref{alg:IBSS}. 

\begin{remark}\label{rmk:ibss.advantage}
We can also implement the VB algorithm for the model proposed in Section~\ref{sec:model} by generalizing VB methods for single-task variable selection~\citep{carbonetto2012scalable, huang2016variational, ormerod2017variational, ray2022variational}. 
However, a key advantage of the IBSS algorithm for SuSiE/muSuSiE is that, in addition to being fast, it does not use a variational family that assumes independence among $\gamma_1, \dots, \gamma_p$ (in single-task variable selection, $\gamma_j$ indicates whether the $j$-th covariate is activated), which is particularly important for high-dimensional applications where high collinearity is expected. 
We refer readers to~\citet{wang2020simple} for more discussion on why this ``sum of single effects'' representation can effectively overcome  collinearity and the advantage of IBSS over deterministic search algorithms that return a single best model.  
\end{remark} 

\section{Simulation Studies for Bayesian Multi-task Variable Selection}\label{sec:vs-simu}
We conduct simulation  studies to illustrate the benefits of performing variable selection for multiple data sets jointly rather than independently.  
We generate  data sets according to~\eqref{eq:linear-regression} using the same $\sigma^2$ for all $K$ data sets. 
For the true model, we consider two types of activated covariates.  
For the first type, each covariate is activated in all $K$ data sets. We denote the set of these covariates by $S_{\mathrm{com}}^*$  and let $s_1^* = |S_{\rm{com}}^*|$ (subscript `com' means `common').  
For the second type, each covariate is activated in only one data set. We choose some $s_2^* > 0$ and draw $s_2^*$ covariates of the second type for each data set;   denote the set of covariates that are only activated in the $k$-th data set by $S_{\mathrm{pri}, k}^*$ (subscript `pri' means `private'). 
The true model size is given by $s^* = s_1^* + K s_2^*$, and $S_k^* = S_{\mathrm{com}}^* \cup S_{\mathrm{pri}, k}^*$ is the true set of activated covariates for the $k$-th data set. For each activated covariate, we sample its regression coefficient $\beta^{(k)}_i$ independently from the normal distribution $\cN(0, 0.6^2)$.  
For the design matrix, we sample each entry of $\bX^{(k)}\in\bbR^{n\times p}$ from the standard normal distribution. Finally, we generate  the response data by drawing $\by^{(k)} \sim \cN(\bX^{(k)}\bsbeta^{(k)}, \sigma^2\bI)$.

After generating the data set $\bigl( (\bX^{(k)}, \by^{(k)}) \bigr)_{k = 1}^K$, we run the IBSS algorithm to fit the muSuSiE model, which does variable selection simultaneously for $K$ data sets. 
For comparison, we also fit the SuSiE model using the algorithm of~\citet{wang2020simple} for each data set separately. 
We will refer to the former as the multi-task method and the latter as the \tcr{separate single-task analysis}. When running simulations, we set $L = s^* + K$ for the multi-task method and $L = s_1^* + s_2^* + 1$ for the separate analysis method. We have also tried other values of $L$ and observed that as long as $L$ is larger than the true number of activated covariates, its choice has negligible effect on the estimates; \tcr{the reason was explained in Remark~\ref{rmk:bernoulli}.}
For the hyperparameter $\bspi$ in the muSuSiE model,  we 
set it by~\eqref{eq:pi_to_omega}, and thus it suffices to specify $\omega_1, \dots, \omega_K$. 
When $K = 2$, we use $p^{-\omega_1} =  p^{-1.1}/2$ and $p^{-\omega_2} =  p^{-1.25}$; when $K = 5$, we use $\omega_k = 1.25 + 0.15 k$ for each $k$. 
Additionally, we tried joint Markov Chain Monte Carlo (MCMC), separate MCMC, and LASSO methods, for which the results and implementation details are deferred to Appendix C, supplementary materials.
 
For the multi-task method, recall that the probability of the $j$-th covariate being activated in the $k$-th dataset,  $r^{(k)}_j$, is defined in Equation \eqref{eq:sel-prob}.
Setting the threshold to $0.5$, we define the selected activated covariates from our multi-task method  by
$S_{\text{mu}, k} = \{j: r^{(k)}_j \geq 0.5\}$ (subscript `mu' means `multi-task'). 
For the standard SuSiE method, we use the \texttt{susie} function from the \texttt{susieR} package \citep{wang2020simple} to find the set of activated covariates, which we denote by $S_{\text{si}, k}$ (subscript `si' means `single-task'). 
To compare the performance of two approaches, we calculate  the sensitivity (sens) and precision (prec) by $\text{sens}(S_{k}) = \frac{\left\lvert S_{k} \cap S_k^* \right\lvert}{\left\lvert S_k^* \right\lvert}$, $\text{prec}(S_{k}) = \frac{\left\lvert S_{k} \cap S_k^* \right\lvert}{\left\lvert S_{k} \right\lvert}$,
where we let $S_{k} = S_{\text{mu},k}$ for the multi-task method and  $S_{k} = S_{\text{si},k}$ for the single-task approach.  

\begin{table}
\captionsetup{font=scriptsize}
\scriptsize
    \centering
    \begin{tabular}{|cccc|cc|cc|}
    \hline
    $p$ & $n$ & $s_1^*$ & $s_2^*$ & sens\_mu & sens\_si & prec\_mu & prec\_si\\
    \hline
600 & 100 & 10 & 2 & 0.4526& 0.2632& 0.9884& 0.9365\\ 
600 & 100 & 10 & 5 & 0.3456& 0.2045& 0.9747& 0.9258\\ 
\hline
1000 & 500 & 10 & 2 & 0.8121& 0.7063& 0.9962& 1\\ 
1000 & 500 & 10 & 5 & 0.7905& 0.7011& 0.9928& 0.9996\\ 
1000 & 500 & 25 & 2 & 0.8191& 0.696& 0.9985& 1 \\
1000 & 500 & 25 & 5 & 0.804& 0.6949& 0.9964& 0.9999\\ 
\hline
    \end{tabular}
    \caption{Simulation results for two data sets with $\sigma = 1$. For each setting, the result is averaged over $500$ replicates.}
    \label{tab:simu-res-K2-sigma1}
\end{table}

Table \ref{tab:simu-res-K2-sigma1} shows the simulation results for $\sigma^2 = 1$ and $K = 2$. We consider two scenarios: one with $p = 600$ and $n = 100$, and the other with $p = 1000$ and $n = 500$. From Table ~\ref{tab:simu-res-K2-sigma1}, we observe that when the sample size is small ($n = 100$), the multi-task method identifies more activated covariates than the single-task approach, resulting in higher sensitivity and precision. When the sample size is increased to $500$, the multi-task method still improves the sensitivity but has a slightly smaller  precision, because the multi-task method tends to treat the covariates with a very strong signal strength in only one data set as simultaneously activated in two data sets. Nevertheless, considering the significant improvement in sensitivity, the overall performance of the multi-task method seems much better. To further examine this phenomenon, we plot the sensitivity and specificity for $|s_1^*| = 10$ and $|s_2^*| = 2$ 
in Appendix C.1, supplementary materials; all other settings yield similar plots. 

The simulation results for $\sigma^2 = 1$ and $K = 5$ are shown in Table C.1 in Appendix C.1, supplementary materials. It is worth noting that when the sample size is small, compared with the case $K=2$, the advantage of the multi-task method with $K=5$ becomes much more significant and it outperforms the single-task method significantly in terms of both sensitivity and precision. When the sample size is large, the multi-task method is still better than the single-task method, but the performance is similar to that for $K = 2$.    
The simulation results for $\sigma^2 = 4$ (which represents a higher noise level) are shown in Appendix C.1, supplementary materials, where we have made very similar observations for the behavior of the two methods. 

In Appendix C.2, supplementary materials, we show the average computation time of the multi-task and separate single-task methods for each setting across $500$ replicates. The two methods take a similar amount of time when $K = 2$. However, as $K$ increases to $5$, the multi-task method takes more time than the separate analysis. For the latter, the time increases linearly with respect to $K$, while the computational time of muSuSiE increases exponentially. Additionally, when the number of individually activated covariates is small ($|s_2^*| = 2$),  the multi-task method is significantly faster than in the case with $|s_2^*| = 5$. 
The stability of our algorithm with respect to the choice of $\omega$ is discussed in Appendix C.3, supplementary materials. 

\section{Differential DAGs Analysis via  Multi-task Variable Selection}\label{sec:Multi-DAGs}

\subsection{From Multi-task Variable Selection to Joint Estimation of Multiple DAG models}
A highly useful application of the proposed Bayesian multi-task variable selection method is that it can be naturally extended to the multi-task structure learning problem, i.e., joint estimation of multiple DAG models.  
The existing Bayesian literature on the statistical learning of multiple graphs mostly focuses on undirected graphical models; see, for example,~\citet{danaher2014joint, peterson2015bayesian, gonccalves2016multi, niu2018latent, peterson2020bayesian, shaddox2020bayesian, peterson2021bayesian}.  
For the learning of multiple DAG models, \citet{oyen2012leveraging} proposed a greedy search algorithm,  \citet{yajima2015detecting} devised an MCMC sampler generalizing the method of~\citet{fronk2004markov}, and \citet{lee2022bayesian} proposed a method based on the joint empirical sparse Cholesky (JESC) prior. 
\citet{castelletti2020bayesian} developed the Bayesian methodology and MCMC algorithm for learning multiple essential graphs. 
For frequentists' approaches, \citet{liu2019joint} proposed the MPenPC method,  a two-stage approach based on the PC-stable algorithm,   
\citet{chen2021multi} proposed an iterative constrained optimization algorithm for calculating an $\ell^1 / \ell^2$-regularized maximum likelihood estimator, 
\citet{wang2020high} extended the well-known greedy equivalence search (GES) algorithm of~\citet{chickering2002optimal} to the case of multiple DAGs, and \citet{ghoshal2019direct} offered an algorithm that learns the difference between DAGs efficiently but seems only applicable to the case $K= 2$.  
The method we will propose in this section is motivated by the observation that once the order of variables is given, the IBSS algorithm for multi-task variable selection can be applied to quickly learn multiple DAG models simultaneously. 
Hence, all we need is just to combine IBSS with an MCMC sampler that traverses the order space. Compared with frequentists' methods, our algorithm can quantify the learning uncertainty since the estimators are averaged over the posterior distribution. 

Consider learning the DAG model for a single data set first. 
Let $\cG = (V, E)$ be a DAG with vertices $V = \{1, \cdots, p\}$ and set of directed edges $E\subset V\times V$.  
Let $|\cG|$ denote the cardinality of the edge set $E$. 
Let $\A \in \bbR^{p\times p}$ be the weighted adjacency matrix of the DAG $\cG$ such that $B_{ij}\neq 0$ if and only if $(i, j)\in E$.  
Suppose that the observed data matrix, denoted by $\bX \in \bbR^{n \times p}$, is generated by the following linear structural equation model (SEM),
\begin{equation}\label{eq:SEM}
\bX_j =  \sum_{i = 1}^p B_{i j} \bX_i + \be_j, \quad \text{ for } j = 1, \dots, p. 
\end{equation}
where $\bX_j$ denotes the $j$-th column of $\bX$, and for each $j$, the error vector $\be_j$ independently follows $\cN_n (0, \sigma_j^2 \bI )$. 
That is, each row of $\bX$ is an i.i.d. copy of a random vector $\sX = (\sX_1, \dots, \sX_p)$, whose distribution is given by  $\sX = \A^{\T} \sX + \be$ with $\be \sim\cN_p\left(0, \text{diag}(\sigma_1^2, \cdots, \sigma_p^2)\right)$. 

Since $\cG$ is acyclic, there exists at least one permutation (i.e., order) $\prec \, \in \bbS_p$ such that $B_{ij} = 0$ for any $j \prec i$ (i.e., $j$ precedes $i$ in the permutation $\prec$),  where $\bbS_p$ is the symmetric group of order $p$.  Hence, if the rows and columns of $\A$ are permuted according to $\prec$, the resulting matrix is strictly upper triangular. To determine which entries in $\A$ are not zero, we can convert this problem to $p$ variable selection problems.  
If we know that the DAG is consistent with the order $\prec$, for each $j$, we only need to identify the parent nodes for $j$ from the set $\{i \in [p] \colon  i \prec j \}$, which can be seen as a variable selection problem with response variable $\sX_j$ and candidate explanatory variables $\{ \sX_i \colon i \prec j  \}$. 
Combining  the results for all  $p$ variable selection problems, we get an estimate for the DAG model underlying the distribution of $\sX$. 
Unfortunately, the true order of nodes is usually unknown in practice and needs to be learned from the data. Since the order space $\bbS^p$ has cardinality $p!$, searching over $\bbS^p$ can be very time consuming, which is one major challenge in structure learning. To overcome this, various order MCMC methods have been proposed in the literature for efficiently generating samples from posterior distributions defined on  $\bbS^p$~\citep{koller2009probabilistic, kuipers2017partition, agrawal2018minimal, kuipers2022efficient}.

Next, consider the joint learning of multiple DAG models from  $K$ data sets, one for each data set. 
This problem, which henceforth is referred to as differential DAG analysis, is motivated by differential gene regulatory network (GRN) analysis in biology, where we may have gene data for samples from different tissues, developmental phases or case-control studies, and the goal is to see how the GRN changes across different samples~\citep{li2020bayesian}. 
Since the advent of the single-cell technology, differential GRN analysis has become increasingly important~\citep{fiers2018mapping, van2020scalable}.  
As in the multi-task variable selection problem, we assume the same $p$ covariates are observed in $K$ data sets, and use $\bX^{(k)} \in \bbR^{n_k \times p}$ to denote the data matrix for the $k$-th data set with sample size $n_k$. 
Denote the $K$ DAGs we want to learn by $( \cG^{(k)} = (V, E^{(k)}) )_{k = 1}^K$, which share the same node set $V = [p]$ and, a priori, are believed to share a large proportion of common edges. 
We further assume that  $\cG^{(1)}, \dots, \cG^{(K)}$ are ``permutation compatible,''  which means that for any $i \neq j$, if $(i,j)\in E^{(k)}$ for some $k \in [K]$, then $(j, i) \notin E^{(k')}$ for any $k' \in [K]$. In other words, we assume there exists a order shared by all the $K$ DAGs.  
This assumption has been widely used in the literature~\citep{liu2019joint, chen2021multi, lee2022bayesian}, and is very reasonable  for problems such as GRN analysis, where an edge may occur only in some data sets but generally does not change direction across data sets.  
Observe that if the  order $\prec$ is known, learning $K$ DAGs can be converted to $p$ multi-task variable selection problems. One just needs to repeatedly apply the IBSS algorithm we have proposed to select the parent nodes for each $j \in [p]$. Denote the resulting $K$ DAGs by $(\cG_\prec^{(k)} )_{k=1}^K$. 
We are interested in the case where the ordering is unknown. To average over the order space, we follow the existing order MCMC works to devise a Metropolis-Hastings algorithm on $\bbS^p$, which we describe in detail in the next subsection. 

\subsection{An Order MCMC Sampler for Differential DAG Analysis} \label{sec:order.mcmc}
We propose to consider the following Gibbs posterior distribution~\citep{jiang2008gibbs},  
\begin{equation}\label{eq:post-order}
P(\prec | (\cD^{(k)})_{k = 1}^K) \propto P( (\cG_\prec^{(k)} )_{k=1}^K | \prec)  \prod_{k = 1}^K \hat{P}(\cD^{(k)} |   \cG_\prec^{(k)})  , \quad \forall \, \prec \, \in \bbS^p, 
\end{equation}
where   $( \cG_\prec^{(k)} )_{k=1}^K$ denotes the DAGs we obtain by applying the IBSS algorithm  with ordering $\prec$. 
The product term in \eqref{eq:post-order} denotes the ``estimated'' likelihood function, which gives the estimated probability of observing the data given that $\cG_\prec^{(k)}$ is the underlying DAG model for the $k$-th data set. 
Denote by $\J^{\prec}_j = \{i \in[p] \colon i \prec j\}$ the index set of variables preceding $\sX_j$ in the order $\prec$.  Let 
\begin{equation}\label{eq:ibss2gm}
    \left(\widehat \sigma^2_{j, \prec},  (\bsalpha^{(l)}_{j, \prec})_{l = 1}^L,  ( \widehat \bsbeta^{(k, l)}_{j, \prec} )_{l \in [L], k \in [K]} \right) \leftarrow \text{IBSS}\left(  
    ( \{ \bX_i^{(k)} \colon i \in \J^{\prec}_j \} )_{k=1}^K, (\bX_j^{(k)})_{k=1}^K, L 
    \right)
\end{equation}
denote the output of Algorithm~\ref{alg:IBSS} for the multi-task variable selection problem with response vector $\bX_j$ and covariates $\{ \bX_i \colon i \in \J^{\prec}_j \}$. 
As in~\eqref{eq:post-reg-coef}, let $\widehat \bsbeta^{(k)}_{j, \prec} = \sum_{l=1}^L \widehat \bsbeta^{(k, l)}_{j, \prec}$ denote the posterior mean aggregated over $L$ single effects. 
Then, we can estimate the  likelihood of the DAGs $( \cG_\prec^{(k)} )_{k=1}^K$ by plugging in the estimates $( \widehat \bsbeta^{(k)}_{j, \prec} )_{k \in [K], j \in [p]}$ and $( \widehat \sigma^2_{j, \prec} )_{j \in [p]}$, which yields 
\begin{equation}\label{eq:likelihood}
\prod_{k = 1}^K \hat{P}(\cD^{(k)} | \cG_\prec^{(k)} ) = \prod_{k = 1}^K \prod_{j = 1}^p \prod_{i = 1}^{n_k} \Phi\left(\frac{X_{i j}^{(k)} -  \bX^{(k)}_{i,  \J^\prec_j } \, 
\widehat \bsbeta^{(k)}_{j, \prec}}{\widehat \sigma_{j, \prec} }\right),
\end{equation}
where $\Phi(x)$ is the density function for the standard normal distribution and $\bX^{(k)}_{i,  \J^\prec_j } $ denotes the row vector with entries $\{ \bX^{(k)}_{i  l } \colon l \in \J^\prec_j   \}$. 
The first term $ P( ( \cG_\prec^{(k)} )_{k=1}^K | \prec)$ in \eqref{eq:post-order} is the prior probability of the DAGs $( \cG_\prec^{(k)} )_{k=1}^K$ given order $\prec$, or more generally can be any positive function that penalizes DAGs with more edges.

Analogously to Equation~\eqref{eq:sel-prob}, given $(\bsalpha^{(l)})_{l=1}^L$, we define  $\tilde \alpha_{i, I} = 1 - \prod_{l=1}^L(1 -  \alpha^{(l)}_{i, I} )$, and we let $\tilde \alpha_{i, I}^{j, \prec}$ denote the corresponding quantity when $( \bsalpha^{(l)} )_{l=1}^L = (\bsalpha^{(l)}_{j, \prec})_{l=1}^L$, where $(\bsalpha^{(l)}_{j, \prec})_{l=1}^L$ is defined in~\eqref{eq:ibss2gm}.   
Write $\bsalpha_{j, \prec} = (\bsalpha^{(l)}_{j, \prec})_{l = 1}^L$, and define 
$\ak_k(\bsalpha_{j, \prec}) = \sum_{i \in [p]} \sum_{\{I\colon |I| = k\}}  \tilde \alpha_{i, I}^{j, \prec}$, which gives the estimated number of covariates that are activated in exactly $k$ distinct data sets.    
We define the prior term in~\eqref{eq:post-order} by
\begin{equation}\label{eq:penalty}
P( ( \cG_\prec^{(k)} )_{k=1}^K | \prec) =  \prod_{k = 1}^K \prod_{j = 1}^p p^{-\omega_k  \ak_k(\bsalpha_{j, \prec}) }. 
\end{equation}  
Recall that $\omega_1, \dots, \omega_K$ are the hyperparameters introduced in~\eqref{eq:prior-gamma} for \SSVS{} and can be seen as a reparameterization of $\bspi$ by~\eqref{eq:pi_to_omega}. 
The reasoning behind~\eqref{eq:penalty} is the same as that behind~\eqref{eq:prior-gamma}. 
Combining~\eqref{eq:likelihood} and~\eqref{eq:penalty}, we get a closed-form expression for the posterior defined in~\eqref{eq:post-order}. 
For later use, let $\bR^{(k)}_\prec \in [0, 1]^{p \times p}$ be the matrix such that
\begin{equation}\label{eq:def.bR}
    (\bR^{(k)}_\prec)_{ij} = \ind_{ \{ i\in\J_j^{\prec} \} } \sum_{ I \colon  k\in I } 
   \tilde \alpha_{i, I}^{j, \prec}.   
\end{equation}
That is, $(\bR^{(k)}_\prec)_{ij}$ is the estimated probability of the edge $(i, j)$ being in the $k$-th data set given the order $\prec$.

Given the target posterior distribution defined in~\eqref{eq:post-order}, we are now ready to introduce our Metropolis-Hastings algorithm for differential DAG analysis. 
Given the current state $\prec \, \in S_p$, we propose another state $\prec'$ from some proposal distribution  
$q(\cdot  | \prec)$ and accept it with probability
\begin{equation}\label{eq:MH}
    \min\left\{1, \frac{P(\prec ' | ( \cD^{(k)} )_{k = 1}^K)q(\prec | \prec')}{P(\prec | ( \cD^{(k)} )_{k = 1}^K) q(\prec' | \prec)}\right\}.
\end{equation}
We choose $q(\cdot | \prec)$ to be the uniform distribution on the set of permutations that can be obtained from $\prec$ by an adjacent transposition. That is, we randomly pick $j \in [p - 1]$ with equal probability  and then propose to move from $\prec = (i_1, \cdots, i_j, i_{j+1}, \cdots, i_p)$ to $\prec'  = (i_1, \cdots, i_{j+1}, i_j, \cdots, i_p)$.  
Clearly,  $q(\prec | \prec') = q(\prec' | \prec)$, and thus the proposal ratio in~\eqref{eq:MH} is always equal to $1$. 
Note that to calculate $P(\prec ' |  ( \cD^{(k)} )_{k = 1}^K)$, we need to run IBSS to find the DAGs $( \cG^{(k)}_\prec )_{k=1}^K$. 
Running this Metropolis-Hastings sampler for $T$ iterations (excluding burn-in), we obtain a sequence of permutations denoted by $( \prec_t )_{t=1}^T$. 
For each $\prec_t$, let $\bR^{(k)}_{\prec_t} \in [0, 1]^{p \times p}$ be the matrix defined in~\eqref{eq:def.bR}, and then $( \bR^{(k)}_{\prec_t} )_{t = 1}^T$ can be used  for making posterior inference. 
For example, to estimate the probability of the edge $i \rightarrow j$ being in the $k$-th DAG model, we can simply calculate the time average 
\begin{equation}\label{eq:def.Gamma.ij}
    \hat{R}_{ij}^{(k)} \coloneqq \frac{1}{T} \sum_{t = 1}^T  (\bR^{(k)}_{\prec_t})_{ij}. 
\end{equation} 

\begin{remark}
We do not consider learning Markov equivalent DAGs (i.e., DAGs that encode the same set of conditional independence relations) via order MCMC in this paper, which can be highly challenging due to the order bias~\citep{ellis2008learning}. However, we note that in multi-task settings,  the permutation compatible assumption allows us to learn the true ordering  more efficiently by pooling information from multiple data sets, which can help overcome the issue of Markov equivalence. 
We refer readers to~\citet{castelletti2020bayesian} for an algorithm that directly learns multiple Markov equivalence classes. 
\end{remark}

\section{Simulation Studies for Bayesian Differential DAG Analysis}\label{sec:dags-simu} 
We use simulation studies to investigate the performance of the order MCMC sampler described in Section~\ref{sec:order.mcmc}, which we denote by muSuSiE-DAG, in two scenarios: $K = 2, n_1 = n_2 = 300$, and $K = 5, n_1 = \cdots = n_5 = 240$. We fix the number of nodes $p$ to $100$ for all experiments.  
For each experiment, we generate the data according to the linear SEM~\eqref{eq:SEM} with true order given by $\prec \, = (1, 2, \dots, p)$.  
Hence, the true weighted adjacency matrices of the $K$ DAGs are strictly upper triangular. 
The true DAGs $( \cG^{(k)} )_{k=1}^K$ are then generated as follows. 
First, we generate a random edge set $\cE_{\text{com}}$ consistent with $\prec$ such that each edge in $\cE_{\text{com}}$ is activated in all the $K$ data sets. 
Second, for each $k \in [K]$, we generate an edge set $\cE_{\text{pri}}^{(k)}$ which consists of edges that are only activated in the $k$-th data set. 
Let $N_{\text{com}} = |\cE_{\text{com}}|$ denote the number of edges shared by all the $K$ DAGs and $N_{\text{pri}} = |\cE^{(k)}_{\text{pri}}|$ denote the number of private edges unique to each data set. 
We consider $N_{\text{com}}\in\{50, 100\}$, and $N_{\text{pri}}\in\{20, 50\}$ in the simulation studies. 
To generate the matrix $\A^{(k)}$ corresponding to DAG $\cG^{(k)}$ and the error variances of the $p$ variables, we follow~\cite{wang2020high} to sample the nonzero entries of $\A^{(k)}$ (determined by $\cG^{(k)}$) independently from the uniform distribution on $[-1, -0.1]\cup[0.1, 1]$ and sample the error variance of each variable independently from the uniform distribution on $[1, 2.25]$. Note that for each edge in $\cE_{\text{com}}$, its weights in the $K$ data sets are drawn independently.  

For each simulation setting, we generate $50$ replicates; the true DAG models and the data $(\cD^{(k)})_{k = 1}^K$ are re-sampled for each replicate. 
We compare the performance of six methods: PC algorithm or GES applied independently to each data set~\citep{spirtes2000causation, harris2013pc, chickering2002optimal}, the joint GES algorithm proposed by~\citet{wang2020high} which is a  state-of-the-art method for joint learning multiple DAG models with theoretical guarantees, MPenPC method of \citep{liu2019joint}, JESC method \citep{lee2022bayesian}, and muSuSiE-DAG. 
We implement  PC and GES algorithms using the \texttt{R} package \texttt{pcalg}~\citep{kalisch2012causal}, and MPenPC and JESC  using publicly available code with default parameters.
In the ensuing results, we select parameter values that yield the most robust empirical performance across our experiments. 
For the PC algorithm, we let the significance level used in the conditional independent tests be $0.005$, and for GES and joint GES methods, we let $\lambda = 2$, where $\lambda$ is the $l_0$-penalization parameter (scaled by $\log p$). For the muSuSiE-DAG method, we need to set the penalty parameters $\omega_1, \dots, \omega_K$. For $K = 2$, we use $p^{-\omega_1} = p^{-2} / 2$ and $p^{-\omega_2} = p^{-2.25}$, and the choice for $K = 5$ is given in Appendix~D, supplementary materials. The results for the four methods obtained by using other parameter values are also provided in Appendix~D, supplementary materials.

\begin{table}[h!]
\captionsetup{font=scriptsize}
\scriptsize
    \centering
    \begin{tabular}{|c|ccc|ccc|}
    \hline
       method  &  K & $N_{\text{com}}$ & $N_{\text{pri}}$ & $N_{\text{wrong}}$ & TP & FP \\
       \hline
       PC  & 2& 100 & 20 & 28.29 & 0.7822 & 4e-04\\
       GES & 2 & 100 & 20 & 19.67 & 0.8482 & 3e-04\\
       joint GES & 2 & 100 & 20& 15.4 & 0.9126 & 0.001\\ 
       MPenPC &2 & 100 & 20 & 76.27 & 0.8758 & 0.0126  \\ 
       JESC & 2 & 100 & 20 & 30.85 & 0.9257 & 0.0045  \\ 
       muSuSiE-DAG & 2 & 100 & 20&\textbf{12.91} & 0.9138 & 5e-04 \\
       \hline
       PC & 2 & 100 & 50 & 39.37 & 0.7475 & 3e-04 \\
       GES & 2 & 100 & 50 & 24.84 & 0.8505 & 6e-4\\
      joint GES & 2 & 100 & 50& 24.7& 0.9003 & 0.002 \\ 
        MPenPC & 2 & 100 & 50 & 62.65 & 0.8513 & 0.0083 \\
       JESC & 2 & 100 & 50 & 31.74 & 0.9316 & 0.0044  \\ 
      muSuSiE-DAG & 2 & 100 & 50&\textbf{18.45} & 0.9003 & 7e-04\\
       \hline
       PC & 2 & 50 & 50 & 21.9 & 0.8121 & 6e-04\\
       GES & 2 & 50 & 50 & 15.74 & 0.8514 & 2e-04\\
      joint GES & 2 & 50 & 50& 22.91& 0.883 & 0.0023 \\ 
        MPenPC & 2 & 50 & 50 & 85.64 & 0.9004 & 0.0154 \\
       JESC & 2 & 50 & 50 & 28.68 & 0.9302 & 0.0044 \\
      muSuSiE-DAG & 2 & 50 & 50&\textbf{15.03} & 0.8762 & 5e-04\\
       \hline
    \end{tabular}
    \caption{Simulation results for joint estimation of multiple DAG models with $K = 2$ (averaged over $50$ replicates). }
    \label{tab:sim-DAG-K2}
\end{table}

Table~\ref{tab:sim-DAG-K2} shows the results for $K = 2$, and the results for $K = 5$ are given in Appendix D, supplementary materials.
For each method, we calculate the average number of incorrect edges, denoted by $N_{\text{wrong}}$, the average true positive rate (TP) and the average false positive (FP) rate by ignoring the edge directions.   
As expected,  joint GES and muSuSiE-DAG have significantly larger true positive rates than PC and GES methods, since the former two methods are able to utilize information from all the $K$ data sets to infer common edges, which is particularly useful when an edge has a relatively small signal size in both data sets. 
Meanwhile, the two joint methods tend to have slightly larger false positive rates as well, since an edge with a very large signal size in one data set is likely to be identified by the joint method as existing concurrently in both data sets. 
However, note that the false positive rate of muSuSiE-DAG is still comparable to that of PC and GES and is much smaller than that of joint GES.  
\tcr{Both MPenPC and JESC   have high TP and FP rates, and JESC seems to perform significantly better than MPenPC. 
}
Overall, muSuSiE-DAG has the best performance among all the six methods in all settings, and its advantage is more significant when the ratio $N_{\text{com}}/N_{\text{pri}}$ is larger. The convergence of our order MCMC is discussed in Appendix D.1, supplementary materials.

\section{A Real Data Example for Differential DAG Analysis}\label{sec:real.data}
To evaluate the performance of the proposed muSuSiE-DAG method in real data analysis, we consider a pre-processed gene expression microarray data set used in~\citet{wang2020high}, which consists of two groups of patients with ovarian cancer. The first group has $83$ patients who have enhanced expression of stromal genes that are associated with a lower survival rate. 
The second group has $168$ patients who have ovarian cancer of other subtypes. 
For both groups, we observe the expression levels of $p =76$ genes, which, according to the KEGG database \citep{kanehisa2012kegg}, participate in the apoptotic pathway. 
For more details about the original data set, see~\citet{tothill2008novel}. 
Let $\cG_1$ denote the underlying DAG model for the first group and $\cG_2$ denote that for the second. 
The objective of this real data analysis is to detect the differences between the two DAGs $\cG_1, \cG_2$, which may be associated with the  survival rate.  
As in Section~\ref{sec:dags-simu}, we compare the performance of four methods: PC, GES, joint GES and muSuSiE-DAG. Table \ref{tab:real-data} lists the number of edges detected by each method. 
The results for all four methods obtained by using other parameter values are provided in Appendix E, supplementary materials, where one can also find results obtained by combining PC, GES and joint GES with stability selection \citep{meinshausen2010stability}. 
The results clearly illustrate the differences between the four methods. 
First, the percentage of shared edges in the two estimated DAGs (i.e., the ``ratio'' column in Table \ref{tab:real-data}) is much larger for the two joint methods, which is consistent with both our theory and simulation results. For PC and GES, this ratio is always less than $0.3$ in all parameter settings we have tried; see Tables E.1 and E.2 in Appendix E, supplementary materials. This shows that when the sample size is not large, applying a structure learning method to two data sets separately is very likely to miss some gene-gene interactions existing in both gene regulatory networks. 
Second, joint GES has the largest shared ratio, and it is often much larger than that of muSuSiE-DAG. 
This is probably because joint GES is a two-step procedure where  the first step is to learn a large DAG $G^{\rm{union}}$,  and in the second step $G_1$ and $G_2$ are constructed separately under the constraint that they must be sub-DAGs of $G^{\rm{union}}$. If an edge only exists in one DAG or it exists in both but has very different regression coefficients in the two SEMs, it is not very likely to be included in $G^{\rm{union}}$ and thus cannot be detected in the second step of joint GES. 
Indeed, since $p = 76$ is relatively large and $n_1 = 83$ and $n_2 = 168$, we expect that more edges (especially those with small signal sizes) can be detected in $G_2$ than in $G_1$, which is observed for PC, GES and muSuSiE-DAG.  

\begin{table}
\captionsetup{font=scriptsize}
\scriptsize
\centering
\begin{tabular}{|c|c|ccccc|}
\hline
Method  &  Parameters & $|\cG_{1}|$ & $|\cG_{2}|$ & $|\cG_{1} \cap \cG_{2}|$ & $N_{\text{total}}$ & ratio\\
\hline
PC & $\alpha =  0.005 $& 33 & 60 & 18 & 75 & 0.24 \\
GES & $\lambda =  2 $& 99 & 148 & 43 & 204 & 0.2108 \\
joint GES &$\lambda =  2 $& 78 & 78 & 72 & 84 & 0.8571 \\
muSuSiE-DAG & $p^{-\omega_1} = p^{-1.5} / 2, \, p^{-\omega_2} = p^{-2}$& 36 & 94 & 35 & 95 & 0.3684 \\
\hline
\end{tabular}
\caption{Results for the real data analysis. 
$|\cG_k|$: number of edges in the estimated DAG for the $k$-th group; $|\cG_{1} \cap \cG_{2}|$: number of edges shared by both DAGs; $N_{\text{total}}$: total number of edges in two DAGs; ratio: the ratio of $|\cG_{1} \cap \cG_{2}|$ to $N_{\text{total}}$.}
    \label{tab:real-data}
\end{table}

\section{Concluding Remarks}\label{sec:conclusion}
In this paper, we study the Bayesian multi-task variable selection problem and prove a high-dimensional strong selection consistency result for the multi-task spike-and-slab variable selection (\SSVS{}) model we propose. 
By extending the SuSiE model of~\citet{wang2020simple} to multiple data sets,  
we show that \SSVS{} can be approximated by a model we call muSuSiE, which further enables us  to propose a variational Bayes algorithm, IBSS, for efficiently approximating the posterior distribution of \SSVS{}.  Simulation results show that, compared with performing variable selection separately for multiple data sets, the proposed method can achieve a significantly larger sensitivity at the cost of a slightly decreased precision.  
Next, we consider the problem of learning multiple DAG models. 
Observing that  we can quickly learn multiple DAGs simultaneously using IBSS given the order of the variables, we propose an efficient order MCMC sampler targeting a Gibbs posterior distribution on the order space.  Both simulation results and real data analysis show that the proposed algorithm is able to identify substantially more edges shared across the data sets while still controlling the false positive rate. 
 
This work also opens up some interesting problems for future research.  
First, we build the strong selection consistency for the \SSVS{} model while the variational algorithm we propose is based on the muSuSiE model. 
It would be interesting to investigate whether we can establish high-dimensional consistency results directly for the SuSiE or muSuSiE model under some mild conditions, which would serve as a more powerful theoretical guarantee for variational Bayesian variable selection. 
Second,  one can extend the posterior consistency result for the \SSVS{} model to multi-task structure learning, but this probably requires assuming some restrictive conditions such as strong faithfulness~\citep{nandy2018high}. 
Last,  the proposed algorithm for learning multiple DAGs can be seen as a combination of the IBSS algorithm and a vanilla Metropolis-Hastings algorithm on the order space. Hence, more advanced MCMC sampling techniques (e.g. parallel tempering) can be used to further accelerate the mixing of the sampler.  

\section*{Acknowledgements}
We thank Yuhao Wang for sharing with us the code for the joint GES method and the  pre-processed real data set. QZ was supported in part by NSF  grant DMS-2245591. 

\bibliographystyle{plainnat}
\bibliography{reference.bib}

\newpage
\appendix
\noindent\textbf{\LARGE Appendices}
\renewcommand\thesection{\Alph{section}}
\numberwithin{equation}{section}
\numberwithin{table}{section}
\numberwithin{figure}{section}

\section{Proof of Posterior Consistency for Bayesian Multi-task Variable Selection}\label{appx:proof-posterior-concen}

Before going into the proof details, we review our notation for the  multi-task variable selection problem in
Table~\ref{tab:notation-bvs}. 

\begin{longtable}{|p{0.15\textwidth} p{0.8\textwidth}|}
\hline
Notation & Definition\\
\hline
$[k]$ & $[k] = \{1, 2, \cdots, k\}$\\
$2^{[k]}$ & power set on $[k]$, i.e., $2^{[k]} = \{ S \colon S \subseteq [k] \}$\\
$|S|$ & cardinality of set $S$ \\ 
$K$ & number of data sets\\
$p$ & number of covariates in each data set\\
$n_k$ & sample size of the $k$-th data set\\ 
$\by^{(k)}$ & response vector for the $k$-th data set with dimension $n_k$ \\
$\bX^{(k)}$ &  design matrix for the $k$-th data set with dimension $n_k \times p$ \\
$\bsbeta^{(k)}$ & vector of regression coefficients for the $k$-th data set\\
$\bX^{(k)}_{S}$ & submatrix of $\bX^{(k)}$ containing columns indexed by $S$\\
$\bsPhi_{S}^{(k)}$ &  $  \bX_{S}^{(k)}\bigl((\bX_{S}^{(k)})^\T\bX_{S}^{(k)}\bigr)^{-1}(\bX_{S}^{(k)})^\T$ \\ 
$L$ & maximum number of activated covariates\\  
$\sigma^2$ & error variance \\
$\bsgamma$ & the set-valued vector such that $\gamma_j = I$ means that the $j$-th covariate is activated in the data sets indexed by the set $I \subseteq [K]$ \\
$|\bsgamma|$ & $\sum_{j=1}^p \ind_{\{\gamma_j \neq \emptyset \}}$, i.e., number of covariates activated in at least one data set \\
$a_k(\bsgamma)$ & number of covariates activated in $k$ distinct data sets according to $\bsgamma$ \\
$(\omega_k)_{k=1}^K$ & hyperparameter for the prior distribution on $\bsgamma$ \\ 
$\tau_j^{(k)}$ & prior variance of $\beta_j^{(k)}$ if it is activated\\
$\beta_j^{(k)*}$ & true vector of regression coefficients\\ 
$C_{\beta, k}$ & detection threshold for a covariate activated in $k$ distinct data sets\\
$m_j^*$ & $ \max\{m\in[K]\colon |\{k\in[K]\colon \bigl(\beta_j^{(k)*}\bigr)^2 \geq C_{\beta, m}\}| = m\}$ \\ 
$\bsgamma^*$ & true model defined by $\gamma_j^* = \{k\in[K]\colon \bigl(\beta_j^{(k)*}\bigr)^2 \geq C_{\beta, m_j^*}\}$\\
$S_k(\bsgamma)$ & $ \{j\in[p]\colon k \in \gamma_j\}$ \\
$S_k^*$ & $S_k(\bsgamma^*)$, i.e., set of influential covariates in the $k$-th data set \\
$B_1, B_2, B_3$ &  constants in high-dimensional assumptions \\
$ \consta, \constb, \nu$ & constants in high-dimensional assumptions \\
\hline
\caption{Notation for Bayesian multi-task variable selection} \label{tab:notation-bvs}
\end{longtable}

\subsection{Posterior Calculation}
By~\eqref{eq:linear-regression} and~\eqref{eq:cond-prior-beta}, we find that, after integrating out $(\bsbeta^{(k)})_{k \in [K]}$, the marginal likelihood for a model $\bsgamma$ is 
\begin{align*}
P\bigl( (\by^{(k)})_{k \in [K]} \mid \bsgamma\bigr) \propto & \prod_{k = 1}^K  \left\lvert\bsSigma_{0(k)}^{-1}\right\lvert^{1/2}\left\lvert\frac{(\bX_{S_k(\bsgamma)}^{(k)})^{\T}\bX_{S_k(\bsgamma)}^{(k)}}{\sigma^2} + \bsSigma_{0(k)}^{-1}\right\lvert^{-1/2}\\
\times \exp\Biggl\{-\frac{1}{2\sigma^2}&\left[(\by^{(k)})^{\T}\left(\bI_n - \bX_{S_k(\bsgamma)}^{(k)}\left((\bX_{S_k(\bsgamma)}^{(k)})^{\T}\bX_{S_k(\bsgamma)}^{(k)} + \sigma^2\bsSigma_{0(k)}^{-1}\right)^{-1}(\bX_{S_k(\bsgamma)}^{(k)})^{\T}\right)\by^{(k)}\right]\Biggr\},
\end{align*}
where $\bsSigma_{0(k)} = \tau \bI_{|S_k(\bsgamma)|}$. Denote
\begin{align*}
R_{S_k(\bsgamma)}^{(k)} &= (\by^{(k)})^{\T}\left(\bI_n - \bX_{S_k(\bsgamma)}^{(k)}\left((\bX_{S_k(\bsgamma)}^{(k)})^{\T}\bX_{S_k(\bsgamma)}^{(k)} + \sigma^2\bsSigma_{0(k)}^{-1}\right)^{-1}(\bX_{S_k(\bsgamma)}^{(k)})^{\T}\right)\by^{(k)},\\
R_{S_k(\bsgamma)}^{(k)*} &= (\by^{(k)})^{\T}\left(\bI - \bX_{S_k(\bsgamma)}^{(k)}\left(\bX_{S_k(\bsgamma)}^{(k)})^{\T}\bX_{S_k(\bsgamma)}^{(k)}\right)^{-1}(\bX_{S_k(\bsgamma)}^{(k)})^{\T}\right)\by^{(k)}\\
&= (\by^{(k)})^{\T}(\bI_n - \bsPhi_{S_k(\bsgamma)}^{(k)})\by^{(k)}.
\end{align*}
To simplify the notation, from now on we will omit superscript $(k)$ whenever the statement applies to all $k = 1, \dots, K$.  
For example, when we write $R_{S(\bsgamma)}$, it means $R^{(k)}_{S_k(\bsgamma)}$ for any $k\in[K]$. It is easy to check that we always have $R_{S(\bsgamma)}^* \leq R_{S(\bsgamma)}$. Indeed, letting $\bX_{S(\bsgamma)} = \bU_{n\times |S(\bsgamma)|}\bsLambda_{|S(\bsgamma)|\times |S(\bsgamma)|}\bV^{\T}_{|S(\bsgamma)|\times |S(\bsgamma)|}$ be the singular value decomposition of $\bX_{S(\bsgamma)}$,  we have 
\begin{align*}
    R_{S(\bsgamma)}^* =& \|\by\|_2^2 - \by^{\T} \bU \bU^{\T} \by, \\ 
    R_{S(\bsgamma)} =& \|\by\|_2^2 - \by^{\T} \bU\bsLambda(\bsLambda^2 + \sigma^2\bsSigma_{0}^{-1})^{-1}\bsLambda \bU^{\T} \by. 
\end{align*}
Observe that $\bsLambda(\bsLambda^2 + \sigma^2\bsSigma_{0}^{-1})^{-1}\bsLambda$ is a diagonal matrix with all diagonal entries being in $[0, 1]$. Hence, $R_{S(\bsgamma)}^* \leq R_{S(\bsgamma)}$.  
Let $D_{S(\bsgamma)}$ denote the determinant term, 
\[D_{S(\bsgamma)} = \left\lvert\bsSigma_{0}^{-1}\right\lvert^{1/2}\left\lvert\frac{\bX_{S(\bsgamma)}^{\T}\bX_{S(\bsgamma)}}{\sigma^2} + \bsSigma_{0}^{-1}\right\lvert^{-1/2} = \left\lvert\bI_{|S(\bsgamma)|} + \widetilde{\vb} \bX_{S(\bsgamma)}^{\T}\bX_{S(\bsgamma)}\right\lvert^{-1/2},\]
where the second equation follows from $\widetilde{\vb} = \vb / \sigma^2$ and our assumption that $\vb_j^{(k)} = \vb$ for $k\in[K]$ and $j\in[p]$. 
Using~\eqref{eq:prior-gamma}, we find that the posterior probability of $\bsgamma$ is
\[\post(\bsgamma\mid \fulldata) \propto \ind_{\{|\bsgamma | \leq L|\}}f(|\bsgamma|, L) \prod_{k = 1}^K \left\{ D_{S_k(\bsgamma)}\exp\left(-\frac{R_{S_k(\bsgamma)}^{(k)}}{2\sigma^2}\right)   p^{-\omega_k  a_k(\bsgamma)
} \right\}. \]

\subsection{Preliminary for Proof of Posterior Consistency}
In this section we prove lemmas that will be needed in the posterior consistency proof later. Recall that the superscript $(k)$ is dropped for ease of notation. Recall $\widetilde{\vb} = \vb / \sigma^2$. 
\begin{lemma}\label{lemma: det}
Under Conditions \ref{cond2.1} and \ref{cond2.2}, for any $\br, \bs \subseteq [p]$ the following hold. 
\begin{enumerate}
    \item If $\br \subset \bs$, we have
    \[\frac{D_{\br}}{D_{\bs}} \leq (1 + n\widetilde{\vb})^{|\bs\setminus \br| / 2}.\]
    \item If $\bs \subseteq \br$, we have
    \[\frac{D_{\br}}{D_{\bs}} \leq 1.\]
\end{enumerate}
\end{lemma}
\begin{proof}
For the first case, we have
\begin{align*}
\frac{D_{\br}^2}{D_{\bs}^2} &= \left\lvert\left(\bI + \widetilde{\vb}\bX_{\br}\bX_{\br}^{\T}\right)^{-1}\left(\bI + \widetilde{\vb}\bX_{\br}\bX_{\br}^{\T} + \widetilde{\vb}\bX_{\bs\setminus \br}\bX_{\bs\setminus \br}^{\T}\right)\right\lvert\\
&= \left\lvert\bI + \left(\bI + \widetilde{\vb}\bX_{\br}\bX_{\br}^{\T}\right)^{-1}\left(\widetilde{\vb}\bX_{\bs\setminus \br}\bX^{\T}_{\bs \setminus \br}\right) \right\lvert\\
&= \left\lvert\bI + \widetilde{\vb} \bX_{\bs \setminus \br}^{\T}\left(\bI + \widetilde{\vb}\bX_{\br}\bX_{\br}^{\T}\right)^{-1}\bX_{\bs \setminus \br} \right\lvert\\
&\leq \left\lvert\bI + \widetilde{\vb}\bX_{\bs \setminus \br}^{\T}\bX_{\bs \setminus \br}\right\lvert, 
\end{align*}
where the third equation follows from Sylvester's determinant theorem and the last inequality  follows from the fact that if $\bA, \bB, \bA - \bB$ are all positive definite, then $|\bA| > |\bB|$. 
Let $\lambda_i(\bA)$ denote the $i$-th eigenvalue of the matrix $\bA$. Recall that 
\begin{equation}\label{eq:det}
|\bI + \widetilde{\vb}\bX_{\bs \setminus \br}^{\T}\bX_{\bs \setminus \br}| = \prod_{i = 1}^{|T \setminus S|}\lambda_i(\bI + \widetilde{\vb}\bX_{\bs \setminus \br}^{\T}\bX_{\bs \setminus \br}). 
\end{equation} 
By Condition \ref{cond2.1},
\[\sum_{i = 1}^{|T \setminus S|} \lambda_i(\bI + \widetilde{\vb}\bX_{\bs \setminus \br}^{\T}\bX_{\bs \setminus \br}) = \text{Trace}(\bI + \widetilde{\vb}\bX_{\bs \setminus \br}^{\T}\bX_{\bs \setminus \br}) = |\bs \setminus \br|(1 + n\widetilde{\vb}).\]
By the inequality of geometric and arithmetic means, this shows that \eqref{eq:det} is bounded from above by $ (1 + n\widetilde{\vb})^{|\bs\setminus \br|  }$. This yields the first bound given in the lemma. 
 
For the second case, we have
\begin{align*}
\frac{D_{\br}^2}{D_{\bs}^2} &= \left\lvert\left(\bI + \widetilde{\vb}\bX_{\br}\bX_{\br}^{\T}\right)^{-1}\left(\bI + \widetilde{\vb}\bX_{\br}\bX_{\br}^{\T} - \widetilde{\vb}\bX_{\br\setminus \bs}\bX_{\br\setminus \bs}^{\T}\right)\right\lvert\\
&= \left\lvert\bI  -\left(\bI + \widetilde{\vb}\bX_{\br}\bX_{\br}^{\T}\right)^{-1}\left(\widetilde{\vb}\bX_{\br\setminus \bs}\bX^{\T}_{\br \setminus \bs}\right) \right\lvert\\
&\leq 1.
\end{align*}
The proof is complete. 
\end{proof}

The next lemma bounds the difference between $R_{\br}$ and $R_{\br}^*$.

\begin{lemma}\label{lm:high-prob}
Under Conditions \ref{cond1.1}, \ref{cond2.1} and \ref{cond4}, we have
\begin{enumerate}
\item $\P[\frac{1}{2}  n\sigma^2\leq \|\be\|_2^2 \leq \frac{3}{2} n\sigma^2] \geq 1 - 2p^{-1}$.
\item $\P[\|\by\|_2^2 \leq 3n\sigma^2 \L_1\log p] \geq 1 - 2p^{-1}$.
\item $\P[R_{\br} - R_{\br}^* \leq  3\L_1 \sigma^2\log p  / (\nu\widetilde{\vb})] \geq 1- 2p^{-1}$ for any index set $\br$.
\end{enumerate}
\end{lemma}
\begin{remark}
Since we assume $K$ is fixed, by a union bound, it follows that with probability at least $1 - 2 K p^{-1} = 1 - O(p^{-1})$, $\frac{1}{2}  n\sigma^2\leq \|\be^{(k)}\|_2^2 \leq \frac{3}{2} n\sigma^2$ for all $k = 1, \dots, K$. The other two statements can be extended to all $K$ data sets analogously. 
\end{remark}
\begin{proof}
For part 1, we know that $\|\be\|_2^2/\sigma^2 \sim \chi_n^2$. By the concentration of the chi-square distribution and Condition \ref{cond4}, we have
\[\P\left[\left\lvert\frac{\|\be\|_2^2}{n\sigma^2} - 1\right\lvert \geq \frac{1}{2}\right] \leq 2e^{-n/25} \leq 2p^{-1},\]
which implies
\[\P\left[\frac{1}{2}  n\sigma^2\leq \|\be\|_2^2 \leq \frac{3}{2} n\sigma^2\right] \geq 1 - 2p^{-1}.\]

For part 2, by the Cauchy-Schwartz inequality,
\[\|\by\|_2^2 = \|\bX\bsbeta^* + \be\|_2^2 \leq 2\|\bX\bsbeta^*\|_2^2 + 2\|\be\|_2^2.\]
Using part 1, we obtain that 
\[\P[\|\by\|_2^2 \geq 2\|\bX\bsbeta^*\|_2^2 + 3n\sigma^2] \leq 2p^{-1}.\]
The bound then can be proved by invoking Condition \ref{cond1.1}.

For part 3, by the Sherman-Morrison-Woodbury identity, we have
\begin{align*}
0 \leq R_{\br} - R_{\br}^* &= \by^{\T}\bX_{\br}\left((\bX_{\br}^{\T}\bX_{\br})^{-1} - (\bX_{\br}^{\T}\bX_{\br} + \widetilde{\vb}\bI)^{-1}\right)\bX_{\br}^{\T}\by^{\T}\\
& = \by^{\T}\bX_{\br}(\bX_{\br}^{\T}\bX_{\br})^{-1}(\widetilde{\vb}\bI + (\bX_{\br}^{\T}\bX_{\br})^{-1})^{-1}(\bX_{\br}^{\T}\bX_{\br})^{-1}\bX_{\br}^{\T}\by\\
& \leq (n\widetilde{\vb})^{-1} n \by^{\T}\bX_{\br}(\bX_{\br}^{\T}\bX_{\br})^{-2}\bX_{\br}^{\T}\by.
\end{align*}
The last inequality is due to the fact that $\widetilde{\vb} \bI \preceq  \widetilde{\vb} \bI  + (\bX_{\br}^{\T}\bX_{\br})^{-1}$. Let $\bM = n\bX_{\br}(\bX_{\br}^{\T}\bX_{\br})^{-2}\bX_{\br}^{\T}$,
where the notation $\bA \preceq \bB$ means $\bB - \bA$ is positive semidefinite. 
By Condition \ref{cond2.2}, 
\[\lambda_{\max}(\bM) = \lambda_{\max}(n(\bX_{\br}^{\T}\bX_{\br})^{-1})\leq \frac{1}{\nu}.\]
That is, $\bM$ has bounded eigenvalues. Thus, by part 2,
\begin{align*}
\P\left[R_{\br} - R_{\br}^* \geq \frac{3 \L_1\sigma^2\log p}{\widetilde{\vb}\nu}\right] & \leq \P\left[\by^{\T} \bM\by \geq 
\frac{3 n \L_1\sigma^2\log p}{\nu}\right] \\
& \leq \P\left[\|\by\|_2^2 \geq 3 \L_1 n \sigma^2 \log p\right] \leq 2p^{-1}, 
\end{align*}
which completes the proof. 
\end{proof}

The third lemma is to bound the quadratic forms of residuals.
\begin{lemma}
Under Conditions \ref{cond2.1} and \ref{cond2.3}, the following hold. 
\begin{enumerate}
\item  For any distinct pair $(\br_1, \br_2)$ satisfying $\br_1 \subset \br_2$ and $|\br_2| \leq L$, we have
\[\lambda_{\min}\left(\bX_{\br_2 \setminus \br_1}^{\T}(\bI_n - \bsPhi_{\br_1})\bX_{\br_2 \setminus \br_1}\right) \geq n\nu.\]
\item For any distinct pair $(\br_1, \br_2)$ satisfying $\br_1 \subset \br_2$ and $|\br_2| \leq L$, we have
\[\P\left[\max_{\br_1 \subset \br_2}\frac{\be^{\T} (\bsPhi_{\br_2} - \bsPhi_{\br_1})\be}{|\br_2| - |\br_1|} \leq \L_3\sigma^2\log p\right]\geq 1 - p^{-1}.\]
\end{enumerate}
Here $\bsPhi_{\br}$ is defined by
\[\bsPhi_{\br} = \bX_{\br}\left(\bX_{\br}^\top \bX_{\br}\right)^{-1}\bX_{\br}^\top.\]
\end{lemma}

\begin{proof}
For part 1, if we write $\bX_{\br_2} = [\bX_{\br_1}, \bX_{\br_2 \setminus \br_1}]$, by the   block matrix inversion formula, the lower right component of $(n^{-1}\bX_{\br_2}^{\T}\bX_{\br_2}^{-1})^{-1}$ is $(n^{-1}\bX_{\br_2 \setminus \br_1}^{\T}(\bI_n - \bsPhi_{\br_1})\bX_{\br_2 \setminus \br_1})^{-1}$, which implies the asserted bound. 

For part 2, by the block matrix inversion formula, we have
\[\bsPhi_{\br\cup \{k\}} - \bsPhi_{\br} = \frac{(\bI - \bsPhi_{\br})\bX_k\bX_k^{\T}(\bI - \bsPhi_{\br})}{\bX_k^{\T}(\bI - \bsPhi_S)\bX_k}.\]
Hence,
\[\be^{\T} (\bsPhi_{\br\cup \{k\}} - \bsPhi_{\br})\be = \frac{\left(\be^{\T}(\bI - \bsPhi_{\br})\bX_k\right)^2/ n}{\bX_k^{\T}(\bI - \bsPhi_{\br})\bX_k / n}.\]
Due to part 1, the denominator $\bX_k^{\T}(\bI - \bsPhi_{\br})\bX_k / n \geq \nu$. For the numerator, define the random variable 
\[V(Z) \coloneqq \max_{|\br| \leq L, k\notin \br}\frac{1}{\sqrt{n}}\left\lvert Z^{\T}(\bI - \bsPhi_{\br})\bX_k\right\lvert,\]
where $Z\sim \cN(0, \bI_n)$. For any two vectors $Z, Z'\in \bbR^n$, by Condition \ref{cond2.1},
\begin{align*}
|V(Z) - V(Z')| &\leq \max_{|\br| \leq L, k\notin \br}\frac{1}{\sqrt{n}}\left\lvert (Z - Z')^{\T}(\bI - \bsPhi_{\br})\bX_k\right\lvert\\
& \leq \frac{1}{\sqrt{n}}\|(\bI - \bsPhi_{\br})\bX_k\|_2\|Z - Z'\|_2 \leq \|Z - Z'\|_2. 
\end{align*}
Thus, by the concentration of measures for Lipschitz functions of Gaussian random variables, we have
\[\P ( V(Z) \geq \bbE[V(Z)] + t )\leq \exp(-t^2 / 2).\]
Due to Condition \ref{cond2.3},
\[\bbE[V(Z)]\leq \frac{1}{2}\sqrt{\L_3\nu\log p}.\]
Thus,
\[\P\left[V(Z) \geq \frac{1}{2}\sqrt{\L_3\nu\log p} + \frac{1}{2}\sqrt{\L_3\nu\log p}\right] \leq \exp\left(-\frac{1}{8}\L_3\nu\log p\right) \leq p^{-1}.\]
Hence,
\[\P\left[\max_{|\br| \leq L, k\notin \br}\be^{\T} (\bsPhi_{\br\cup \{k\}} - \bsPhi_{\br})\be \geq \L_3\sigma^2\log p\right] \leq p^{-1},\]
which implies part 2.
\end{proof}

\subsection{Proof of Posterior Consistency}
We  prove the posterior consistency in this section. 
For simplicity, all universal constants other than $c_1$, $c_2$ and $c_3$ are denoted by $C$ or $C'$.   

\begin{proof} 
Throughout our proof, we always consider the event set on which the events in Lemma 2 (parts 1, 2 and 3) and Lemma 3 (part 2) all happen, which occurs with probability at least $1 - c_2 p^{-c_3}$ for some universal constants $c_2, c_3 > 0$. 

We divide the proof into two parts depending on whether the model being considered is overfitted or underfitted.  
First, consider the overfitted case. Let 
\[\bbM_{1\bsgamma} = \{\bsgamma \colon |\bsgamma|\leq L, S_k^* \subseteq S_k(\bsgamma), \, \forall k\in[K]\}\]
denote the collection of all models other than the true model $\bsgamma^*$ that include all influential covariates. 
Fix an arbitrary $\bsgamma \in \bbM_{1\bsgamma}$, and note that  $l = \sum_{k=1}^K |S_k(\bsgamma) \setminus S_k^*| \geq 1$. Denote $m_j = |\bsgamma_j|$ and recall that $m_j^* = |\bsgamma_j^*|$.
Let 
\begin{equation}\label{eq:l_k}
l_k = a_k(\bsgamma) = \sum_{j\in [p]} \ind_{\{m_j = k\}}, \quad \quad l^*_k = a_k(\bsgamma^*).
\end{equation}
It follows that \[\sum_{k = 1}^K |S_k(\bsgamma)| = \sum_{k = 1}^K k \, l_k.\]
Since $\bsgamma$ is overfitted, we have  
\begin{equation}\label{eq:r-to-gamma-diff}
\sum_{k = 1}^K |S_k(\bsgamma) \setminus S_k^*| = \sum_{k = 1}^K k(l_k - l_k^*), 
\end{equation}
which implies that 
\begin{equation}\label{eq:gamma-to-r-diff}
|\bsgamma| - |\bsgamma^*| = \sum_{k = 1}^K (l_k - l_k^*) \leq \sum_{k = 1}^K |S_k(\bsgamma) \setminus S_k^*| = l. 
\end{equation}

By \eqref{eq:gamma-to-r-diff}, Lemma 1, and Conditions \ref{cond3.1}, \ref{cond3.2} and \ref{cond3.4}, we have
\begin{align*}
\frac{\post (\bsgamma \mid  \fulldata )}{\post (\bsgamma^* \mid \fulldata)} &= \frac{f(|\bsgamma|, L)}{f(|\bsgamma^*|, L)}\prod_{k = 1}^K \left(p^{-(l_k - l_k^*)\omega_k}\right)\prod_{k = 1}^K\left( \frac{D_{S_k(\bsgamma)}\exp\left(-\frac{1}{2\sigma^2}R_{S_k(\bsgamma)}^{(k)}\right)}{D_{S_k^*}\exp\left(-\frac{1}{2\sigma^2}R_{S_k^*}^{(k)}\right)}\right)\\
& \leq C p^{l\constb}p^{-\sum_{k = 1}^K (l_k - l_k^*) \omega_k}\prod_{k = 1}^K \exp\left(-\frac{1}{2\sigma^2}\left(R^{(k)}_{S_k(\bsgamma)} - R^{(k)}_{S_k^*}\right)\right).
\end{align*}
By Condition \ref{cond1.2} and Lemma 3, if $|\br \setminus \br^*| \geq 1$, we have
\begin{align*}
R_{\br^*}^* - R_{\br}^* &= \|(\bsPhi_{\br} - \bsPhi_{\br^*}) \by \|_2^2 = \|(\bsPhi_{\br} - \bsPhi_{\br^*})\bX_{-\br^*}\bsbeta_{-\br^*}^* + (\bsPhi_{\br} - \bsPhi_{\br^*})\be\|_2^2\\
&\leq 2\|(\bsPhi_{\br} - \bsPhi_{\br^*})\bX_{-\br^*}\bsbeta_{-\br^*}^*\|_2^2 + 2\|(\bsPhi_{\br} - \bsPhi_{\br^*})\be\|_2^2\\
&\leq 2\L_2\sigma^2\log p + 2 |\br \setminus \br^*| \L_3\sigma^2\log p,
\end{align*}
with probability at least $1 - c_2 p^{-c_3}$.  Combining it with Lemma 2, we have 
\begin{align*}
R_{S_k^*} - R_{S_k(\bsgamma)} &\leq R_{S_k^*} - R^*_{S_k(\bsgamma)} = R_{S_k^*} - R_{S_k^*}^* + R_{S_k^*}^* - R_{S_k(\bsgamma)}^*\\
& \leq 3\L_1\log p \sigma^2 / (\nu\widetilde\tau) + 2\L_2\sigma^2\log p + 2 |S_k(\bsgamma) \setminus S_k^*| \L_3\sigma^2\log p\\
& \leq 3 |S_k(\bsgamma) \setminus S_k^*| \left(\L_1 / (\nu\widetilde \tau) + \L_2 + \L_3\right)\sigma^2\log p,
\end{align*}
for all $K$ data sets with probability at least $1 - c_2 p^{-c_3}$. 

By Equation~\eqref{eq:kappa-condition} and Condition \ref{cond3.3}, the posterior ratio becomes
\begin{align*}
\frac{\post(\bsgamma \mid \fulldata)}{\post(\bsgamma^* \mid \fulldata)} &\leq C p^{l\constb} p^{-\sum_{k = 1}^K (l_k - l_k^*)\omega_k}p^{\sum_{k = 1}^K 3|S_k(\bsgamma) \setminus S_k^*| \left(\L_1 / (\nu\widetilde \tau) + \L_2 + \L_3\right) / 2}\\
&= C p^{l \constb \sum_{k = 1}^K k(l_k - l_k^*)} p^{-\sum_{k = 1}^K k(l_k - l_k^*) (\omega_k / k)} p^{\left(3\left(\L_1 / (\nu\widetilde \tau) + \L_2 + \L_3\right) / 2\right) \sum_{k = 1}^K k(l_k -l_k^*)}\\
&\leq C p^{-2\sum_{k = 1}^K|S_k(\bsgamma)\setminus S_k^*|},
\end{align*}
with probability at least $1 - c_2 p^{-c_3}$, where we have used \eqref{eq:r-to-gamma-diff} in the second equality and the third inequality. Hence,
\begin{align}\label{eq:pos-ratio-overfit}
\begin{split}
\frac{\post(\bbM_{1\bsgamma} \mid \fulldata)}{\post(\bsgamma^* | \fulldata)} & = \sum_{\bsgamma\in\bbM_{1\bsgamma}}\frac{\post(\bsgamma | \fulldata)}{\post(\bsgamma^* | \fulldata)}\\
&\leq \sum_{l = 1}^{\infty}C (Kp)^l p^{-2l}\leq C' p^{-1},  
\end{split}
\end{align}
with probability at least $1 - c_2 p^{-c_3}$, where the first inequality in \eqref{eq:pos-ratio-overfit} follows from the fact that there are at most $(Kp)^l$ models that satisfy $\sum_{k = 1}^K |S_k(\bsgamma) \setminus S_k^*| = l$.

Second, consider the underfitted case. Let 
\[\bbM_{2\bsgamma} = \{\bsgamma\colon |\bsgamma|\leq L, \, l = \sum_{k=1}^K |S_k^* \setminus S_k(\bsgamma)| \geq 1\}\]
be the collection of models which do not include at least one influential covariate. 
Fix an arbitrary $\bsgamma \in \bbM_{2\bsgamma}$, and let $l = \sum_{k=1}^K|S_k^*\setminus S_k(\bsgamma)|$, $\widetilde S_k(\bsgamma) = S_k(\bsgamma) \cup S_k^*$, and $l_0 = \sum_{k=1}^K S_k(\bsgamma)$.
Let $\widetilde \bsgamma$ be defined by $\widetilde \bsgamma_j = \{k\in[K]\colon j \in \widetilde S_k(\bsgamma)\}$. Then, $\sum_{k=1}^K |\widetilde S_k(\bsgamma) \setminus S_k(\bsgamma)| = l$ and $\sum_{k=1}^K |\widetilde S_k(\bsgamma) \setminus S_k^*| = l + l_0 - \sum_{k=1}^K S_k^*$. Let $\widetilde m_j$ and $\widetilde l_k$ be defined in the same manner as \eqref{eq:l_k} by replacing $\bsgamma_j$ with $\widetilde \bsgamma_j$. Then,
$|\bsgamma| \leq |\widetilde \bsgamma|$ and
\[\sum_{k = 1}^K |\widetilde S_k(\bsgamma) \setminus S_k(\bsgamma)| = \sum_{k = 1}^K k(\widetilde l_k - l_k).\]

By Lemma 1, Conditions \ref{cond3.1}, \ref{cond3.2} and \ref{cond3.4},
\begin{align*}
\frac{\post(\bsgamma \mid \fulldata)}{\post(\widetilde \bsgamma \mid \fulldata)} &= \frac{f(|\bsgamma|, L)}{f(|\widetilde \bsgamma|, L)}\prod_{k = 1}^K \left(p^{(\tilde l_k - l_k)\omega_k}
\frac{D_{S_k(\bsgamma)}\exp\left(-\frac{1}{2\sigma^2}R_{S_k(\bsgamma)}^{(k)}\right)}{D_{\widetilde S_k(\bsgamma)}\exp\left(-\frac{1}{2\sigma^2}R_{\widetilde S_k(\bsgamma)}^{(k)}\right)}\right)\\
&\leq C p^{\sum_{k = 1}^K \mid \widetilde S_k(\bsgamma) \setminus S_k(\bsgamma)\mid \consta} p^{\sum_{k = 1}^K (\tilde l_k - l_k) \omega_k}\prod_{k = 1}^K\exp\left(-\frac{1}{2\sigma^2}(R^{(k)}_{S_k(\bsgamma)} - R^{(k)}_{\widetilde S_k(\bsgamma)})\right).
\end{align*}
By Condition \ref{cond1.2} and Lemma 3, if $|\widetilde \br \setminus \br| \geq 1$, we have, with probability at least $1 - c_2 p^{-c_3}$
\begin{align*}
R_{\br}^{*} - R_{\widetilde \br }^{*} &=  \by^{\T}(\bsPhi_{\widetilde \br }  - \bsPhi_{\br})\by = \|(\bsPhi_{\widetilde \br} - \bsPhi_{\br})(\bX_{\br^{*}} \bsbeta_{\br^{*}}^{*} + \bX_{-\br^*}\bsbeta_{-\br^*}^* + \be ) \|_2^2\\
&\geq \left(\|(\bsPhi_{\widetilde \br} - \bsPhi_{\br})\bX_{\br^*}\bsbeta_{\br^*}^*\|_2 - \|(\bsPhi_{\widetilde \br} - \bsPhi_{\br})\bX_{-\br^*}\bsbeta_{-\br^*}^*\|_2 + \|(\bsPhi_{\widetilde \br} - \bsPhi_{\br})\be\|_2\right)^2\\
&\geq \left(\|(\bsPhi_{\widetilde \br} - \bsPhi_{\br})\bX_{\br^*}\bsbeta_{\br^*}^*\|_2 - \sqrt{\L_2\sigma^2\log p} - \sqrt{|\widetilde \br \setminus \br| \L_3\sigma^2\log p}\right)^2.
\end{align*}
Due to Condition \ref{cond5} and Lemma 3, we have
\begin{align*}
\|(\bsPhi_{\widetilde \br } - \bsPhi_{\br}) \bX_{\br^{*}} \bsbeta_{\br^{*}}^{*} \|_2^2 &= \|\left(\bI - \bsPhi_{\br} \right) \bX_{\br^{*}} \bsbeta_{\br^{*}}^{*}\|_2^2 \\
& = \|\left(\bI - \bsPhi_{\br} \right)\bX_{\br^{*} \setminus \br} \bsbeta_{\br^{*} \setminus \br}^{*}\|_2^2\\
& \geq n\nu\|\bsbeta_{\br^{*} \setminus \br}^{*}\|_2^2\\
& \geq 8 \left\lvert \br^* \setminus \br\right\lvert (\L_2 + \L_3)\sigma^2\log p. 
\end{align*}
Thus,
\[R_{\br}^{*} - R_{\widetilde \br}^{*} \geq \frac{1}{4} \|\left(\bsPhi_{\widetilde \br} - \bsPhi_{\br} \right)\bX_{\br^{*}} \bsbeta_{\br^{*}}^{*}\|_2^2,\]
with probability at least $1 - c_2 p^{-c_3}$. 
Combining it with Lemma 2, we have
\[R_{S_k(\bsgamma) }  - R_{\widetilde S_k(\bsgamma) }  \geq R_{S_k(\bsgamma) }^{*} - R_{\widetilde S_k(\bsgamma)}^{*} + R_{\widetilde S_k(\bsgamma) }^{*} - R_{\widetilde S_k(\bsgamma) } \geq \frac{n \nu \|\bsbeta_{S_k^{*} \setminus S_k(\bsgamma) }^{*}\|_2^2}{4} - \frac{3 \L_1\sigma^2\log p}{\nu \widetilde \tau},\]
for all $k\in[K]$ with probability at least $1 - c_2 p^{-c_3}$. Observe that
\[\sum_{k = 1}^K \|\bsbeta_{S_k^* \setminus S_k(\bsgamma)}^{(k)*}\|_2^2 \geq \sum_{j \in [p]} |\bsgamma_j^* \setminus \bsgamma_j| C_{\beta, m_j^*}.\]
Due to Condition \ref{cond5} and Equation~\eqref{eq:kappa-condition}, the posterior ratio becomes
\begin{align*}
& \frac{\post(\bsgamma \mid \fulldata)}{\post(\widetilde \bsgamma \mid \fulldata)}\\
\leq &p^{\sum_{k = 1}^K |\widetilde S_k(\bsgamma) \setminus S_k(\bsgamma)| \consta} p^{\sum_{k = 1}^K \omega_k (\tilde l_k - l_k)} p^{-\sum_{k = 1}^K |S_k^* \setminus S_k(\bsgamma)| (\consta + 2)} p^{-\sum_{j \in [p]} |\bsgamma_j^* \setminus \bsgamma_j| (\omega_{m_j^*} / m_j^*)} \\
\leq &C p^{-2\sum_{k = 1}^{K}|\tilde S_k(\bsgamma) \setminus S_k(\bsgamma)|},
\end{align*}
where in the last inequality we have used
\begin{equation}\label{eq:omega_claim}
\sum_{k = 1}^K \omega_k (\tilde l_k - l_k) - \sum_{j\in [p]}|\bsgamma_j^* \setminus \bsgamma_j| (\omega_{m_j^*} / m_j^*) \leq 0, 
\end{equation}
for which we will give a proof at the end. By the result for the overfitted case, the posterior ratio becomes
\begin{align*}
\frac{\post(\bsgamma \mid \fulldata)}{\post(\bsgamma^* \mid \fulldata)} &= \frac{\post(\bsgamma \mid \fulldata)}{\post(\widetilde\bsgamma \mid \fulldata )} \frac{\post(\widetilde \bsgamma \mid \fulldata)}{\post(\bsgamma^* \mid \fulldata)}\\
&\leq C p^{-2\left(\sum_{k = 1}^K \left\lvert S_k(\bsgamma)\setminus S_k^*\right\lvert + \sum_{k = 1}^K \left\lvert S_k^*\setminus S_k(\bsgamma)\right\lvert \right)},
\end{align*}
with probability at least $1 - c_2 p^{-c_3}$. It follows that 
\begin{align}\label{eq:pos-ratio-underfit}
\begin{split}
\frac{\post(\bbM_{2\bsgamma} \mid \fulldata)}{\post(\bsgamma^* \mid \fulldata)} &= \sum_{\bsgamma\in\bbM_{2\bsgamma}}\frac{\post(\bsgamma \mid \fulldata)}{\post(\bsgamma^* \mid \fulldata)}\\
&\leq \sum_{l_1 = 0}^{\infty}\sum_{l_2 = 1}^{\infty}C (Kp)^{l_1 + l_2}p^{-2l_1 - 2l_2} \leq C' p^{-1}.    
\end{split}
\end{align}
with probability at least $1 - c_2 p^{-c_3}$. In the first inequality, we use the fact that there are at most $ (Kp)^{l_1 + l_2}$ models such that $\sum_{k=1}^K |S_k(\bsgamma)\setminus S_k^*| = l_1$ and $\sum_{k=1}^K |S_k^* \setminus S_k(\bsgamma)| = l_2$.

Combining \eqref{eq:pos-ratio-overfit} and \eqref{eq:pos-ratio-underfit}, we obtain that 
\[\post(\bsgamma^* | \fulldata ) \geq 1 - c_1p^{-1},\]
with probability at least $1 - c_2p^{-c_3}$, where $c_1 > 0$ is some universal constant. 

Finally, we prove \eqref{eq:omega_claim} via induction. When $\sum_{k=1}^K |S_k^* \setminus S_k(\bsgamma)| = l = 1$, $\bsgamma$ misses one influential covariate.  
Assume that $\bsgamma$ misses the $i$-th covariate  in one data set and $\widetilde m_i = k_0 \geq m_i^* \geq 1$. Then,
\[\sum_{k = 1}^K \omega_k (\tilde l_k - l_k) = \omega_{k_0} - \omega_{k_0 - 1},\]
where we define $\omega_0 = 0$. 
Since
\[\sum_{j \in [p]} |\bsgamma_j^* \setminus \bsgamma_j| (\omega_{m_j^*} / m_j^*) = \omega_{m_i^*} / m_i^*,\] 
it follows from~\eqref{eq:kappa-condition} that 
\[\omega_{k_0} - \omega_{k_0 - 1} - \omega_{m_i^*} / m_i^* \left\{\begin{array}{ll}
   = 0,  & \text{ if }k_0 = 1, \\
    < \frac{k_0 \omega_{k_0 - 1}}{k_0 - 1} - \omega_{k_0 - 1} - \frac{\omega_{k_0 - 1}}{k_0 - 1} = 0, & \text{ if }k_0 > 1, 
\end{array}\right.\]
which completes the proof for $l = 1$. 

Assume that the claim holds for $l = l_0$ and now we prove it also holds for $l = l_0 + 1$. 
Clearly, there exists $\check \bsgamma$ such that 
$S_k(\bsgamma) \subseteq S_k(\check\bsgamma) \subseteq \widetilde{S}_k(\bsgamma)$ for every $k \in [K]$ and 
$\sum_{k=1}^K |S_k^* \setminus S_k(\check\bsgamma)| =  1$.  
Observe that for any $j \in [p]$, $ \bsgamma_j^* \setminus \bsgamma_j$ is the disjoint union of $\bsgamma_j^* \setminus \check \bsgamma_j$ and $  \check \bsgamma_j \setminus \bsgamma_j $. 
Letting $\check{l}_k = a_k(\check \bsgamma)$, we find that  
\begin{align*}
& \sum_{k = 1}^K  \omega_k (\tilde l_k - l_k) -  \sum_{j \in [p]} |\bsgamma_j^* \setminus \bsgamma_j| (\omega_{m_j^*} / m_j^*) \\
=& \left(\sum_{k = 1}^K \omega_k (\tilde l_k - \check l_k) - \sum_{j\in[p]}\lvert\bsgamma_j^* \setminus \check \bsgamma_j\lvert\ (\omega_{m_j^*} / m_j^*)\right) + \left(\sum_{k = 1}^K \omega_k (\check l_k - l_k) - \sum_{j\in[p]} \lvert\check \bsgamma_j \setminus \bsgamma_j\lvert(\omega_{m_j^*} / m_j^*)\right) 
\end{align*}
where the first term is non-positive since it corresponds to the case $l = 1$, and the second term is non-positive due to the induction assumption.  
This proves~\eqref{eq:omega_claim}. 
\end{proof}

\newpage
\section{Fitting muSuSiE}\label{appx:fit-muSuSiE}
We first briefly review the notation used in the main text for muSuSiE. 
\begin{longtable}{|p{0.15\textwidth} p{0.84\textwidth}|}
\hline
Notation & Definition\\
\hline 
$\bsbeta^{(k, l)}$ &  $l$-th single-effect regression coefficient vector for the $k$-th data set \\
$\bsbeta^{(k)}$ & $\sum_{l = 1}^L \bsbeta^{(k, l)}$, i.e., aggregated regression coefficient vector for the $k$-th data set \\
$\bsgamma^{(l)}$ & the set-valued vector such that $\gamma_j^{(l)} = I \subseteq [K]$ means that the $j$-th covariate is activated in the data sets indexed by  $I$ in the $l$-th single effect\\ 
$\dset$ & some probability distribution on $2^{[K]} \setminus \emptyset$\\
$\sphi_l$ &  indicator variable; $\sphi_l = 0$ means that the $l$-th single effect is not activated \\
$\sel_l$ & the covariate selected to be activated in the $l$-th single effect \\ 
$I_l$ &  the index set of data sets in which the $\sel_l$-th covariate is to be activated \\ 
$\mathrm{Unif}([p])$ &  uniform distribution on $[p]$\\ 
$\pi_{\sphi}$ &  hyperparameter for the prior distribution on $\sphi_l$ \\ 
$(\pi_k)_{k=1}^K$ &  hyperparameter for the prior distribution on $I_l$ \\  
\hline
\caption{Notation  for muSuSiE.} \label{tab:notation-muSuSiE}
\end{longtable} 

Recall that the prior distribution we put on $\{ \bsgamma^{(l)}\colon  l \in [L] \}$  encodes the following procedure for selecting and activating covariates:  
for each $l \in [L]$, we first draw  $\sphi_l \sim \mathrm{Bernoulli}(\pi_{\sphi})$, $\sel_l \sim \mathrm{Unif}([p])$ and $I_l \sim \dset$; if $\sphi_l = 1$, we activate the $\sel_l$-th covariate in the data sets indexed by $I_l$ (and we do nothing if $\sphi_l = 0$). The distribution $\dset$ is defined by 
\begin{align*}
    \dset(  I ) = p \, \pi_{|I|}, \quad \forall \, I \in 2^{[K]} \setminus \emptyset, 
\end{align*}
where  $\pi_1, \dots, \pi_K$ are normalized so that  $\dset(2^{[K]} \setminus \emptyset) = 1$. 

\subsection{Iterative Bayesian Stepwise Selection Algorithm}
Consider the muSER model defined in~\eqref{eq:muSER}. 
To find its posterior distribution,  denote the $j$-th column of $\bX^{(k)}$ by $\bX_j^{(k)}$ and define
\[\widehat \beta_j^{(k)} = \left(\left(\bX_j^{(k)}\right)^{\T}\bX_j^{(k)}\right)^{-1}\left(\bX_j^{(k)}\right)^{\T} \by^{(k)},\quad
s_{j(k)}^2 = \frac{\sigma^2}{ (\bX_j^{(k)})^{\T} \bX_j^{(k)}},\quad
z_{j(k)} = \frac{\widehat \beta_j^{(k)}}{s_{j(k)}}.\]
Let the Bayes Factor (BF) for activating covariate $j$ in the $k$-th data set be
\[\BF(j, k) = \frac{ P(\by^{(k)} \mid \bX_j^{(k)}, \sigma^2, \vb, \zeta = 1, u = j, k \in \gamma_j)}{P_0(\by^{(k)} \mid  \sigma^2)} =
\sqrt{\frac{s_{j(k)}^2}{\vb + s_{j(k)}^2}}\exp\left(\frac{z_{j(k)}^2}{2}\times \frac{\vb}{\vb + s_{j(k)}^2}\right),\]
where we define $P_0(\by^{(k)}\mid \sigma^2)$ as the probability of observing $\by^{(k)}$ when the $j$-th covariate is not activated for the $k$-th data set. Then, for any $I\in2^{[K]} \setminus \emptyset$, the BF for activating covariate $j$ in all data sets indexed by $I$ is given by 
\[\BF(j, I) = \prod_{k\in I}\BF(j, k).\] 
It follows that the posterior distribution of $(\bsgamma, 1 - \sphi)$ given $\sigma^2$ and $\tau$ is   a multinomial distribution with
\begin{equation}\label{eq:pos-single}
\prior_{\muSER}( \sphi = 0 \mid (\by^{(k)})_{k \in [K]} ) = \alpha_0, \quad  \prior_{\muSER}(\gamma_j = I\mid (\by^{(k)})_{k \in [K]} ) = \alpha_{j, I},   
\end{equation}
where
\begin{align*}
    \alpha_{j, I} &\propto \pi_{\sphi}\pi_{ |I| } \BF(j, I),\\
    \alpha_0 &\propto  1 - \pi_{\sphi}.
\end{align*}
The posterior distribution of $\beta_j^{(k)}$ given $\sphi = 1$, $u = j$ and $k\in\gamma_j$ is  
\[\beta_j^{(k)}|\fulldata, \sigma^2, \vb \sim \cN(\mu_{j}^{(k)}, \phi_j^{(k)}),\]
where 
\[\phi_j^{(k)} = \frac{1}{1/s_{j(k)}^2 + 1/\vb},\quad \mu_{j}^{(k)} = \frac{\phi_j^{(k)}}{s_{j(k)}^2} \times \widehat \beta_j^{(k)}.\]
The above calculation shows that to obtain the posterior distribution for the muSER model, we only need to calculate $\BF(j, k)$ for each $j \in [p]$ and $k \in [K]$.  

\subsubsection{Estimation of \texorpdfstring{$\tau^{(l)}$}{TEXT}}
Given $\sigma^2$, we use an empirical bayes approach to estimating the hyperparameter  $\tau^{(l)}$. Since at most one covariate is activated in \eqref{eq:muSER}, in total there are $|2^{[K]} \setminus \emptyset| \times p + 1$ possible models, where $1$ indicates the null model. Hence, the likelihood of the variance components $\sigma^2, \tau$ under the muSER model is 
\begin{align*}
&\prod_{k = 1}^K P(\by^{(k)} \mid \bX^{(k)}, \sigma^2, \tau)\\
=&\sum_{I\in2^{[K]} \setminus \emptyset}\sum_{j=1}^p \pi_{\sphi}\pi_{|I|} \left\{\prod_{k\in I}P(\by^{(k)} | \bX_j^{(k)}, \sigma^2, \tau, \sphi = 1, u = j, \gamma_j = I)\right\}\times \\
& \qquad \left\{\prod_{k\notin I}P(\by^{(k)}\mid \bX_j^{(k)}, \sigma^2, \sphi = 1, u = j, \gamma_j = I)\right\}   \\
+& (1 - \pi_{\sphi})  \prod_{k = 1}^K P(\by^{(k)}| \bX_j^{(k)}, \sigma^2, \sphi = 0).  
\end{align*} 
Using the Bayes factors we have defined in the last subsection, we can rewrite the likelihood  by 
\[\prod_{k = 1}^K P(\by^{(k)} \mid \bX^{(k)}, \sigma^2, \tau) = \left\{ \prod_{k = 1}^K P_0(\by^{(k)} \mid \sigma^2)\right\} \left\{ \sum_{I\in2^{[K]} \setminus \emptyset}\sum_{j = 1}^p \pi_{\sphi}\pi_{|I|}\BF(j, I; \sigma^2, \tau) + (1 - \pi_{\sphi})\right\},\]
where we write $\BF(j, I; \sigma^2, \tau)$ to emphasize that the Bayes factors depend on both $\sigma^2$ and $\tau$.
Thus, by removing the terms that do not involve $\tau$, we can define an empirical Bayes estimator of $\tau$ as the value that maximizes the function 
\begin{equation}\label{eq:est-sigma0}
\sum_{I\in2^{[K]} \setminus \emptyset}\sum_{j = 1}^p\pi_{|I|}\BF(j, I; \sigma^2, \tau).
\end{equation}
In our code, we use the \texttt{optimize} function in \texttt{R} to solve this one dimensional optimization problem. To estimate  $\tau^{(l)}$ for each $l = 1, \dots, L$, we only need to replace $\by^{(k)}$ by $\resid^{(k, l)}$ when we calculate $\BF(j, I; \sigma^2, \tau^{(l)})$ in \eqref{eq:est-sigma0}, where 
\begin{equation}\label{eq:def.resid}
    \resid^{(k, l)} = \by^{(k)} - \bX^{(k)}\sum_{l' \neq l}\widehat \bsbeta^{(k, l')}. 
\end{equation}

\subsection{IBSS Algorithm is CAVI}
In this section, we show that our IBSS algorithm is actually the coordinate ascent variational inference (CAVI) algorithm for maximizing the evidence lower bound (ELBO) over a certain variational family for the muSuSiE model. The main idea of the proof is similar to that for the SuSiE model (see Supplement B of~\citet{wang2020simple}). 

We begin with a brief review of variational inference. Denote the parameters that we are interested in as $\bstheta$ and the posterior distribution as $\Pi(\bstheta | \data)$ where $\data$ denotes the data. For any distribution function $p$ and $q$, let
\[\KL(p \| q) = \int p(\bstheta)\log \frac{p(\bstheta)}{q(\bstheta)} \,d\bstheta\]
be the Kullback–Leibler (KL) divergence between $p$ and $q$. Let $\cQ$ be a density family of $\bstheta$. The main idea behind variational Bayes is to find some $q\in\cQ$ to approximate the posterior distribution $\Pi(\bstheta | \data)$ by minimizing the KL divergence $\KL(q \| \Pi(\cdot | \data) )$. That is, we try to solve the following optimization problem
\[q^*  = \arg\min_{q\in\cQ} \KL\left(q  \, \| \, \Pi(\cdot | \data)\right).\]
Although $\KL(q \, \| \, \Pi(\cdot | \data) )$ itself is difficult to evaluate, it can be expressed by using another function which is called evidence lower bound (ELBO) and is much easier to calculate: 
\[ \KL (q \| p) = \log(P(\data)) - \ELBO(q),\]
where 
\begin{equation}\label{eq:ELBO}
\ELBO(q) = \bbE_q[\log P(\bstheta, \data)] - \bbE_q[\log q(\bstheta)], 
\end{equation}
where $P(\bstheta, \data) = P(\bstheta) P(\data | \bstheta)$. Since $P(\data)$ does not depend on $\bstheta$, instead of minimizing the KL divergence, we can aim to find $q\in\cQ$ that maximizes the ELBO.

Notice that our muSuSiE model \eqref{eq:muSuSiE} can be considered as a special case of the following additive effects model:
\begin{align}\label{eq:sum-model}
\begin{split}
    \by^{(k)} &= \sum_{l = 1}^L\bsmu^{(k, l)} + \be, \text{ where }\be \sim \cN(0, \sigma^2\bI),\\
    \bsmu^{(l)} &= \left((\bsmu^{(1, l)})^{\T}, \cdots, (\bsmu^{(k, l)})^{\T}\right)^{\T} \sim g_l,   \text{ independently for }l = 1, \cdots, L,    
\end{split}
\end{align}
where $g_l$ denotes some prior probability distribution. 
The mean-field variational family we propose to consider is the collection of probability distributions of the form
\begin{equation}\label{eq:VB-family}
q(\bsmu^{(1)}, \cdots, \bsmu^{(L)}) = \prod_{l = 1}^L q_l(\bsmu^{(l)}),
\end{equation}
that is, we let the variational family $\cQ$ be the class of distributions on $\bsmu = \bigl(\bsmu^{(1)}, \cdots, \bsmu^{(L)}\bigr)$ that factorize over $\bsmu^{(1)}, \cdots, \bsmu^{(L)}$. Then, the ELBO that we want to optimize becomes
\begin{align}\label{eq:IBSS-ELBO}
\begin{split}
& \ELBO(q; \sigma^2, \fulldata)  \\
=& \bbE_{q}[\log P(\fulldata | \bsmu)] + \bbE_q[ \sum_{l = 1}^L \log g_l(\bsmu^{(l)}] - \bbE_q[\log q(\bsmu)]\\
=& -\frac{\sum_{k = 1}^K n_k}{2}\log(2\pi\sigma^2) - \frac{1}{2\sigma^2}\sum_{k = 1}^K \bbE_q[\|\by^{(k)} - \sum_{l = 1}^L\bsmu^{(k, l)}\|^2] + \sum_{l = 1}^L \bbE_{q_l}\left[\log \frac{g_l(\bsmu^{(l)})}{q_l(\bsmu^{(l)})}\right].    
\end{split}
\end{align}
In the CAVI algorithm \citep{blei2017variational}, each step we only update one $\bsmu^{(l)}$ and fix  $\{\bsmu^{(l')}\}_{l' \neq l}$. 
For $q_l$, its ELBO can be expressed by 
\[\ELBO(q_l; \sigma^2, \fulldata) = C - \frac{1}{2\sigma^2}\sum_{k=1}^K \bbE_{q_l}[\|\resid^{(k, l)} - \bsmu^{(l)}\|^2] +  \bbE_{q_l}\left[\log \frac{g_l(\bsmu^{(l)})}{q_l(\bsmu^{(l)})}\right],\]
where $\resid^{(k, l)}$ is defined in~\eqref{eq:def.resid} and 
$C$ is a constant independent of $q_l$. Because we do not impose any constraint on $q_l$, by the standard result in variational inference \citep{blei2017variational}, the distribution which maximizes $\ELBO(q_l; \sigma^2, \fulldata)$ is
\[q_l^*(\bsmu^{(l)}) = \Pi\left(\bsmu^{(l)} | (\resid^{(k, l)})_{k\in[K]}\right),\]
where $\Pi (\bsmu^{(l)} | (\resid^{(k, l)})_{k\in[K]} )$ is the posterior distribution for the model 
\begin{align}\label{eq:single-model}
\begin{split}
    \resid^{(k, l)} &= \bsmu^{(k, l)} + \be, \text{ where }\be \sim \cN(0, \sigma^2\bI),\\
    \bsmu^{(l)} &\sim g_l.    
\end{split}
\end{align}

Now consider the muSuSiE model given in~\eqref{eq:muSuSiE}. By comparing  \eqref{eq:muSuSiE} with \eqref{eq:sum-model}, we see that we can let $\bsmu^{(k, l)} = \bX^{(k)} \bsbeta^{(k, l)}$ and $g_l$ be as described by model \eqref{eq:muSuSiE}. Let
\[\bsbeta^{(l)} = \left[\bsbeta^{(1, l)}, \cdots, \bsbeta^{(K, l)}\right]^{\T}.\]
Then the variational family we propose becomes
\[q(\bsbeta^{(1)}, \cdots, \bsbeta^{(L)}) = \prod_{l = 1}^L q_l(\bsbeta^{(l)}).\]
Because we do not impose any constraint on $q_l$,  by the CAVI algorithm, we should update $q_l$ by 
\[q_l^*(\bsbeta^{(l)}) = \Pi(\bsbeta^{(l)} |(\resid^{(k,l)})_{k\in[K]}),\]
where $ \Pi(\bsbeta^{(l)} |(\resid^{(k,l)})_{k\in[K]})$ is the posterior distribution for the muSER model defined in~\eqref{eq:muSER} with  $\by^{(k)}$ replaced by $\resid^{(k, l)}$. This is exactly how we update $\bsbeta^{(l)}$ in IBSS algorithm; that is, the IBSS algorithm we propose is a CAVI algorithm for muSuSiE.

\subsubsection{Estimation of \texorpdfstring{$\sigma^2$}{TEXT}}
We can estimate $\sigma^2$ using the value that maximizes the ELBO given in \eqref{eq:IBSS-ELBO}. To numerically calcualte it, we take partial derivative of~\eqref{eq:IBSS-ELBO} with respect to $\sigma^2$ and set it to zero,   which results in
\begin{equation}\label{eq:est-sigma}
    \widehat \sigma^2 = \frac{\bbE_q\left[\sum_{k = 1}^K \|\by^{(k)} - \sum_{l = 1}^L\bsmu^{(k, l)}\|^2\right]}{\sum_{k = 1}^K n_k}.
\end{equation}
This can also be seen as a generalization of Equation (B.10) in \cite{wang2020simple} to the multi-task variable selection problem.

\subsubsection{Stopping Criterion}
We calculate ELBO \eqref{eq:IBSS-ELBO} after updating all $L$ single-effect vectors. If the change in ELBO is less than certain small threshold, we stop the IBSS algorithm; otherwise, we update all $L$ single-effect vectors again. In our numerical experiments, we always let the threshold be $10^{-4}$. 

\clearpage 
\newpage

\section{More Simulation Results for Variable Selection}\label{appx:vs-simu}
\subsection{More Simulation Results for muSuSiE and SuSiE}
The simulation results for $K=2$ and $\sigma^2 = 1$ or $4$ are shown in Table \ref{tab:simu-res-K2}, and those for $K=5$ and $\sigma^2 = 1$ or $4$ are shown in Table \ref{tab:simu-res-K5}. We make two key observations. 

First, as we can see, when the sample size is small ($n = 100$), the multi-task method identifies more activated covariates than the single-task method, resulting in increased sensitivity and precision. When the sample size increases to $500$, the multi-task method improves sensitivity but has a slightly smaller precision, because the multi-task approach favors activating covariates simultaneously in two data sets, which can give rise to false positives when some covariate is activated in only one data set but has a very large signal size.   
 
Second, when all the other parameters are fixed, the multi-task method on five data sets outperforms the multi-task method on two data sets, with the former significantly improving both sensitivity and precision. This is evident when $n$ is small. 
When $n$ is large, the advantage of the multi-task method on five data sets over on two data sets is  still noticeable (especially when $\sigma^2 = 4$) but less significant.

\begin{figure}[h!]
    \centering
    \includegraphics[width = 0.8\textwidth]{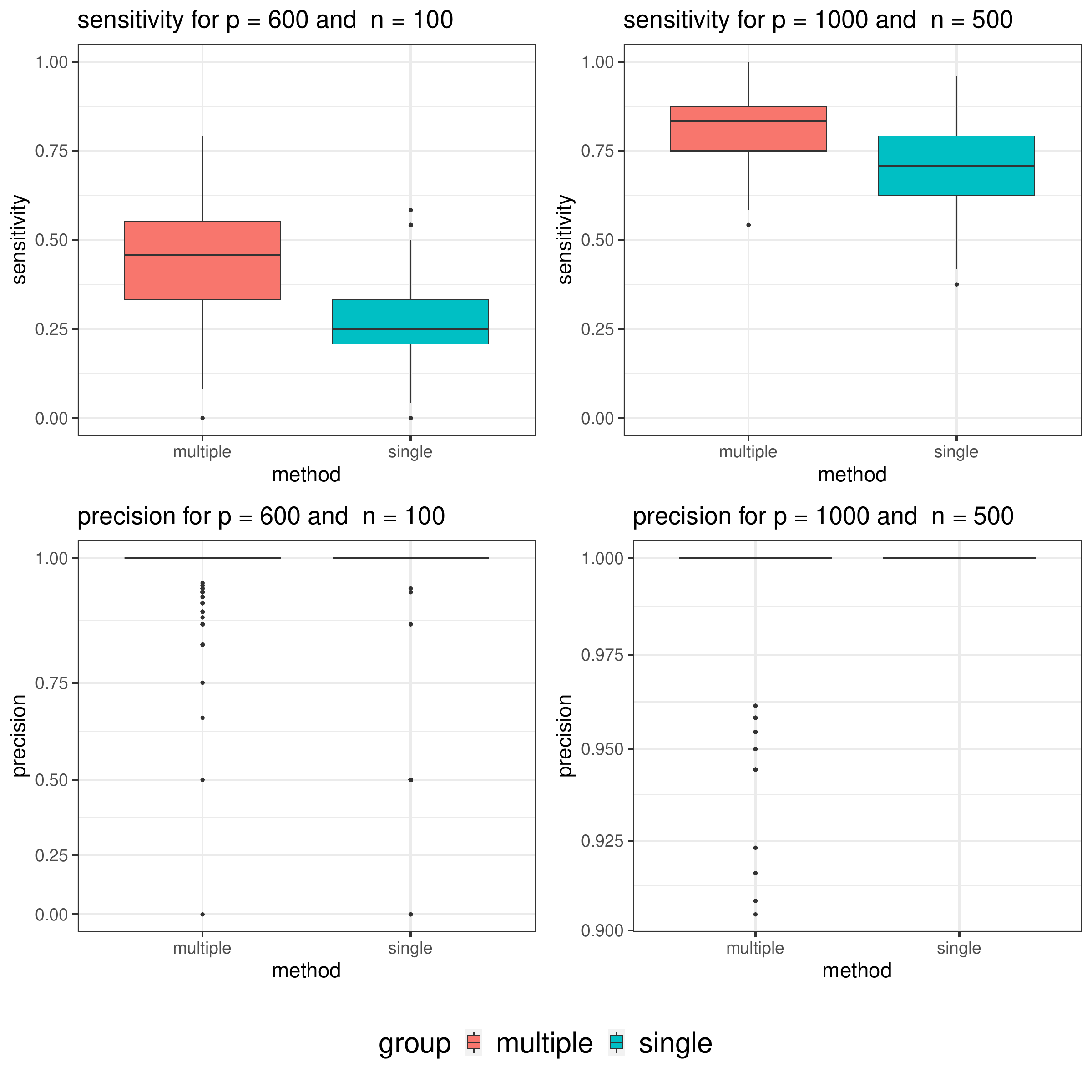}
    \caption{Sensitivity and precision for the simulation study with $K=2$. Each box shows the distribution of sensitivity or precision among  $500$ replicates.
    }\label{fig:K2sigma1}
\end{figure}

\begin{table}
\centering
\begin{tabular}{|cccc|cc|cc|}
\hline
$p$ & $n$ & $s_1^*$ & $s_2^*$  & sens\_mu & sens\_si & prec\_mu & prec\_si\\
\hline
600 & 100 & 10 & 2  & 0.5976& 0.2551& 0.9907& 0.9328\\ 
600 & 100 & 10 & 5  & 0.495& 0.2007& 0.9795& 0.9269\\
\hline
1000 & 500 & 10 & 2  & 0.8181& 0.7062& 0.9936& 0.9999\\ 
1000 & 500 & 10 & 5  & 0.7921& 0.7025& 0.9887& 0.9998\\ 
1000 & 500 & 25 & 2  & 0.8261& 0.6927& 0.9974& 0.9999\\ 
1000 & 500 & 25 & 5  & 0.8101& 0.6916& 0.9938& 0.9999\\ 
\hline
\end{tabular}
\caption{Simulation results for five data sets ($K=5$) with $\sigma = 1$. For each setting, the result is averaged over $500$ replicates.}
\label{tab:simu-res-K5-sigma1}
\end{table}  
 
\begin{table}
\centering
\begin{tabular}{|ccccc|cc|cc|}
\hline
$p$ & $n$ & $s_1^*$ & $s_2^*$ &$\sigma^2$ & sens\_mu & sens\_si & prec\_mu & prec\_si\\
\hline
600 & 100 & 10 & 2 & 1 & 0.4526& 0.2632& 0.9884& 0.9365\\ 
600 & 100 & 10 & 5 & 1 & 0.3456& 0.2045& 0.9747& 0.9258\\ 
600 & 100 & 25 & 2 & 1 & 0.1259& 0.0656& 0.9408& 0.7608\\
600 & 100 & 25 & 5 & 1 & 0.089& 0.0499& 0.8976& 0.7229\\
\hline
600 & 100 & 10 & 2 & 4 & 0.1233& 0.0656& 0.7694& 0.5547\\ 
600 & 100 & 10 & 5 & 4 & 0.0931& 0.0521& 0.7576& 0.5482\\ 
600 & 100 & 25 & 2 & 4 & 0.0499& 0.024& 0.7389& 0.4605\\ 
600 & 100 & 25 & 5 & 4 & 0.0364& 0.0187& 0.6643& 0.4184\\ 
\hline
1000 & 500 & 10 & 2 & 1 & 0.8121& 0.7063& 0.9962& 1\\ 
1000 & 500 & 10 & 5 & 1 & 0.7905& 0.7011& 0.9928& 0.9996\\ 
1000 & 500 & 25 & 2 & 1 & 0.8191& 0.696& 0.9985& 1\\ 
1000 & 500 & 25 & 5 & 1 & 0.804& 0.6949& 0.9964& 0.9999\\ 
\hline
1000 & 500 & 10 & 2 & 4 & 0.613& 0.4655& 0.9945& 0.9987\\ 
1000 & 500 & 10 & 5 & 4 & 0.5735& 0.4549& 0.9901& 0.9997\\ 
1000 & 500 & 25 & 2 & 4 & 0.6077& 0.4389& 0.9978& 0.9999\\ 
1000 & 500 & 25 & 5 & 4 & 0.577& 0.4332& 0.9951& 0.9997\\ 
\hline
\end{tabular}
\caption{Simulation results for two data sets ($K=2$). For each setting, the result is averaged over $500$ replicates.}
\label{tab:simu-res-K2}
\end{table}

\begin{table}
\centering
\begin{tabular}{|ccccc|cc|cc|}
\hline
$p$ & $n$ & $s_1^*$ & $s_2^*$ &$\sigma^2$ & sens\_mu & sens\_si & prec\_mu & prec\_si\\
\hline
600 & 100 & 10 & 2 & 1 & 0.5976& 0.2551& 0.9907& 0.9328\\ 
600 & 100 & 10 & 5 & 1 & 0.495& 0.2007& 0.9795& 0.9269\\
600 & 100 & 25 & 2 & 1 & 0.3344& 0.066& 0.9877& 0.7662\\ 
600 & 100 & 25 & 5 & 1 & 0.1635& 0.0501& 0.9655& 0.7161\\
\hline
600 & 100 & 10 & 2 & 4 & 0.2263& 0.0657& 0.9408& 0.5503\\ 
600 & 100 & 10 & 5 & 4 & 0.1495& 0.0511& 0.8889& 0.539\\ 
600 & 100 & 25 & 2 & 4 & 0.0859& 0.0241& 0.8875& 0.4687\\ 
600 & 100 & 25 & 5 & 4 & 0.0611& 0.02037& 0.8263& 0.4428\\ 
\hline
1000 & 500 & 10 & 2 & 1 & 0.8181& 0.7062& 0.9936& 0.9999\\ 
1000 & 500 & 10 & 5 & 1 & 0.7921& 0.7025& 0.9887& 0.9998\\ 
1000 & 500 & 25 & 2 & 1 & 0.8261& 0.6927& 0.9974& 0.9999\\ 
1000 & 500 & 25 & 5 & 1 & 0.8101& 0.6916& 0.9938& 0.9999\\ 
\hline
1000 & 500 & 10 & 2 & 4 & 0.6641& 0.4593& 0.992& 0.9993\\ 
1000 & 500 & 10 & 5 & 4 & 0.6167& 0.454& 0.9865& 0.9996\\ 
1000 & 500 & 25 & 2 & 4 & 0.6776& 0.4362& 0.9967& 0.9998\\ 
1000 & 500 & 25 & 5 & 4 & 0.6503& 0.4301& 0.9935& 0.9998\\ 
\hline
\end{tabular}
\caption{Simulation results for five data sets ($K=5$). For each setting, the result is averaged over $500$ replicates.}
\label{tab:simu-res-K5}
\end{table}

\subsection{Computation Time for muSuSiE and SuSiE}\label{app:time}
\begin{table}
    \centering
    \begin{tabular}{|ccccc|ccc|ccc|}
    \hline
 $p$ & $n$ & $s_1^*$ & $s_2^*$ &$\sigma^2$ & $K$ & time\_mu & time\_si & $K$ & time\_mu & time\_si  \\
 \hline
600 & 100 & 10 & 2 & 1 & 2 & 1.75 &1.83 & 5 & 19.78& 4.42 \\
600 & 100 & 10 & 5 & 1 & 2 & 3.09 &2.56 & 5 & 33.2& 6.22 \\
600 & 100 & 25 & 2 & 1 & 2 & 5.18 &4.9  & 5 & 17.81& 12.18 \\
600 & 100 & 25 & 5 & 1 & 2 & 6.46 &5.41 & 5 & 126.6& 13.37 \\
\hline
600 & 100 & 10 & 2 & 4 & 2 & 1.43 &1.33  &5& 22.3& 3.13 \\ 
600 & 100 & 10 & 5 & 4 & 2 & 2.33 &1.77  &5& 35.97& 4.2 \\
600 & 100 & 25 & 2 & 4 & 2 & 4.33 &3.65  &5& 38.31& 8.98 \\
600 & 100 & 25 & 5 & 4 & 2 & 5.34 &4.08  &5& 526.8& 10.13 \\
\hline
1000 & 500 & 10 & 2 & 1 & 2 & 2.57 &2.45  &5& 15.57& 5.73 \\
1000 & 500 & 10 & 5 & 1 & 2 & 3.97 &3.28  &5& 33.23& 7.96 \\
1000 & 500 & 25 & 2 & 1 & 2 & 6.79 &8.7  &5& 31.74& 21.23 \\
1000 & 500 & 25 & 5 & 1 & 2 & 9.69 &10.62 &5& 284.6& 25.94 \\
\hline
1000 & 500 & 10 & 2 & 4 & 2 & 2.89 &2.63  &5& 18.41& 6.22 \\
1000 & 500 & 10 & 5 & 4 & 2 & 4.79 &3.64  &5& 34.18& 9.04 \\
1000 & 500 & 25 & 2 & 4 & 2 & 9.73 &10.66  &5& 38.11& 25.92 \\
1000 & 500 & 25 & 5 & 4 & 2 & 13.56 &13.25 & 5 & 259.2& 32.24 \\
\hline
\end{tabular}
\caption{Average computation time for muSuSiE and SuSiE measured in seconds. For each setting, the result is averaged over $500$ replicates.}
\label{tab:vs-time}
\end{table}

Table \ref{tab:vs-time} shows the average computation time of the muSuSiE and SuSiE methods for each setting across $500$ replicates. It is evident that the two methods take a similar amount of time when $K = 2$. However, as $K$ increases to $5$, the muSuSiE method takes more time than SuSiE. The SuSiE method's time increases linearly with respect to $K$, while the muSuSiE method's time increases exponentially. 

\subsection{Stability Analysis of muSuSiE}\label{app:stab}
\begin{table}
\renewcommand{\arraystretch}{1.2}
    \centering
    \begin{tabular}{|c|cc|}
    \hline
      prior   &  $p^{-\omega_1}$ & $p^{-\omega_2}$\\
      \hline
      prior 1   & $p^{-1.1} / 2$& $p^{-1.25}$\\
    prior 2   & $p^{-1.1}/2$ & $p^{-1.5}$\\
     prior 3   & $p^{-1.25}$ & $p^{-1.5}$\\
    prior 4   & $p^{-1.25}$ & $p^{-1.75}$\\
    prior 5   & $p^{-1.25}$ & $p^{-2}$\\
    prior 6   & $p^{-1.25} / 2$ & $p^{-1.5}$\\
    \hline
    \end{tabular}
    \caption{Prior hyperparameters for muSuSiE method with $K=2$ used in Figure~\ref{fig:stable}.}
    \label{tab:prior_vs}
\end{table}

\begin{figure}
    \centering
    \includegraphics[width=0.7\linewidth]{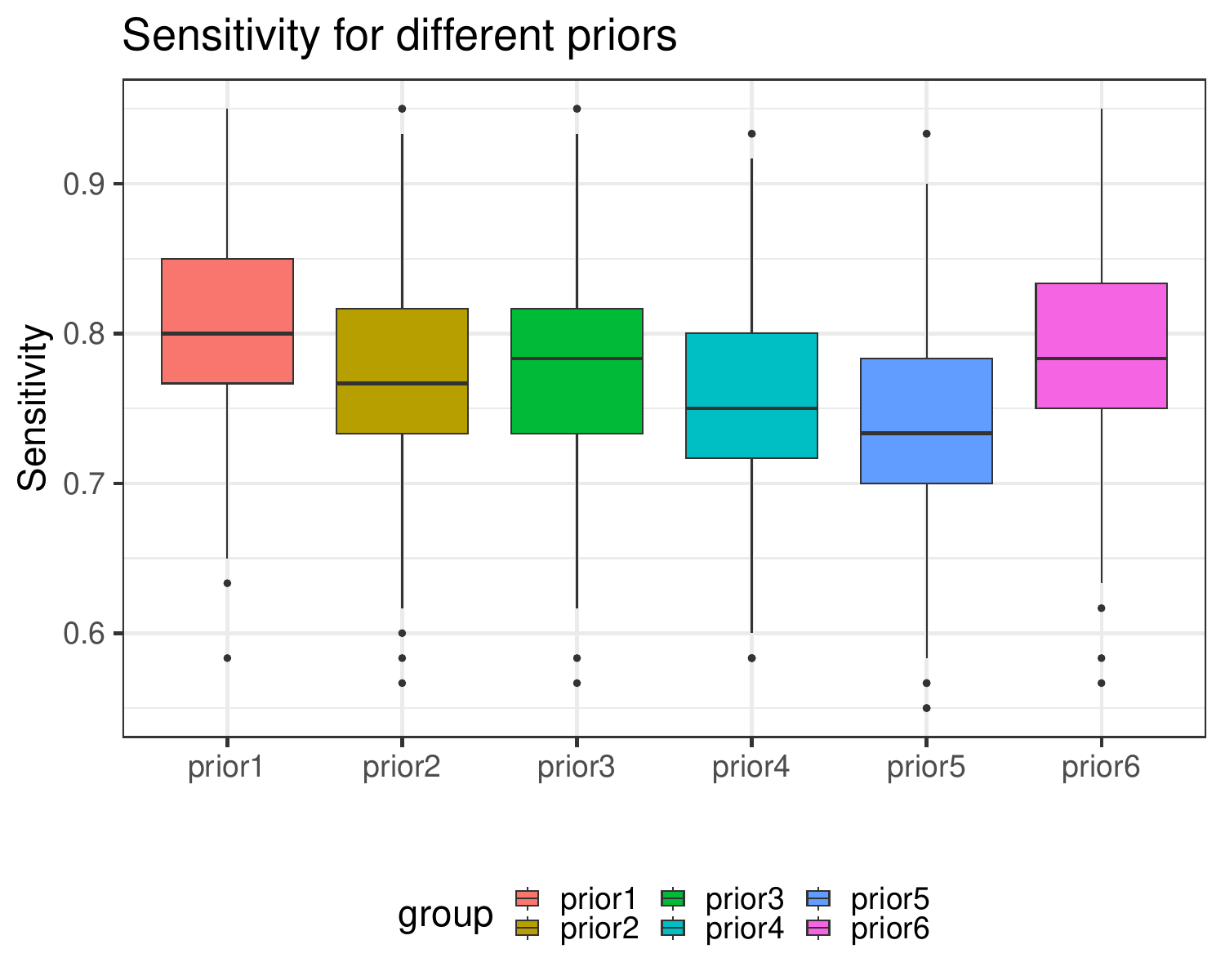} \\
    \includegraphics[width=0.7\linewidth]{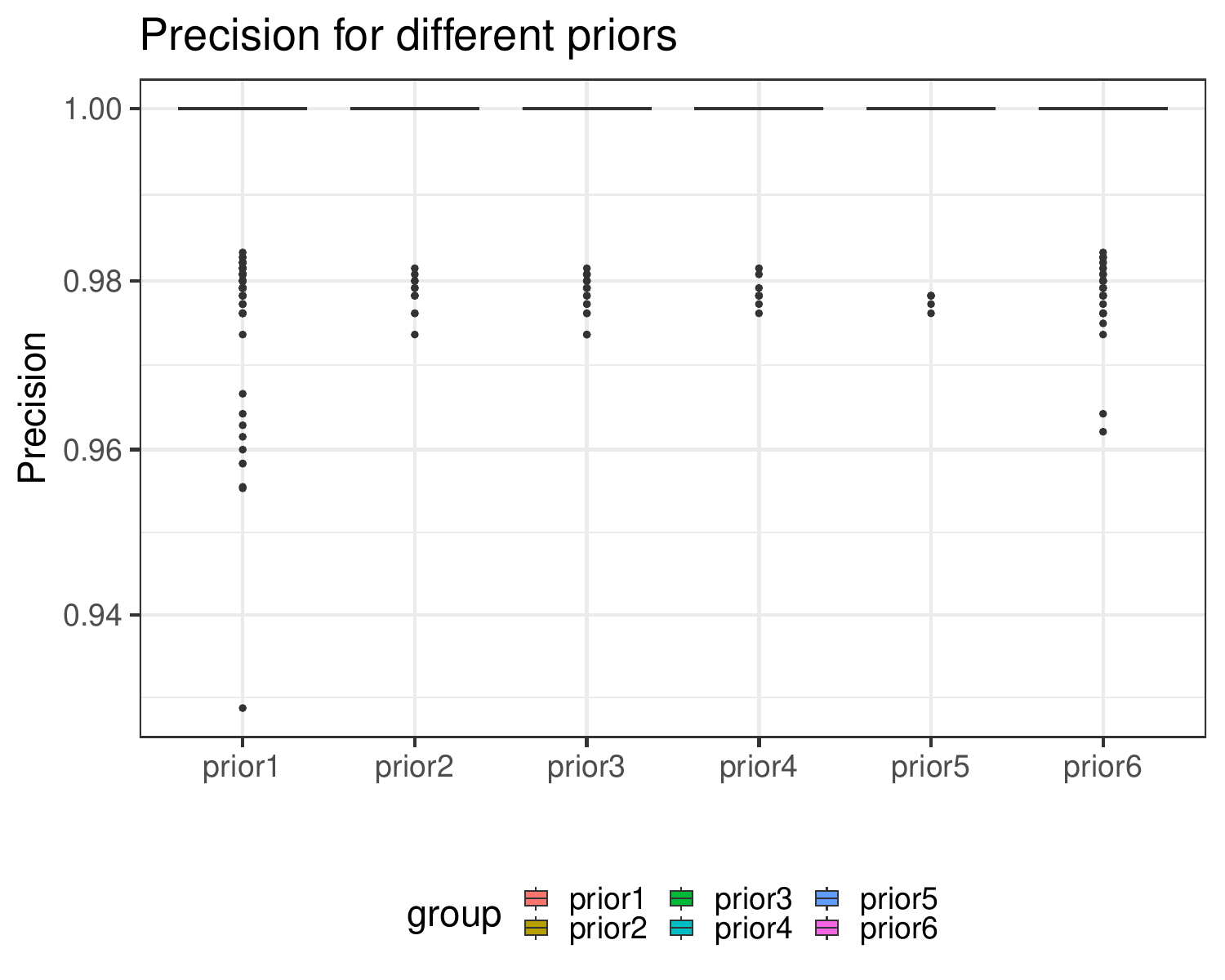} \\
    \caption{Simulation results for muSuSiE method with $K = 2$, $n = 500$, $p = 1000$, $\sigma^2 = 1$, $s_1^* = 25$ and $s_2^* = 5$. For each setting, the result is averaged over $500$ replicates.}
    \label{fig:stable} 
\end{figure} 

Consider the setting with $K = 2$, $n = 500$, $p = 1000$, $\sigma^2 = 1$, $s_1^* = 25$, and $s_2^* = 5$.
We evaluate the performance of our method with six different priors. Table \ref{tab:prior_vs} enumerates the six priors, and Figure \ref{fig:stable} shows the sensitivity and precision of the different priors across 500 replicates. We observe that the performance of muSuSiE is stable with respect to the choice of prior.

\subsection{Simulation Results for Joint MCMC and Separate MCMC}\label{app:mcmc}

Following \citet{yang2016computational, zhou2022rapid}, we consider the following posterior distribution for a single-task variable selection problem, 
\begin{equation}\label{eq:gpost}
    \pi(S) \propto p^{-(\kappa_0 + \kappa_1)|S|}\frac{1}{(1 + g(1 - r_S^2))^{n/2}},
\end{equation}
where $S\in[p]$ represents the set of activated covariates, $r_S^2$ refers to the standard R-squared statistic for the model $S$, and $\kappa_0$, $\kappa_1$ and $g$ are hyperparameters. 
In this simulation, we set $\kappa_0 = \kappa_1 = 1$ and $g = p^{2\kappa_1} - 1$. We run the Metropolis-Hasting (MH)  algorithm for each dataset separately in the ``separate MCMC'' method, only adding and deleting one covariate in each iteration. In each simulation, we run $5\times 10^4$ iterations, with the first $10^4$ samples for burn-in. 

For the ``joint MCMC'' method, we consider the following joint posterior distribution which is obtained by modifying the prior in~\eqref{eq:gpost}: 
\[\pi(\bsgamma) \propto \prod_{k=1}^K p^{-\omega_k a_k(\bsgamma)} \frac{1}{(1 + g(1 - r_{S_k(\bsgamma)}^2))^{n/2}}.\] 
In each iteration, we propose the next state by uniformly sampling one covariate $j\in[p]$ and a set $I$ from $2^{[K]} \setminus \emptyset$ and then flipping the $j$-th covariate's activation status in the data sets indexed by $I$. For $K = 2$, we set $\omega_1 = 2$ and $\omega_2 = 2.25$. 
For $K = 5$, we use $\omega_1 = 2$, $\omega_2 = 2.25$, $\omega_3 = 2.5$, $\omega_4 = 2.75$, and $\omega_5 = 3$. In each simulation, we run $5\times 10^4$ iterations, with the first $10^4$ samples for burn-in.

Tables \ref{tab:simu-mcmc-K2} and \ref{tab:simu-mcmc-K5} present the results of the two MCMC methods. 
``sens\_sep'' and ``prec\_sep'' represent the sensitivity and precision of the separate Metropolis-Hastings (MH) method, respectively; ``sens\_joint'' and ``prec\_joint'' denote the sensitivity and precision of the joint MH method, respectively. We observe that the joint method exhibits greater sensitivity in comparison to the separate method. When the sample size is small, the joint method also has higher precision.  
This observation aligns with analogous findings from the comparison between muSuSiE and SuSiE.  
When $K=2$, the joint MCMC approach almost always has lower sensitivity and precision than muSuSiE, except in the setting with $n=500, \sigma^2 = 1$. 
When $K = 5$, the performance of joint MCMC is comparable to that of muSuSiE: joint MCMC tends to have slightly higher sensitivity but lower precision. 
However, the two MCMC algorithms is considerably more time-consuming than the variational methods, as shown in Table~\ref{tab:vs-time-mcmc}. 

\begin{table}[h!]
\centering
\begin{tabular}{|ccccc|cc|cc|}
\hline
$p$ & $n$ & $s_1^*$ & $s_2^*$ &$\sigma^2$ & sens\_sep & sens\_joint & prec\_sep & prec\_joint\\
\hline
600 & 100 & 10 & 2 & 1 & 0.1687 &0.3933& 0.81& 0.9743\\ 
600 & 100 & 10 & 5 & 1 & 0.1111& 0.2619& 0.7447& 0.9247 \\
600 & 100 & 25 & 2 & 1 & 0.024& 0.075& 0.4258& 0.7789\\ 
600 & 100 & 25 & 5 & 1 & 0.0169& 0.0461& 0.357& 0.6419 \\
\hline
600 & 100 & 10 & 2 & 4 & 0.0324& 0.0718& 0.327& 0.567\\ 
600 & 100 & 10 & 5 & 4 & 0.0248& 0.0505& 0.305& 0.5242 \\
600 & 100 & 25 & 2 & 4 & 0.0097& 0.0228& 0.215& 0.41 \\
600 & 100 & 25 & 5 & 4 & 0.0066& 0.0147& 0.181& 0.3413 \\
\hline
1000 & 500 & 10 & 2 & 1 & 0.6872& 0.8148& 1& 0.9869 \\
1000 & 500 & 10 & 5 & 1 & 0.6765& 0.7875& 1& 0.9738 \\
1000 & 500 & 25 & 2 & 1 & 0.6659& 0.8304& 1& 0.9945 \\
1000 & 500 & 25 & 5 & 1 & 0.6637& 0.8114& 1& 0.9855 \\
\hline
1000 & 500 & 10 & 2 & 4 & 0.4148& 0.5916& 0.995& 0.9888\\ 
1000 & 500 & 10 & 5 & 4 & 0.4001& 0.5391& 1& 0.9752 \\
1000 & 500 & 25 & 2 & 4 & 0.359& 0.5808& 1& 0.9958 \\
1000 & 500 & 25 & 5 & 4 & 0.3491& 0.5567& 1& 0.9893\\
\hline
\end{tabular}
\caption{Simulation results for standard MCMC method for two data sets ($K=2$). For each setting, the result is averaged over $500$ replicates.}
\label{tab:simu-mcmc-K2}
\end{table}

\begin{table}[h!]
\centering
\begin{tabular}{|ccccc|cc|cc|}
\hline
$p$ & $n$ & $s_1^*$ & $s_2^*$ &$\sigma^2$ & sens\_sep & sens\_joint & prec\_sep & prec\_joint\\
\hline
600 & 100 & 10 & 2 & 1 & 0.167& 0.6832& 0.812& 0.9475 \\
600 & 100 & 10 & 5 & 1 & 0.1098& 0.5613& 0.7385& 0.9049 \\
600 & 100 & 25 & 2 & 1 & 0.0242& 0.5713& 0.4403& 0.9797 \\
600 & 100 & 25 & 5 & 1 & 0.018& 0.3401& 0.3776& 0.9657 \\
\hline
600 & 100 & 10 & 2 & 4 & 0.0324& 0.2788& 0.3224& 0.9441 \\
600 & 100 & 10 & 5 & 4 & 0.0246& 0.1861& 0.3092& 0.8979\\ 
600 & 100 & 25 & 2 & 4 & 0.0091& 0.1285& 0.214& 0.9316\\ 
600 & 100 & 25 & 5 & 4 & 0.0071& 0.0902& 0.1888& 0.8865 \\
\hline
1000 & 500 & 10 & 2 & 1 & 0.6847& 0.8311& 1& 0.9237\\ 
1000 & 500 & 10 & 5 & 1 & 0.6802& 0.8016& 1& 0.8537\\ 
1000 & 500 & 25 & 2 & 1 & 0.6636& 0.8413& 1& 0.9645\\ 
1000 & 500 & 25 & 5 & 1 & 0.6599& 0.8233& 1& 0.921\\
\hline
1000 & 500 & 10 & 2 & 4 & 0.4102& 0.7161& 0.9972& 0.9403\\ 
1000 & 500 & 10 & 5 & 4 & 0.3993& 0.6565& 0.9996& 0.8851\\ 
1000 & 500 & 25 & 2 & 4 & 0.3584& 0.7424& 0.9996& 0.9741\\ 
1000 & 500 & 25 & 5 & 4 & 0.3476& 0.709& 1& 0.9416\\ 
\hline
\end{tabular}
\caption{Simulation results for standard MCMC method for five data sets ($K=5$). For each setting, the result is averaged over $500$ replicates.}
\label{tab:simu-mcmc-K5}
\end{table}

\begin{table}
    \centering
\begin{tabular}{|ccccc|ccc|ccc|}
\hline
 $p$ & $n$ & $s_1^*$ & $s_2^*$ &$\sigma^2$ & $K$ & time\_joint & time\_sep & $K$ & time\_joint & time\_sep  \\
 \hline
600 & 100 & 10 & 2 & 1 & 2 & 9.4 & 6.3 & 5 & 23.2 & 8.6\\ 
600 & 100 & 10 & 5 & 1 & 2 & 9.3 & 6.3 & 5 & 22.8 & 8.6\\ 
600 & 100 & 25 & 2 & 1 & 2 & 8.9 & 6.3 & 5 & 21.9 & 8.6\\
600 & 100 & 25 & 5 & 1 & 2 & 8.8 & 6.2 & 5 & 21.7 & 8.6\\
\hline
600 & 100 & 10 & 2 & 4 & 2 & 8.7 & 6.2 & 5 & 21.6 & 8.6\\ 
600 & 100 & 10 & 5 & 4 & 2 & 8.7 & 6.2 & 5 & 21.5 & 8.6 \\
600 & 100 & 25 & 2 & 4 & 2 & 8.7 & 6.2 & 5 & 21.4 & 8.6\\ 
600 & 100 & 25 & 5 & 4 & 2 & 8.6 & 6.2 & 5 & 21.3 & 8.5\\
\hline
1000 & 500 & 10 & 2 & 1 & 2 & 15.8 & 10.5 & 5 & 39.4 & 14.3\\
1000 & 500 & 10 & 5 & 1 & 2 & 17.1 & 10.5 & 5 & 43.2 & 14.4\\
1000 & 500 & 25 & 2 & 1 & 2 & 22.9 & 10.6 & 5 & 43.7 & 14.6\\
1000 & 500 & 25 & 5 & 1 & 2 & 24.4 & 10.7 & 5 & 20.0 & 14.7\\
\hline
1000 & 500 & 10 & 2 & 4 & 2 & 13.8 & 10.4 & 5 & 34.2 & 14.3\\
1000 & 500 & 10 & 5 & 4 & 2 & 14.5 & 10.5 & 5 & 35.9 & 14.3\\
1000 & 500 & 25 & 2 & 4 & 2 & 16.8 & 10.6 & 5 & 42.3 & 14.5\\ 
1000 & 500 & 25 & 5 & 4 & 2 & 17.5 & 10.6 & 5 & 43.5 & 14.5\\ 
\hline
\end{tabular}
\caption{Average computation time for separate MCMC and joint MCMC measured in minutes. For each setting, the result is averaged over $500$ replicates.}
\label{tab:vs-time-mcmc}
\end{table}

\subsection{Simulation Results for LASSO}\label{app:lasso}

Tables \ref{tab:simu-res-K2-l} and \ref{tab:simu-res-K5-l} show the results obtained using the LASSO method. In comparison with the Bayesian variable selection method, the LASSO method displays higher sensitivity, particularly when the sample size is small ($n = 100$). However, the precision of the LASSO method is significantly lower than that of the Bayesian variable selection method. This indicates that the LASSO method tends to select many non-activated covariates. 

\begin{table}
    \centering
    \begin{tabular}{|ccccc|cc|}
    \hline
 $p$ & $n$ & $s_1^*$ & $s_2^*$ &$\sigma^2$ & sens & prec\\
 \hline
600 & 100 & 10 & 2 & 1 & 0.6608 &0.2227 \\
600 & 100 & 10 & 5 & 1 & 0.6443 &0.2256 \\
600 & 100 & 25 & 2 & 1 & 0.556 &0.2597 \\
600 & 100 & 25 & 5 & 1 & 0.5274 &0.2678 \\
\hline
600 & 100 & 10 & 2 & 4 & 0.397 &0.2661 \\
600 & 100 & 10 & 5 & 4 & 0.3817 &0.2584 \\
600 & 100 & 25 & 2 & 4 & 0.3494 &0.285 \\
600 & 100 & 25 & 5 & 4 & 0.3262 &0.2989 \\
\hline
1000 & 500 & 10 & 2 & 1 & 0.8683 &0.2077 \\
1000 & 500 & 10 & 5 & 1 & 0.8733 &0.2105 \\
1000 & 500 & 25 & 2 & 1 & 0.8859 &0.2264 \\
1000 & 500 & 25 & 5 & 1 & 0.8825 &0.2277 \\
\hline
1000 & 500 & 10 & 2 & 4 & 0.7414 &0.2142 \\
1000 & 500 & 10 & 5 & 4 & 0.7433 &0.2199 \\
1000 & 500 & 25 & 2 & 4 & 0.7687 &0.2353 \\
1000 & 500 & 25 & 5 & 4 & 0.7678 &0.2324 \\
\hline
\end{tabular}
\caption{Lasso results for two data sets ($K=2$). For each setting, the result is averaged over $500$ replicates.}
\label{tab:simu-res-K2-l}
\end{table}

\begin{table}
    \centering
    \begin{tabular}{|ccccc|cc|}
    \hline
 $p$ & $n$ & $s_1^*$ & $s_2^*$ &$\sigma^2$ & sens & prec\\
 \hline
600 & 100 & 10 & 2 & 1 & 0.6587 &0.2185 \\
600 & 100 & 10 & 5 & 1 & 0.6409 &0.226 \\
600 & 100 & 25 & 2 & 1 & 0.5567 &0.2572 \\
600 & 100 & 25 & 5 & 1 & 0.5269 &0.2689 \\
\hline
600 & 100 & 10 & 2 & 4 & 0.3931 &0.2633 \\
600 & 100 & 10 & 5 & 4 & 0.3842 &0.2591 \\
600 & 100 & 25 & 2 & 4 & 0.3502 &0.2833 \\
600 & 100 & 25 & 5 & 4 & 0.331 &0.2969 \\
\hline
1000 & 500 & 10 & 2 & 1 & 0.8709 &0.2069 \\
1000 & 500 & 10 & 5 & 1 & 0.8758 &0.2101 \\
1000 & 500 & 25 & 2 & 1 & 0.8837 &0.2257 \\
1000 & 500 & 25 & 5 & 1 & 0.881 &0.2281 \\
\hline
1000 & 500 & 10 & 2 & 4 & 0.741 &0.2133 \\
1000 & 500 & 10 & 5 & 4 & 0.7474 &0.2174 \\
1000 & 500 & 25 & 2 & 4 & 0.7654 &0.233 \\
1000 & 500 & 25 & 5 & 4 & 0.7651 &0.2327 \\
\hline
\end{tabular}
\caption{Lasso results for five data sets ($K=5$). For each setting, the result is averaged over $500$ replicates.}
\label{tab:simu-res-K5-l}
\end{table}

\clearpage 
\newpage
\section{More Simulation Results for Differential DAG Analysis}\label{appx:dags-simu} 
Recall that we use simulation studies to compare the performance of six methods for joint estimation of multiple DAG models: PC, GES, joint GES, MpenPC, JESC and muSuSiE-DAG.  We use $N_{\text{com}}$ to denote the number of shared edges and $N_{\text{pri}}$ to denote the number of edges unique to each data set. 
In each simulation setting, we report the average number of wrong edges $N_{\text{wrong}}$, the average true positive (TP) rate and the average false positive (FP) rate by ignoring edge directions. In addition, we introduce the fourth measurement metric, the squared Frobenius Norm (F-norm) between the true adjacency matrix and estimated adjacency matrix, which can be calculated as follows. 
For PC, GES and joint GES method, we let $\hat{\bR}^{(k)} \in \{0, 1\}^{p \times p}$ be the estimated adjacency matrix such that $\hat{R}^{(k)}_{ij} = 1$ if the edge $(i, j)$ is in the estimated DAG for the $k$-th data set and $\hat{R}^{(k)}_{ij} = 0$ otherwise. 
For muSuSiE-DAG, we let $\hat{\bR}^{(k)}$ be as defined in~\eqref{eq:def.bR} where each entry is the estimated probability of the edge and thus takes value in $[0, 1]$. Let $( \bR^{(k)} )_{k = 1}^K$ be the true adjacency matrices where $R^{(k)}_{ij} = 1$ if the edge $(i, j)$ is in the $k$-th true DAG model and $R^{(k)}_{ij} = 0$ otherwise.  For each method, the F-norm metric is defined as 
\begin{align*}
    \sum_{i = 1}^{p - 1} \sum_{j = i+1}^p  \left(  \hat{R}^{(k)}_{ij} + \hat{R}^{(k)}_{ji} -   R^{(k)}_{ij} - R^{(k)}_{ji}  \right)^2.   
\end{align*}
For PC, GES and joint GES, this is equivalent to counting the number of wrong edges. But for muSuSiE-DAG, this statistic in general is different from $N_{\rm{wrong}}$. 
Table \ref{tab:simu_sepa_pc} shows the results for the PC method, Table \ref{tab:simu_sepa_ges} for GES and Table \ref{tab:simu_joint_ges} for joint GES.
Note that the values of the tuning parameters are also given in the corresponding tables, including the significance level $\alpha$ used in the conditional independent tests  for the PC method, and $l_0$-penalization parameter $\lambda$ for GES and joint GES. 
For muSuSiE-DAG, the choice of the parameters $\omega_1, \dots, \omega_K$ is detailed in Table \ref{tab:prior_inform_2} (for $K = 2$) and Table \ref{tab:prior_inform_5} (for $K=5$). 
Table \ref{tab:simu_muSuSiEDAG} shows the results for our muSuSiEDAG method. Table \ref{tab:sim-DAG-K5} compares results for all methods when $K = 5$, where for each method we use the optimal tuning parameter.

\begin{table}[h!]
\renewcommand{\arraystretch}{1.2}
    \centering
    \begin{tabular}{|c|cc|}
    \hline
      prior   &  $p^{-\omega_1}$ & $p^{-\omega_2}$\\
      \hline
      prior 1   & $p^{-1.25} / 2$& $p^{-1.5}$\\
    prior 2   & $p^{-1.5}/2$ & $p^{-2}$\\
     prior 3   & $p^{-2}$ & $p^{-2.25}$\\
    prior 4   & $p^{-2}/2$ & $p^{-2.25}$\\
    \hline
    \end{tabular}
    \caption{Prior hyperparameters for muSuSiE-DAG for $K=2$.}
    \label{tab:prior_inform_2}
\end{table}

\begin{table}
\renewcommand{\arraystretch}{1.2}
    \centering
    \begin{tabular}{|c|ccccc|}
    \hline
      prior   & $p^{-\omega_1}$ & $p^{-\omega_2}$ & $p^{-\omega_3}$ & $p^{-\omega_4}$ & $p^{-\omega_5}$ \\
      \hline
       prior 1  & $ p^{-1.5} / 5$ & $p^{-1.75}/10$ & $p^{-2}/10$ & $p^{-2.25}/5$ & $p^{-2.5}$\\
    prior 2  & $p^{-1.75}$&$p^{-2}$&$p^{-2.25}$&$p^{-2.5}$&$p^{-3}$\\
    prior 3  & $p^{-1.75} / 5$ & $p^{-2}/10$ &$p^{-2.25}/10$ & $p^{-2.5}/5$ & $p^{-3}$\\
    prior 4  & $p^{-2}$ & $p^{-2.25}$ &$p^{-2.5}$ &$p^{-2.75}$&$p^{-3}$\\
    \hline
    \end{tabular}
    \caption{Prior hyperparameters for muSuSiE-DAG for $K=5$.}
    \label{tab:prior_inform_5}
\end{table}

As expected, joint methods, joint GES and muSuSiEDAG, have much larger true positive rates and slightly larger false positive rates than the two separate-analysis methods, PC and GES.  
This is because the joint method can identify edges with low signal strength if it is expressed concurrently in all $K$ data sets, which leads to a higher true positive rate, while the joint method may also identify edges with extremely high signal strength in a single data set as expressed simultaneously in more than one data sets, resulting in a higher false positive rate. 
Additionally, when the ratio $N_{\text{com}}/N_{\text{pri}}$ is large,  implying that the majority of edges are shared, the joint method outperforms the other two by a large margin. When the ratio $N_{\text{com}}/N_{\text{pri}}$ is small, GES may even outperform joint GES. In all cases, our muSuSiE-DAG method has the best performance in terms of the metric $N_{\rm{wrong}}$. We also observe that the results may depend on the choice of the prior parameters (though not significantly), which suggests that in reality one may want to tune the parameters to improve the performance of the algorithm.

\begin{table}
    \centering
    \begin{tabular}{|c|ccc|ccc|}
    \hline
       method  &  K & $N_{\text{com}}$ & $N_{\text{pri}}$ & $N_{\text{wrong}}$ & TP & FP \\
       \hline
       PC  & 5& 100 & 20 & 31.844 & 0.7504 & 4e-04  \\ 
       GES & 5 & 100 & 20 & 23.108 & 0.8178 & 3e-04\\
       joint GES & 5 & 100 & 20& 26.828 & 0.8952 & 0.0029 \\ 
        MPenPC & 5 & 100 & 20 & 194.24 & 0.8955 & 0.0372 \\
       JESC & 5 & 100 & 20 & 33.24 & 0.9063 & 0.0045 \\
       muSuSiE-DAG & 5 & 100 & 20 & \textbf{17.064} & 0.9058 & 0.0012\\
       \hline
       PC & 5 & 100 & 50 & 43.904& 0.7044 & 3e-04 \\
       GES & 5 & 100 & 50 & 30.196 & 0.8157 & 5e-4\\
      joint GES & 5 & 100 & 50& 36.464 & 0.8794 & 0.0038 \\ MPenPC & 5 & 100 & 50 & 153.788 & 0.8654 & 0.0275 \\
       JESC & 5 & 100 & 50 & 34.996 & 0.9118 & 0.0045 \\
    muSuSiE-DAG & 5 & 100 & 50 & \textbf{28.572} & 0.8755 & 0.002\\
       \hline
       PC & 5 & 50 & 50 & 25.116 & 0.7309 & 1e-04\\
       GES & 5 & 50 & 50 & 19.288 & 0.8166 & 2e-04\\
      joint GES & 5 & 50 & 50& 32.836 & 0.8492 & 0.0036\\
       MPenPC & 5 & 50 & 50 & 224.876 & 0.9098 & 0.0441 \\
       JESC & 5 & 50 & 50 & 31.496 & 0.909 & 0.0046 \\ 
    muSuSiE-DAG & 5 & 50 & 50 & \textbf{18.58} & 0.8529 & 8e-04\\
       \hline
    \end{tabular}
    \caption{Simulation results for joint estimation of multiple DAG models with $K = 5$.}
    \label{tab:sim-DAG-K5}
\end{table}

\begin{table}
\centering
\begin{tabular}{|cccc|cccc|}
\hline
$K$& $\alpha$ & $N_{\text{com}}$ & $N_{\text{pri}}$& $N_{\text{wrong}}$ & TP & FP &F-norm\\
\hline
2 & 1e-04 & 100 & 20 & 38.52 & 0.68 & 0 & 38.52 \\ 
2 & 5e-04 & 100 & 20 & 34.1 & 0.7183 & 1e-04 & 34.1 \\ 
2 & 0.001 & 100 & 20 & 32.04 & 0.7373 & 1e-04 & 32.04 \\ 
2 & 0.005 & 100 & 20 & \textbf{28.29} & 0.7822 & 4e-04 & 28.29 \\ 
2 & 0.01 & 100 & 20 & 28.55 & 0.8009 & 0.001 & 28.55 \\ 
2 & 0.05 & 100 & 20 & 45.83 & 0.8523 & 0.0058 & 45.83 \\ 
\hline
2 & 1e-04 & 100 & 50 & 54.04 & 0.6404 & 0 & 54.04 \\ 
2 & 5e-04 & 100 & 50 & 47.93 & 0.6823 & 1e-04 & 47.93 \\ 
2 & 0.001 & 100 & 50 & 45.32 & 0.7007 & 1e-04 & 45.32 \\ 
2 & 0.005 & 100 & 50 & \textbf{39.37} & 0.7475 & 3e-04 & 39.37 \\ 
2 & 0.01 & 100 & 50 & 37.93 & 0.7674 & 6e-04 & 37.93 \\ 
2 & 0.05 & 100 & 50 & 44.9 & 0.8225 & 0.0038 & 44.9 \\ 
\hline
2 & 1e-04 & 50 & 50 & 28.13 & 0.7192 & 0 & 28.13 \\ 
2 & 5e-04 & 50 & 50 & 24.59 & 0.7572 & 1e-04 & 24.59 \\ 
2 & 0.001 & 50 & 50 & 23.44 & 0.7723 & 1e-04 & 23.44 \\ 
2 & 0.005 & 50 & 50 & \textbf{21.9} & 0.8121 & 6e-04 & 21.9 \\ 
2 & 0.01 & 50 & 50 & 23.34 & 0.8312 & 0.0013 & 23.34 \\ 
2 & 0.05 & 50 & 50 & 47.78 & 0.8796 & 0.0073 & 47.78 \\ 
\hline
5 & 1e-04 & 100 & 20 & 43.108 & 0.6413 & 0 & 43.108 \\ 
5 & 5e-04 & 100 & 20 & 38.428 & 0.6816 & 0 & 38.428 \\ 
5 & 0.001 & 100 & 20 & 36.264 & 0.701 & 1e-04 & 36.264 \\ 
5 & 0.005 & 100 & 20 & 32.044 & 0.7504 & 4e-04 & 32.044 \\ 
5 & 0.01 & 100 & 20 & \textbf{31.844} & 0.7725 & 9e-04 & 31.844 \\ 
5 & 0.05 & 100 & 20 & 47.924 & 0.8286 & 0.0056 & 47.924 \\ 
\hline
5 & 1e-04 & 100 & 50 & 62.304 & 0.5854 & 0 & 62.304 \\ 
5 & 5e-04 & 100 & 50 & 55.508 & 0.6313 & 0 & 55.508 \\ 
5 & 0.001 & 100 & 50 & 52.528 & 0.6521 & 1e-04 & 52.528 \\ 
5 & 0.005 & 100 & 50 & 45.78 & 0.7044 & 3e-04 & 45.78 \\ 
5 & 0.01 & 100 & 50 & \textbf{43.904} & 0.7276 & 6e-04 & 43.904 \\ 
5 & 0.05 & 100 & 50 & 50.608 & 0.7897 & 0.0039 & 50.608 \\ 
\hline
5 & 1e-04 & 50 & 50 & 33 & 0.6703 & 0 & 33 \\ 
5 & 5e-04 & 50 & 50 & 29.064 & 0.7114 & 0 & 29.064 \\ 
5 & 0.001 & 50 & 50 & 27.424 & 0.7309 & 1e-04 & 27.424 \\ 
5 & 0.005 & 50 & 50 & \textbf{25.116} & 0.779 & 6e-04 & 25.116 \\ 
5 & 0.01 & 50 & 50 & 26.036 & 0.8002 & 0.0012 & 26.036 \\ 
5 & 0.05 & 50 & 50 & 49.848 & 0.8535 & 0.0072 & 49.848 \\ 
\hline
\end{tabular}
\caption{Simulation results for the PC algorithm. $N_{\rm{wrong}}$ is the same as F-norm by definition.}
\label{tab:simu_sepa_pc}
\end{table}

\begin{table}
\centering
\begin{tabular}{|cccc|cccc|}
\hline
$K$& $\lambda$ & $N_{\text{com}}$ & $N_{\text{pri}}$& $N_{\text{wrong}}$ & TP & FP &F-norm\\
\hline
2 & 1 & 100 & 20 & 32.47 & 0.9152 & 0.0046 & 32.47 \\ 
2 & 2 & 100 & 20 & \textbf{19.67} & 0.8482 & 3e-04 & 19.67 \\ 
2 & 3 & 100 & 20 & 26.76 & 0.7837 & 2e-04 & 26.76 \\ 
2 & 4 & 100 & 20 & 33.26 & 0.7269 & 1e-04 & 33.26 \\ 
2 & 5 & 100 & 20 & 39.03 & 0.6777 & 1e-04 & 39.03 \\ 
\hline
2 & 1 & 100 & 50 & 35.26 & 0.9191 & 0.0048 & 35.26 \\ 
2 & 2 & 100 & 50 & \textbf{24.84} & 0.8505 & 5e-04 & 24.84 \\ 
2 & 3 & 100 & 50 & 33.12 & 0.7895 & 3e-04 & 33.12 \\ 
2 & 4 & 100 & 50 & 40.31 & 0.7383 & 2e-04 & 40.31 \\ 
2 & 5 & 100 & 50 & 47.35 & 0.6893 & 2e-04 & 47.35 \\ 
\hline
2 & 1 & 50 & 50 & 29 & 0.9223 & 0.0043 & 29 \\ 
2 & 2 & 50 & 50 & \textbf{15.74} & 0.8514 & 2e-04 & 15.74 \\ 
2 & 3 & 50 & 50 & 21.15 & 0.7925 & 1e-04 & 21.15 \\ 
2 & 4 & 50 & 50 & 26.47 & 0.7374 & 0 & 26.47 \\ 
2 & 5 & 50 & 50 & 31.16 & 0.6904 & 0 & 31.16 \\ 
\hline
5 & 1 & 100 & 20 & 35.548 & 0.8955 & 0.0047 & 35.548 \\ 
5 & 2 & 100 & 20 & \textbf{23.108} & 0.8178 & 3e-04 & 23.108 \\ 
5 & 3 & 100 & 20 & 31.008 & 0.7466 & 1e-04 & 31.008 \\ 
5 & 4 & 100 & 20 & 38.44 & 0.6824 & 1e-04 & 38.44 \\ 
5 & 5 & 100 & 20 & 46.012 & 0.6184 & 0 & 46.012 \\ 
\hline
5 & 1 & 100 & 50 & 40.36 & 0.8953 & 0.0051 & 40.36 \\ 
5 & 2 & 100 & 50 & \textbf{30.196} & 0.8157 & 5e-04 & 30.196 \\ 
5 & 3 & 100 & 50 & 39.664 & 0.744 & 3e-04 & 39.664 \\ 
5 & 4 & 100 & 50 & 48.988 & 0.6791 & 2e-04 & 48.988 \\ 
5 & 5 & 100 & 50 & 58.332 & 0.615 & 1e-04 & 58.332 \\  
\hline
5 & 1 & 50 & 50 & 32.58 & 0.8976 & 0.0046 & 32.58 \\ 
5 & 2 & 50 & 50 & \textbf{19.288} & 0.8166 & 2e-04 & 19.288 \\ 
5 & 3 & 50 & 50 & 25.852 & 0.7453 & 1e-04 & 25.852 \\ 
5 & 4 & 50 & 50 & 32.2 & 0.68 & 0 & 32.2 \\ 
5 & 5 & 50 & 50 & 38.18 & 0.6195 & 0 & 38.18 \\ 
\hline
\end{tabular}
\caption{Simulation results for the GES method. $N_{\rm{wrong}}$ is the same as F-norm by definition.}
\label{tab:simu_sepa_ges}
\end{table}

\begin{table}
\centering
\begin{tabular}{|cccc|cccc|}
\hline
$K$& $\lambda$ & $N_{\text{com}}$ & $N_{\text{pri}}$& $N_{\text{wrong}}$ & TP & FP &F-norm\\
\hline
2 & 1 & 100 & 20 & 38.84 & 0.9172 & 0.0059 & 38.84 \\ 
2 & 2 & 100 & 20 & \textbf{15.4} & 0.9126 & 0.001 & 15.4 \\ 
2 & 3 & 100 & 20 & 16.48 & 0.8957 & 8e-04 & 16.48 \\ 
2 & 4 & 100 & 20 & 18.62 & 0.875 & 7e-04 & 18.62 \\ 
2 & 5 & 100 & 20 & 20.93 & 0.8531 & 7e-04 & 20.93 \\ 
\hline
2 & 1 & 100 & 50 & 46.5 & 0.9179 & 0.007 & 46.5 \\ 
2 & 2 & 100 & 50 & \textbf{24.7} & 0.9003 & 0.002 & 24.7 \\ 
2 & 3 & 100 & 50 & 25.72 & 0.879 & 0.0016 & 25.72 \\ 
2 & 4 & 100 & 50 & 28.65 & 0.8554 & 0.0014 & 28.65 \\ 
2 & 5 & 100 & 50 & 31.98 & 0.8295 & 0.0013 & 31.98 \\ 
\hline
2 & 1 & 50 & 50 & 38.28 & 0.9042 & 0.0059 & 38.28 \\ 
2 & 2 & 50 & 50 & \textbf{22.91} & 0.883 & 0.0023 & 22.91 \\ 
2 & 3 & 50 & 50 & 24.68 & 0.8536 & 0.002 & 24.68 \\ 
2 & 4 & 50 & 50 & 27.78 & 0.8186 & 0.002 & 27.78 \\ 
2 & 5 & 50 & 50 & 30.53 & 0.788 & 0.0019 & 30.53 \\ 
\hline
5 & 1 & 100 & 20 & 74.608 & 0.8935 & 0.0127 & 74.608 \\ 
5 & 2 & 100 & 20 & 26.828 & 0.8952 & 0.0029 & 26.828 \\ 
5 & 3 & 100 & 20 & \textbf{23.332} & 0.8877 & 0.002 & 23.332 \\ 
5 & 4 & 100 & 20 & 24.196 & 0.8773 & 0.0019 & 24.196 \\ 
5 & 5 & 100 & 20 & 24.728 & 0.8698 & 0.0019 & 24.728 \\  
\hline
5 & 1 & 100 & 50 & 90.516 & 0.8964 & 0.0155 & 90.516 \\ 
5 & 2 & 100 & 50 & 36.464 & 0.8794 & 0.0038 & 36.464 \\ 
5 & 3 & 100 & 50 & \textbf{35.704} & 0.8591 & 0.003 & 35.704 \\ 
5 & 4 & 100 & 50 & 38.364 & 0.8399 & 0.003 & 38.364 \\ 
5 & 5 & 100 & 50 & 40.32 & 0.8222 & 0.0028 & 40.32 \\
\hline
5 & 1 & 50 & 50 & 70.392 & 0.8726 & 0.0118 & 70.392 \\ 
5 & 2 & 50 & 50 & \textbf{32.836} & 0.8492 & 0.0036 & 32.836 \\ 
5 & 3 & 50 & 50 & 33.716 & 0.8198 & 0.0032 & 33.716 \\ 
5 & 4 & 50 & 50 & 36.356 & 0.7922 & 0.0032 & 36.356 \\ 
5 & 5 & 50 & 50 & 39.744 & 0.7621 & 0.0033 & 39.744 \\ 
\hline
\end{tabular}
\caption{Simulation results for the joint GES method. $N_{\rm{wrong}}$ is the same as F-norm by definition.}
\label{tab:simu_joint_ges}
\end{table}

\begin{table}
\centering
\begin{tabular}{|cccc|cccc|}
\hline
$K$& prior & $N_{\text{com}}$ & $N_{\text{pri}}$& $N_{\text{wrong}}$ & TP & FP &F-norm\\
\hline
2 & prior 1 & 100 & 20 & 18.63 & 0.9248 & 0.002 & 16.6747 \\ 
2 & prior 2 & 100 & 20 & 14.77 & 0.9052 & 7e-04 & 12.6597 \\ 
2 & prior 3 & 100 & 20 & 13.03 & 0.9081 & 4e-04 & 10.9038 \\ 
2 & prior 4 & 100 & 20 & \textbf{12.91} & 0.9138 & 5e-04 & 11.1051 \\ 
\hline
2 & prior 1 & 100 & 50 & 27.17 & 0.9103 & 0.0028 & 24.693 \\ 
2 & prior 2 & 100 & 50 & 19.84 & 0.8983 & 9e-04 & 17.2605 \\ 
2 & prior 3 & 100 & 50 & 19.56 & 0.8933 & 7e-04 & 17.0741 \\ 
2 & prior 4 & 100 & 50 & \textbf{18.45} & 0.9003 & 7e-04 & 16.45 \\ 
\hline
2 & prior 1 & 50 & 50 & 17.84 & 0.8995 & 0.0016 & 16.547 \\ 
2 & prior 2 & 50 & 50 & 15.03 & 0.8819 & 7e-04 & 13.2731 \\ 
2 & prior 3 & 50 & 50 & 15.4 & 0.8727 & 5e-04 & 13.6451 \\ 
2 & prior 4 & 50 & 50 & \textbf{15.03} & 0.8762 & 5e-04 & 13.4162 \\
\hline
5 & prior 1 & 100 & 20 & 19.684 & 0.9243 & 0.0022 & 16.9472 \\ 
5 & prior 2 & 100 & 20 & 26.64 & 0.8808 & 0.0025 & 23.518 \\ 
5 & prior 3 & 100 & 20 & \textbf{17.064} & 0.9058 & 0.0012 & 14.2844 \\ 
5 & prior 4 & 100 & 20 & 20.636 & 0.8955 & 0.0017 & 17.7576 \\ 
\hline
5 & prior 1 & 100 & 50 & 30.74 & 0.8913 & 0.003 & 27.2652 \\ 
5 & prior 2 & 100 & 50 & 43.292 & 0.8515 & 0.0043 & 39.0037 \\ 
5 & prior 3 & 100 & 50 & \textbf{28.572} & 0.8755 & 0.002 & 25.0959 \\ 
5 & prior 4 & 100 & 50 & 32.144 & 0.8679 & 0.0025 & 28.5295 \\ 
\hline
5 & prior 1 & 50 & 50 & 19.684 & 0.8681 & 0.0013 & 17.643 \\ 
5 & prior 2 & 50 & 50 & 25.832 & 0.8418 & 0.002 & 23.5776 \\ 
5 & prior 3 & 50 & 50 & \textbf{18.58} & 0.8529 & 8e-04 & 16.5723 \\ 
5 & prior 4 & 50 & 50 & 21.304 & 0.85 & 0.0013 & 19.215 \\ 
\hline
\end{tabular}
\caption{Simulation results for the muSuSiE-DAG method.}
\label{tab:simu_muSuSiEDAG}
\end{table}

\begin{table}
    \centering
    \begin{tabular}{|cc|cc|}
    \hline
    $N_{\text{com}}$ &  $N_{\text{pri}}$ & Joint GES & muSuSiE-DAG \\
    \hline
50 & 50 & 0.3268 & 4.1788 \\
100 & 50 & 0.4663 & 5.0471 \\
100 & 20 & 0.3173 & 4.3901 \\
\hline
    \end{tabular}
    \caption{Average computation time for the joint GES and muSuSiE-DAG method for $K = 2$ measured in hours.}
    \label{tab:simu_dag_time}
\end{table}

\subsection{Convergence of MCMC}\label{app:con-mcmc}
The structure learning is by nature computationally very expensive. In order to demonstrate the convergence of our MCMC algorithm, we simulate one instance of $(\cG^{(k)})_{k=1}K$  and $(\bX^{(k)})_{k=1}^K$
with $K = 2$, $n_{\text{com}} = 50$, and $n_{\text{pri}} = 50$. 
We run the algorithm 50 times (for the same data set), each with a maximum of $10^6$ iterations. The log-likelihood with respect to the number of iterations is depicted in Figure \ref{fig:mcmc-covg}. It can be observed that the algorithm converges after approximately $6\times10^5$ iterations, which is relatively large. The simulation study presented in Section~\ref{sec:dags-simu} uses a total of $10^5$ MCMC iterations, which appears to be sufficient for yielding satisfactory results. 

\begin{figure}
    \centering
    \includegraphics[width = 0.8\textwidth]{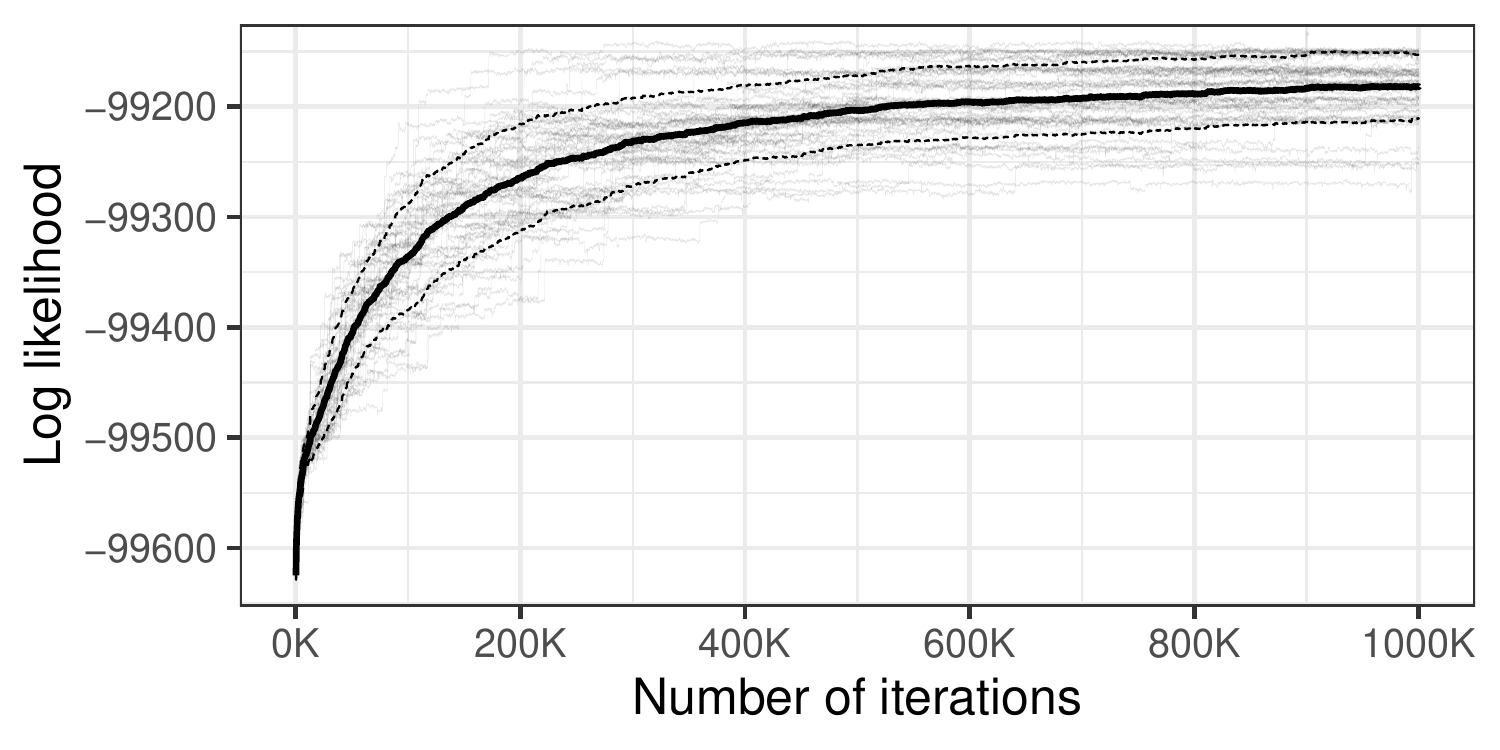}
    \caption{
    Log-likelihood trace plot for muSuSiE-DAG under the setting with $K = 2$, $n_{\text{com}} = 50$ and $n_{\text{pri}} = 50$.     
    Trajectories of all 50 runs are shown individually in gray. 
    The solid line denotes the average over 50 runs, and dashed lines indicate one standard derivation above and below the average.   
    }
    \label{fig:mcmc-covg}
\end{figure}

\clearpage 
\newpage
\section{Additional Results for Real Data Analysis}\label{appx:real-simu}
Table \ref{tab:real-data-all} shows additional results for the real data example presented in Section~\ref{sec:real.data} with other choices of the tuning parameters.   
Results for  PC, GES and joint GES methods combined with stability selection \citep{meinshausen2010stability}, which we implement using \texttt{stabsel} function in the  \texttt{stabs} package, are shown in Table \ref{tab:real-data-ss}. 
There is a hyperparameter \texttt{cutoff} in \texttt{stabsel} function, which we denote by ``cutoff1'' in the table. The \texttt{stabsel} function returns a selection probability for each edge, and as a result, we need to choose a threshold, denoted by ``cutoff2'', to obtain a DAG from the stable selection result. 
In Table \ref{tab:real-data-ss}, we list the results for $\text{cutoff1}=0.6,0.7,0.8,0.9$ and $\text{cutoff2}=0.55, 0.75$. 

\begin{table}[h!]
\centering
\begin{tabular}{|c|c|ccccc|}
\hline
Method  &  Parameters & $|\cG_{1}|$ & $|\cG_{2}|$ & $|\cG_{1} \cap \cG_{2}|$ & $N_{\text{total}}$ & ratio\\
\hline
PC & $\alpha =  1e-04 $& 12 & 26 & 7 & 31 & 0.2258 \\
PC & $\alpha =  5e-04 $& 19 & 38 & 13 & 44 & 0.2955 \\
PC & $\alpha =  0.001 $& 23 & 39 & 14 & 48 & 0.2917 \\
PC & $\alpha =  0.005 $& 33 & 60 & 18 & 75 & 0.24 \\
PC & $\alpha =  0.01 $& 42 & 69 & 19 & 92 & 0.2065 \\
PC & $\alpha =  0.05 $& 73 & 109 & 24 & 158 & 0.1519 \\
\hline
GES & $\lambda =  1 $& 150 & 238 & 49 & 339 & 0.1445 \\
GES & $\lambda =  2 $& 99 & 148 & 43 & 204 & 0.2108 \\
GES & $\lambda =  3 $& 78 & 108 & 34 & 152 & 0.2237 \\
GES & $\lambda =  4 $& 75 & 92 & 32 & 135 & 0.237 \\
GES & $\lambda =  5 $& 75 & 87 & 31 & 131 & 0.2366 \\
\hline
joint GES &$\lambda =  1 $& 77 & 85 & 68 & 94 & 0.7234 \\
joint GES &$\lambda =  2 $& 78 & 78 & 72 & 84 & 0.8571 \\
joint GES &$\lambda =  3 $& 76 & 76 & 72 & 80 & 0.9 \\
joint GES &$\lambda =  4 $& 76 & 76 & 73 & 79 & 0.9241 \\
joint GES &$\lambda =  5 $& 76 & 75 & 73 & 78 & 0.9359 \\
\hline
muSuSiE-DAG & $p^{-\omega_1} = p^{-1.25}, \, p^{-\omega_2} = p^{-2}$ & 33 & 115 & 30 & 118 & 0.2542 \\
muSuSiE-DAG & $p^{-\omega_1} = p^{-1.5}, \, p^{-\omega_2} = p^{-2.5}$& 27 & 95 & 25 & 97 & 0.2577 \\
muSuSiE-DAG & $p^{-\omega_1} = p^{-1.5} / 2, \, p^{-\omega_2} = p^{-2}$& 43 & 93 & 42 & 94 & 0.4468 \\
muSuSiE-DAG & $p^{-\omega_1} = p^{-2}, \, p^{-\omega_2} = p^{-3.5}$ & 17 & 68 & 14 & 71 & 0.1972 \\
muSuSiE-DAG & $p^{-\omega_1} = p^{-2} / 2, \, p^{-\omega_2} = p^{-3.5}$& 20 & 67 & 19 & 68 & 0.2794 \\
muSuSiE-DAG & $p^{-\omega_1} = p^{-\omega_2} = p^{-2}$  & 57 & 83 & 57 & 83 & 0.6867 \\
\hline
\end{tabular}
\caption{More results for the real data analysis. 
$|\cG_k|$: number of edges in the estimated DAG for the $k$-th group; $|\cG_{1} \cap \cG_{2}|$: number of edges shared by both DAGs; $N_{\text{total}}$: total number of edges in two DAGs; ratio: the ratio of $|\cG_{1} \cap \cG_{2}|$ to $N_{\text{total}}$. }
\label{tab:real-data-all}
\end{table}

\begin{table}
\centering
\begin{tabular}{|c|cc|ccccc|}
\hline
Method  &  cutoff1 & cutoff2 & $|\cG_{1}|$ & $|\cG_{2}|$ & $|\cG_{1} \cap \cG_{2}|$ & $N_{\text{total}}$ & ratio\\
\hline
PC & 0.6 & 0.55 & 49 & 85 & 19 & 115 & 0.1652 \\
PC & 0.6 & 0.75 & 36 & 63 & 18 & 81 & 0.2222 \\
PC & 0.7 & 0.55 & 48 & 85 & 19 & 114 & 0.1667 \\
PC & 0.7 & 0.75 & 37 & 63 & 19 & 81 & 0.2346 \\
PC & 0.8 & 0.55 & 51 & 87 & 20 & 118 & 0.1695 \\
PC & 0.8 & 0.75 & 35 & 65 & 18 & 82 & 0.2195 \\
PC & 0.9 & 0.55 & 50 & 87 & 20 & 117 & 0.1709 \\
PC & 0.9 & 0.75 & 36 & 62 & 18 & 80 & 0.225 \\
\hline
GES & 0.6 & 0.55 & 99 & 150 & 41 & 208 & 0.1971 \\
GES & 0.6 & 0.75 & 65 & 97 & 32 & 130 & 0.2462 \\
GES & 0.7 & 0.55 & 96 & 152 & 41 & 207 & 0.1981 \\
GES & 0.7 & 0.75 & 68 & 100 & 34 & 134 & 0.2537 \\
GES & 0.8 & 0.55 & 99 & 149 & 39 & 209 & 0.1866 \\
GES & 0.8 & 0.75 & 69 & 94 & 32 & 131 & 0.2443 \\
GES & 0.9 & 0.55 & 98 & 155 & 45 & 208 & 0.2163 \\
GES & 0.9 & 0.75 & 68 & 97 & 33 & 132 & 0.25 \\
\hline
joint GES & 0.6 & 0.55 & 60 & 61 & 57 & 64 & 0.8906 \\
joint GES & 0.6 & 0.75 & 57 & 58 & 55 & 60 & 0.9167 \\
joint GES & 0.7 & 0.55 & 67 & 70 & 63 & 74 & 0.8514 \\
joint GES & 0.7 & 0.75 & 58 & 58 & 57 & 59 & 0.9661 \\
joint GES & 0.8 & 0.55 & 65 & 71 & 56 & 80 & 0.7 \\
joint GES & 0.8 & 0.75 & 53 & 56 & 51 & 58 & 0.8793 \\
joint GES & 0.9 & 0.55 & 65 & 72 & 60 & 77 & 0.7792 \\
joint GES & 0.9 & 0.75 & 53 & 56 & 53 & 56 & 0.9464 \\
\hline
\end{tabular}
\caption{More results for PC, GES and joint GES methods in the real data analysis. $|\cG_k|$: number of edges in the estimated DAG for the $k$-th group; $|\cG_{1} \cap \cG_{2}|$: number of edges shared by both DAGs; $N_{\text{total}}$: total number of edges in two DAGs; ratio: the ratio of $|\cG_{1} \cap \cG_{2}|$ to $N_{\text{total}}$. }
\label{tab:real-data-ss}
\end{table}
 
\end{document}